\documentclass[11pt]{article}
\usepackage{xargs}                
\usepackage[pdftex,dvipsnames]{xcolor}  
\usepackage{amsthm}
\usepackage{graphicx} 
\usepackage{array} 
\usepackage{svg}
\usepackage{amsmath, amssymb, amsfonts, verbatim}
\usepackage{thmtools,thm-restate}
\usepackage{hyphenat,epsfig,subcaption,multirow}
\usepackage{nicefrac}
\usepackage{paralist}
\usepackage{multirow}

\DeclareSymbolFont{mathx}{U}{mathx}{m}{n}
\DeclareMathSymbol{\bigtimes}{1}{mathx}{"91}
\usepackage{etoolbox}\apptocmd\appendix{\pretocmd\section{\clearpage}{}{}}{}{}
\usepackage{tcolorbox}
\tcbuselibrary{skins,breakable}
\tcbset{enhanced jigsaw}

\usepackage[normalem]{ulem}
\usepackage[compact]{titlesec}

\definecolor{DarkRed}{rgb}{0.5,0.1,0.1}
\definecolor{DarkBlue}{rgb}{0.1,0.1,0.5}

\usepackage{nameref}
\definecolor{ForestGreen}{rgb}{0.1333,0.5451,0.1333}
\definecolor{Red}{rgb}{0.9,0,0}
\usepackage[linktocpage=true,
pagebackref=true,colorlinks,
urlcolor=DarkRed,
linkcolor=DarkRed,citecolor=ForestGreen,
bookmarks,bookmarksopen,bookmarksnumbered]
{hyperref}
\usepackage[noabbrev,nameinlink]{cleveref}
\crefname{property}{property}{Property}
\creflabelformat{property}{(#1)#2#3}
\crefname{equation}{Equation}{Equation}
\creflabelformat{equation}{(#1)#2#3}

\usepackage{bm}
\usepackage{url}
\usepackage{xspace}
\usepackage[mathscr]{euscript}

\usepackage{tikz}
\usetikzlibrary{arrows}
\usetikzlibrary{arrows.meta}
\usetikzlibrary{shapes}
\usetikzlibrary{backgrounds}
\usetikzlibrary{positioning}
\usetikzlibrary{decorations.markings}
\usetikzlibrary{patterns}
\usetikzlibrary{calc}
\usetikzlibrary{fit}

\usepackage{mdframed}

\usepackage[noend]{algpseudocode}
\makeatletter
\def\BState{\State\hskip-\ALG@thistlm}
\makeatother

\usepackage{cite}
\usepackage[shortlabels]{enumitem}
\usepackage[margin=1in]{geometry}
\usepackage{bbm}
\usepackage[font=small,labelfont=bf]{caption}

\DeclareMathOperator*{\argmax}{\arg\!\max}
\DeclareMathAlphabet\mathbfcal{OMS}{cmsy}{b}{n}
\usepackage{amsthm}
\makeatletter
\def\th@plain{%
  \thm@notefont{}
  \itshape 
}
\def\th@definition{%
  \thm@notefont{}
  \normalfont 
}
\makeatother

\newtheorem{theorem}{Theorem}
\newtheorem{lemma}{Lemma}[section]

\newtheorem{proposition}[lemma]{Proposition}
\newtheorem{corollary}[lemma]{Corollary}
\newtheorem{claim}[lemma]{Claim}
\newtheorem{fact}[lemma]{Fact}

\newtheorem{definition}[lemma]{Definition}

\newtheorem{question}{Question}

\newtheorem*{claim*}{Claim}
\newtheorem*{proposition*}{Proposition}
\newtheorem*{lemma*}{Lemma}
\newtheorem*{problem*}{Problem}
\crefname{lemma}{Lemma}{Lemmas}
\crefname{claim}{Claim}{Claims}
\newtheorem{remark}[lemma]{Remark}

\newtheoremstyle{restate}{}{}{\itshape}{}{\bfseries}{~(restated).}{.5em}{\thmnote{#3}}
\theoremstyle{restate}

\theoremstyle{definition}
\newtheorem{mdalg}{Algorithm}
\newenvironment{Algorithm}{\begin{tbox}\begin{mdalg}}{\end{mdalg}\end{tbox}}
\newtheorem{mdimp}{Implementation}

\allowdisplaybreaks
\renewcommand{\qed}{\nobreak \ifvmode \relax \else
	\ifdim\lastskip<1.5em \hskip-\lastskip
	\hskip1.5em plus0em minus0.5em \fi \nobreak
	\vrule height0.75em width0.5em depth0.25em\fi}

\newcommand{\Ot}{\ensuremath{\widetilde{O}}}
\newcommand{\eps}{\ensuremath{\varepsilon}}

\newcommand{\R}{\ensuremath{\mathbb{R}}}

\newcommand{\norm}[1]{\ensuremath{\|#1\|}}

\newcommand{\ceil}[1]{{\left\lceil{#1}\right\rceil}}
\newcommand{\floor}[1]{{\left\lfloor{#1}\right\rfloor}}

\DeclareMathOperator*{\var}{\mathrm{Var}}

\newcommand{\poly}{\mbox{\rm poly}}

\DeclareMathOperator*{\Prob}{\ensuremath{\textnormal{Pr}}}
\renewcommand{\Pr}{\Prob}

\newenvironment{tbox}{\begin{tcolorbox}[
		enlarge top by=5pt,
		enlarge bottom by=5pt,
		breakable,
		boxsep=0pt,
		left=4pt,
		right=4pt,
		top=10pt,
		arc=0pt,
		boxrule=1pt,toprule=1pt,
		colback=white
		]
	}
	{\end{tcolorbox}}

\newcommand{\II}{\ensuremath{\mathbb{I}}}

\newcommand{\mireal}[1][]{
	\ifx\relax#1\relax%
	\II(\mione \,; \mitwo)%
	\else%
	\II(\mione \,; \mitwo\mid #1)%
	\fi
}

\newcommand{\E}{\mathop{\mathbb{E}}}

\newcommand{\dist}{\textnormal{dist}}

\newcommand{\vecy}{\ensuremath{\mathbf{y}}}

\newcommand{\cost}{\textnormal{cost}}

\newcommand{\opt}{\ensuremath{\mathsf{OPT}}}

\newcommand{\costom}{\cost_\Omega} 
\newcommand{\ccint}{N} 
\newcommand{\uvec}[1]{u^{#1}} 
\newcommand{\uq}[2]{u^{#1}_{#2}} 

\newcommand{\init}{\ensuremath{\textnormal{\texttt{Init}}}}
\newcommand{\fin}{\ensuremath{\textnormal{\texttt{Fin}}}}
\newcommand{\notfar}[1]{{P_C}(#1)}
\newcommand{\allcenters}{\mathcal{S}}
\newcommand{\csize}{m}
\newcommand{\wi}[1]{w_{q_{#1}}}

\newcommand{\err}{\textnormal{err}}

\newcommand{\inSR}{\in \mathcal{S}(r)}
\newcommand{\hmax}{h_\ensuremath{\textnormal{\texttt{max}}}}
\newcommand{\rmax}{r_\ensuremath{\textnormal{\texttt{max}}}}
\newcommand{\bmax}{b_\ensuremath{\textnormal{\texttt{max}}}}
\newcommand{\tmax}{t_\ensuremath{\textnormal{\texttt{max}}}}

\newcommand{\lmax}{l_\ensuremath{\textnormal{\texttt{max}}}}

\newcommand{\lowc}[1]{J_{\text{L}}(#1)}
\newcommand{\highc}[1]{J_{\text{H}}(#1)}
\newcommand{\calE}{\mathcal E}
\newcommand{\sign}{\text{Sign}}
\newcommandx{\unsure}[2][1=]{\todo[linecolor=red,backgroundcolor=red!25,bordercolor=red,#1]{#2}}
\newcommandx{\change}[2][1=]{\todo[linecolor=blue,backgroundcolor=blue!25,bordercolor=blue,#1]{#2}}
\newcommandx{\info}[2][1=]{\todo[linecolor=OliveGreen,backgroundcolor=OliveGreen!25,bordercolor=OliveGreen,#1]{#2}}
\newcommandx{\improvement}[2][1=]{\todo[linecolor=Plum,backgroundcolor=Plum!25,bordercolor=Plum,#1]{#2}}
\newcommandx{\thiswillnotshow}[2][1=]{\todo[disable,#1]{#2}}

\title{
Sensitivity Sampling for $k$-Means:\\
Worst Case and Stability Optimal Coreset Bounds
}
\author{ 
Nikhil Bansal\thanks{University of Michigan, \url{bansal@gmail.com }} \and 
Vincent Cohen-Addad\thanks{Google Research, France,  \url{cohenaddad@google.com}} \and 
Milind Prabhu \thanks{University of Michigan, \url{milindpr@umich.edu}}  \and
David Saulpic \thanks{IST Austria, \url{david.saulpic@lip6.fr}} \and
Chris Schwiegelshohn \thanks{Aarhus University, \url{cschwiegelshohn@gmail.com}}
}
\date{}

\begin{document}

 \maketitle

    \smallskip
	
	\setcounter{tocdepth}{3}
	\begin{abstract}
    Coresets are arguably the most popular compression paradigm for center-based clustering objectives such as $k$-means. Given a point set $P$, a coreset $\Omega$ is a small, weighted summary that preserves the cost of all candidate solutions $S$ up to a $(1\pm \varepsilon)$ factor. For $k$-means in $d$-dimensional Euclidean space the cost for solution $S$ is $\sum_{p\in P}\min_{s\in S}\|p-s\|^2$. 
    
A very popular method for coreset construction, both in theory and practice, is Sensitivity Sampling, where
points are sampled in proportion to their importance. 
We show that Sensitivity Sampling yields optimal coresets of size $\Ot(k/\varepsilon^2\min(\sqrt{k},\varepsilon^{-2}))$ for worst-case instances. 
Uniquely among all known coreset algorithms, for well-clusterable data sets with $\Omega(1)$ cost stability, 
Sensitivity Sampling gives coresets 
of size $\Ot(k/\varepsilon^2)$, improving over the worst-case lower bound.
Notably, Sensitivity Sampling does not have to know the cost stability in order to exploit it: It is appropriately sensitive to the clusterability of the data set while being oblivious to it. 

We also show that any coreset for stable instances consisting of only input points must have size $\Omega(k/\varepsilon^2)$. Our results
for Sensitivity Sampling also extend to the $k$-median problem, and more general metric spaces. 
\end{abstract}

        \thispagestyle{empty}
        \clearpage
        \setcounter{page}{1}
        \allowdisplaybreaks
        \section{Introduction}
The ability to process and analyze large quantities of data is the driving force for both theoretical and applied machine learning research.
When analyzing large data sets, one typically makes them more manageable via distributed computation or outright compression using data reduction/summarization.
A natural approach here is using coresets or sparsifiers where the input data is replaced by a small set of weighted input points that provably approximate the loss function for a given set of queries.
Most coreset constructions are also highly parallelizable, as a union of coresets for distributed data is also a coreset of the union.

A bit more formally, given a large set of points $P$ and a set of queries $\mathcal{S}$ with associated loss function $f:P\times \mathcal{S}\rightarrow \mathbb{R}_{\geq 0}$, a coreset $\Omega$ for $P$ satisfies $\sum_{p\in \Omega} w_p f(p,S)\approx \sum_{p\in P}f(p,S)$ for all $S\in \mathcal{S}$, where the $w_p$ are reweighting factors. 
The problem of designing small coresets and understanding the tradeoff between size and error has been extensively studied in various contexts ranging from graph sparsification \cite{SpielmanS11, BatsonSS12}, linear regression \cite{DasguptaDHKM09}, subspace approximation \cite{WoodruffY23}, computational geometry \cite{AHV04}, clustering \cite{HaM04}, classification \cite{MaiMR21}, learning mixture models \cite{LucicFKF17} etc., and has been highly influential in fast algorithm design and  machine learning.

A very popular method for coreset construction, both in theory and practice, is sampling points in proportion to their importance. E.g., by effective resistances in graph sparsification \cite{SpielmanS11}, leverage scores in matrix approximation \cite{DrineasMMW12}, and so on.
More generally, Langberg and Schulman \cite{LS10} introduced the 
sensitivity sampling framework for coreset construction which was subsequently codified by Feldman and Langberg \cite{FL11}. This is highly effective for a wide range of problems, including several complicated clustering objectives.

In a nutshell, sensitivity sampling picks a sufficient number of points, where each point $p\in P$ is sampled with probability proportional to 
\[\sigma_p := \sup_{S} \frac{f(\{p\},S)}{f(P,S)},\]
and weighted inversely by its sampling probability. Intuitively, the sensitivity score captures the importance a point can have in any given solution and how ``indispensable" it is to include it in the sample.
In some cases, such as graph sparsification or leverage scores, the sensitivities may be computed up to arbitrary precision, though obtaining fast approximate algorithms still has seen significant research. For $k$-clustering, computing the exact scores is usually infeasbile, but there is extensive work on efficiently approximating them by various proxies \cite{BachemLL18,FL11,FeldmanSS20}.

\paragraph{Euclidean $k$-means.} Arguably, the most widely studied and important clustering objective in this line of research is the Euclidean $k$-means problem. Here we are given a data set $P$ consisting of points in $d$-dimensional Euclidean space and the task is to find $k$ centers $S$ such that \[
  \text{cost}(P,S):=\sum_{p\in P} \min_{c\in S}\|p-c\|^2  \]
is minimized, where $\|x\|=(\sum_{i=1}^d x_i^2)^{1/2}$ denotes the length of a vector. 
A weighted subset $\Omega$ is an $\varepsilon$-coreset of $P$ if for \emph{all} candidate solutions $S$, one has the property
\[|\text{cost}(P,S)-\text{cost}_\Omega(P,S)|\leq \varepsilon\cdot \text{cost}(P,S).\]
The typical parameter of interest for $\Omega$ is the number of distinct points $m$, the size of the coreset. For coresets constructed by sampling, the cost estimator is the weighted average
\begin{equation}
\label{eq:estimator}
 \text{cost}_\Omega(P,S)  = \frac{1}{m}\sum_{p\in \Omega} \min_{c\in S}\|p-c\|^2 \cdot \frac{1}{\mathbb{P}[p]}.
\end{equation}

\paragraph{Optimal Coreset bounds.} There has been a long line of work \cite{Chen09,FL11,LS10,FeldmanSS20,BecchettiBC0S19,huang2020coresets} culminating recently \cite{CSS21,Cohen-AddadLSSS22} in  $\tilde{O}(k\varepsilon^{-2}\cdot \min(\sqrt{k},\varepsilon^{-2}))$ size coreset bounds.\footnote{here and throughout the paper $\tilde{O}(\cdot), \tilde{\Omega}(\cdot)$ hides polylogarithmic factors in $k,\eps^{-1}$.} We discuss some of the ideas in \Cref{sec:overview}, see also 
 \Cref{table:core}, as they will be relevant for us. 

In a surprising recent result,  Huang, Li, and Wu \cite{HLW23}  showed that this unusual looking bound is actually the best possible. To do this, they construct an ingenious worst-case instance and show that any $\varepsilon$-coreset for it must have size  $\tilde{\Omega}(k\varepsilon^{-2}\cdot \min(\sqrt{k},\varepsilon^{-2}))$.

Despite this impressive recent progress, two natural questions remain.
\begin{question}
    Can the optimal $k$-means coreset be obtained by Sensitivity Sampling?
\end{question}

These optimal coreset bounds are not based on sensitivity sampling, but instead use the group sampling algorithm (GS) introduced by \cite{Chen09}, and refined further by \cite{CSS21} to yield its current analysis. A major drawback of group sampling is that it requires substantial preprocessing of the input into several groups, and moreover, the algorithm itself is tailored to these groups.

On the other hand, sensitivity sampling is extremely simple to implement (see \Cref{alg: sensitivity-sampling} below), 

and is significantly preferred to group sampling for several reasons: (i) it is extremely fast and more accurate in practice,\footnote{In the experiments of \cite{SchwiegelshohnS22}
it is consistently 10\%-400\% faster than group sampling on standard datasets.} (ii)
many downstream applications can use sensitivity sampling as a black box \cite{BravermanFLR19,BravermanHMSSZ21,Cohen-AddadWZ23,WoodruffZZ23}, and (iii)  
it can be applied to any problem whereas group sampling is limited to center-based clustering.

\begin{Algorithm}\label{alg: sensitivity-sampling}
	Sensitivity Sampling.
	
	\medskip

    \textbf{Input: } A set of $n$ points $P\subset \R^d$, integer $k$.
	\begin{enumerate}
        \item Compute a $O(1)$-approximate $k$-means solution $A = \{a_1, a_2, \ldots,a_k\}$ for $P$. Let $C_j \subset P$ be the cluster centered at $a_j$. For a point $p$ in $C_j$, let $\Delta_p:= \cost(C_j,A)/|C_j|$ denote the average cost of $C_j$. 
     
        \item Let $\mu: P \rightarrow \R^+$ be the following probability distribution. For a point $p \in C_j$, 
        \begin{align}
        \label{defn:mu}
            \mu(p) := \frac 14 \cdot\left(\frac{1}{k|C_j|} +\frac{\cost(p,A)}{k\, \cost(C_j,A)} +\frac{\cost(p,A)}{\cost(P,A)} + \frac{\Delta_p}{\cost(P,A)}\right).
        \end{align}
        \item For  $i$ from $1$ to $\csize $:
            \begin{itemize}
                \item Sample point $q_i$ independently from the distribution $\mu$. 
                \item Add $q_i$ to $\Omega$ with weight $\wi{i} :=  1/(\csize \cdot \mu(q_i))$.
            \end{itemize}
	\end{enumerate}
  \textbf{Output: } The set of points $\Omega = \{q_1, \ldots, q_m\}$ and the weights $\{w_{q_1}, \ldots, w_{q_m}\}$.
\end{Algorithm}

Perhaps the biggest advantage of sensitivity sampling is that the analysis can take a holistic, unified view of the data. 
This is especially relevant for clustering as the instances arising in practice usually have some global structure that makes them well-clusterable (to paraphrase \cite{DLS12}, if the data had no structure, we would not care about clustering it anyways). 
Identifying and analyzing such properties is a part of the beyond-worst-case-analysis framework, which has seen several successes in recent years \cite{Roughgarden19}.

So, it is natural to ask whether sensitivity sampling can exploit data-specific properties. 

\begin{question}
    Do compression tasks become provably easier if the data set is well-clusterable? 
\end{question}

Despite extensive work on designing good clustering algorithms under beyond-worst-case assumptions \cite{AngelidakisMM17,BalcanBG13,Cohen-AddadS17,KuK10,ORSS12}, we are not aware of any such work on coresets. 
Finally, we remark that group sampling based methods seem inherently incapable of addressing such questions as they take a very limited and local view of the data.

\subsection{Our Results}
We answer both of these questions. 
Our first main result is that sensitivity sampling gives optimal coresets for $k$-means in Euclidean spaces.
\begin{theorem}\label{thm: worst-case-coreset-thm}
Sensitivity Sampling yields a $k$-means coreset of size $\Tilde{O}(k/\varepsilon^2\cdot \min(\sqrt{k},\varepsilon^{-2}))$ in Euclidean space with constant probability.
\end{theorem}
As mentioned above, this bound is optimal, matching both the upper bound of group sampling \cite{CSS21,Cohen-AddadLSSS22} and the general lower bound of \cite{HLW23}.
Previously, the state-of-the-art $k$-means coreset bounds for Sensitivity Sampling were $\tilde{O}(k\cdot \varepsilon^{-4}\min(k,\varepsilon^{-2}))$ \cite{huang2020coresets}.

Note that in any analysis of sampling-based coresets, one needs to control the error of the unbiased estimator \eqref{eq:estimator}, which depends on (i) loss due to variance and (ii) loss due to the union bound over all possible clustering $S$ (up to some discretization).
The previous analyses for sensitivity sampling used coarse variance bounds of $k$ or $\eps^{-2}$, and $\exp(kd)$ on the number of clusterings.
However, these bounds are tight by themselves and cannot give any further improvements.

Our first technical contribution is a chaining based analysis of sensitivity sampling that carefully trades off the variance with the size of the clusterings at each distance scale.
Previously, all applications of chaining required group sampling due to its various useful properties. Instead, we give both tight bounds on the variance and the union bound by leveraging some of group sampling's structural properties, without having to compute them.

\paragraph{Coresets for well-clusterable inputs.}
For beyond-worst-case analysis, we consider the cost-stable clusterability criterion of Ostrovsky, Rabani, Schulman, and Swamy \cite{ORSS12}, which is arguably the oldest and most well understood stability criterion for clustering. A data set is called $\beta$-stable if the ratio between the cost of an optimal $k$-clustering $\text{OPT}_k$ and the cost of an optimal $k-1$ clustering $\text{OPT}_{k-1}$ satisfies $ \text{OPT}_{k-1}/{\text{OPT}_k} \geq 1+\beta$. The most important and widely studied setting is when $\beta= \Omega(1)$; see the related work  in Section \ref{sec:related-work} for more details. 

Our main result is as follows.
\begin{theorem}\label{thm: main-coreset-thm}
For cost-stable data sets in Euclidean space with $\beta= \Omega(1)$, Sensitivity Sampling yields a $k$-means coreset of size $\Tilde{O}(k/\varepsilon^2)$ with constant probability.
\end{theorem}
We emphasize that we use the same algorithm in both cases and that stability is only used in the analysis. Our second (key) technical contribution is to show more refined bounds between variance and the number of clusterings at various scales by exploiting the stability property.
We also note that achieving similar results for group sampling is significantly more difficult, as the groups do \emph{not} satisfy the aforementioned stability criterion, even if the entire data set does.
In a nutshell, our results show that Sensitivity Sampling is appropriately sensitive to stability while being oblivious to it.

One may wonder whether the $\tilde{O}(k/\varepsilon^2)$ bound is optimal for stable instances?

We show that this is indeed the case for coresets, which consist only of input points.
This assumption holds for most known coreset algorithms.

\begin{theorem}
For cost-stable data sets with any $\beta$ in Euclidean space, any coreset using non-negatively weighted input points must have size $\Omega(k/\varepsilon^2)$.
\end{theorem}

In \Cref{sec: perturbation}, we also remark that for another popular notion of well-clusterability known as perturbation resilience, one cannot obtain coresets that beat the worst-case bound. 

As sensitivity sampling outputs only the input points, it can also be applied to non-Euclidean metrics. Indeed, as a simple, straightforward consequence of our work, we also show that Sensitivity Sampling achieves optimal bounds for $k$-means in doubling metrics and finite metrics\footnote{
Finite $n$-point metrics are a special case of doubling metrics with doubling constant $\log n$.}.
\begin{theorem}\label{thm:doubling}
In doubling metrics with bounded doubling dimension $D$, Sensitivity Sampling computes a $k$-means coreset of size $\tilde{O}(kD/\varepsilon^2)$.
\end{theorem}
This improves over the $\tilde{O}(k^3D \varepsilon^{-2})$ analysis for Sensitivity Sampling by \cite{HuangJLW18} and matches the lower bound by \cite{CGSS22} (which was already attained by Group Sampling). See \Cref{app: doubling} for the details.

Finally, we remark that these results also carry over to the $k$-median problem.  
Specifically, we achieve a $\Tilde{O}(k/\varepsilon^{2}\min(\sqrt[3]{k},\varepsilon^{-1}))$ worst case bound for $k$-median and a $\tilde{O}(k/\varepsilon^2)$ bound for stable instances (see \Cref{sec: kmedian-coreset}). Given these results, our work confirms the common belief that Sensitivity Sampling is {\em the} right coreset algorithm for clustering.

\subsection{Related Work and Technical Overview}
\label{sec:overview}
We now give a more detailed overview of our ideas and place our work in the proper context. This requires understanding some previous work, which we now describe at a high level.

Clearly, the analysis of any sampling-based coreset algorithm has two components. Bounding the variance of estimator \eqref{eq:estimator}, and the loss due to union bound over the candidate solutions $S$. 
Consider a fixed solution $S$, and let $X(S) =  \text{cost}_\Omega(P,S)$ denote the estimator in \eqref{eq:estimator}.
The relative error for $S$ with $m$ samples is about $\sigma_S/m^{1/2}$, where $\sigma_S^2 = \E[X(S)^2]/\E[X(S)]^2$ is the (normalized) variance, and thus setting $\sigma_S/m^{1/2} = \eps$, it follows that $m=  \eps^{-2}\sigma_S^2$ samples suffice.
To control the error for a set $N$ of solutions $S$, by standard concentration and union bounds $m = \eps^{-2} \sigma^2 \log |N|$ samples suffice, where $\sigma$ is some upper bound on the variance for different $S$.   

Thus, all previous analyses focus on bounding $\log |N|$ and $\sigma^2$.
These analyses fall into two generations: those that don't use chaining and those that do (starting from \cite{CGSS22}).
See \Cref{table:core}.

\begin{table*}[ht!]
\begin{center}
\begin{tabular}{r|c|c|c|c}
Reference & Variance & Log of Net Size & Size &Algorithm
\\
\hline
\hline
\multicolumn{4}{l}{{\bf Upper Bounds for Euclidean $k$-means}}\\
\hline
\hline
\multirow{1}{*}{\cite{Chen09} (Sicomp'09)}& \multirow{1}{*}{$O(1) $} & \multirow{1}{*}{$kd $} & \multirow{1}{*}{$k^2d \varepsilon^{-2}\log n$} & GS* 
\\\hline
\cite{LS10} (SODA'10) & $k$ & $kd^2$ & $k^2d^2 \varepsilon^{-2}$ & SS
\\\hline
\cite{FL11} (STOC'11)& $\eps^{-2}$ & $kd$ & $kd \varepsilon^{-4}$ & SS
\\\hline
\cite{FeldmanSS20} (Sicomp'20)& $k$ & $k^2\varepsilon^{-2}$& $k^3 \varepsilon^{-4}$  & SS 
\\\hline
\cite{BecchettiBC0S19} (STOC'19) & $\eps^{-2}$ & $k\varepsilon^{-4}$ & $k \varepsilon^{-8}$ & SS 
\\
\hline
\cite{huang2020coresets} (STOC'20)&  $\eps^{-2}$ & $k\varepsilon^{-2}$ & $k \varepsilon^{-6}$ & SS
\\
\hline
\cite{BravermanJKW21} (SODA'21)& $k$ & $k\varepsilon^{-2}$ & $k^2 \varepsilon^{-4}$ & SS
\\
\hline
\multirow{2}{*}{\cite{CSS21} (STOC'22)} & $k$ & $O(1)$, $\varepsilon>2^{-t/2}$  &  \multirow{2}{*}{$k \varepsilon^{-4}$} & \multirow{2}{*}{GS}\\
& $O(1)$ & $k \varepsilon^{-2}$, $\varepsilon < 2^{-t/2}$ & & 
\\
\hline
\multirow{2}{*}{\cite{CGSS22} (STOC'23)} & $2^{-2h}k$ & \multirow{2}{*}{$k 2^{2h}$}  & \multirow{1}{*}{$k^2 \varepsilon^{-2}$}  & \multirow{2}{*}{GS}\\
& $2^{-2h}2^{t}$ & & \multirow{1}{*}{$k \varepsilon^{-4}$}  &\\
\hline
\multirow{3}{*}{\cite{Cohen-AddadLSSS22} (NeurIPS'23)$^\dagger$}  & $k$ & $O(1)$\,\, {\small ($2^h \leq 2^{t/2}$)}  & \multirow{3}{*}{$k \varepsilon^{-2}\min(\sqrt{k},\varepsilon^{-2}) $} & \multirow{2}{*}{GS}  \\
& \multirow{1}{*}{$ 2^{-2h}\min(\frac{k}{k_t\cdot 2^t},1)$} & $k\cdot \min(k_t,2^t)\cdot 2^{2h}$ &   & 
\\
&  &{\small $(2^h > 2^{t/2})$}  &    & \\
\hline
\multirow{2}{*}{(here) Theorem 1} & \multirow{2}{*}{$ 2^{-2h}\min(\frac{k}{k_t\cdot 2^t},1)$} & $k\cdot k_t\cdot 2^{2h}$ &\multirow{2}{*}{$k\varepsilon^{-2}\min(\sqrt{k},\varepsilon^{-2})$} & \multirow{2}{*}{SS} \\
&  & $k\cdot 2^t\cdot 2^{2h}$ & &
\\
\hline
\multirow{2}{*}{(here) Theorem 2} &\multirow{2}{*}{$ 2^{-2h}/\beta\cdot k_t$} & $k\cdot k_t\cdot 2^{2h}$ &\multirow{3}{*}{$k\varepsilon^{-2}$}& \multirow{3}{*}{SS} \\
&   & {\small ($2^{t}>k$, $2^h>2^{t/2}$)} &   & \\
\multirow{1}{*} {$(\beta \in \Omega(1))$}
& $O(2^{-2h}k)$ & $O(1)$ \,\,{\small $(2^t>k, 2^h \leq 2^{t/2})$} & & \\
\hline\hline
\multicolumn{4}{l}{{\bf Lower Bounds for Euclidean $k$-means}}\\
\hline\hline
\cite{CGSS22} (STOC'23) & \multicolumn{4}{c}{$k\varepsilon^{-2}$}\\
\hline
\cite{HLW23} (STOC'24) & \multicolumn{4}{c}{$k\varepsilon^{-2}\cdot \min(\sqrt{k},\varepsilon^{-2})$}\\
\hline
(here) Theorem 3 & \multicolumn{4}{c}{$k\varepsilon^{-2}$ for any constant $\beta$ and assuming $\Omega\subseteq P$}\\
\end{tabular}
\end{center}
\caption{Variance/Net size tradeoffs and resulting for previous coreset constructions for Euclidean $k$-means. All polylogarithmic factors are suppressed. GS denotes group sampling and SS denotes sensitivity sampling. Net sizes in terms of $t$ and the corresponding variance can be used in a chaining analysis. $^*$: Chen's paper \cite{Chen09} was an early version of group sampling with $k\log n$ groups, rather than the proposed $\log^2 \varepsilon^{-1}$ groups by \cite{CSS21}. $^\dagger$: The optimal analysis for group sampling was later independently reproved by \cite{HLW23}.}
\label{table:core}
\end{table*}

\paragraph{Generation 1.} These analyses gave generic upper bounds on $N$ and $\sigma^2$. The set $N$ consisted of an $\eps$-net of all possible solutions $S$, obtained by selecting $k$ points as centers from a sufficiently fine discretization of the unit Euclidean sphere of size $(1/\eps)^{d}$, as initially proposed by both Chen \cite{Chen09} and Feldman and Langberg \cite{FL11}.
Black box applications of dimension reduction techniques subsequently replaced $d$ with $\varepsilon^{-2}$, ignoring logarithmic factors, see \cite{BecchettiBC0S19,huang2020coresets}.
As a result, we have $\log |N| \approx k\min(d,\varepsilon^{-2})$, which remains the state of the art to this day.

For variance, a simple bound, given initially by \cite{LS10}, shows that $\sigma^2 \approx k$. A more sophisticated analysis by \cite{FL11} and \cite{CSS21} showed $\sigma^2 \approx \varepsilon^{-2}$. These bounds are based on considering how much each cluster in some approximately optimum clustering can cost in a non-trivial candidate solution $S$.
However, the bound $\sigma^2 = \min(k,\eps^{-2})$ is tight,
and thus $m = \eps^{-2} \sigma^2 \log |N| \approx  k\eps^{-4} \min(k,\eps^{-2})$ is the limit of these methods.
So any further improvement needs to come from more sophisticated ways of doing union bounds.

\paragraph{Generation 2.} This is where chaining enters the picture. Chaining is a very powerful approach, but often technically challenging, to optimally bound the suprema of Gaussian processes (the suprema over $S$ of the estimator \eqref{eq:estimator} can be written in this form via a standard symmetrization technique). The idea in chaining is to write the estimator as a telescoping sum of other estimators at increasingly finer distance scales and bound the variances and nets at each such scale. Some excellent references are \cite{vershynin_2018,talagrand2014upper}.

 \cite{CGSS22} were the first to introduce chaining to coreset analyses. Roughly, if we parameterize the distance scales by $h$ (scale $h$ has precision $2^{-h}$), the variance of the estimator $\sigma_h$ at scale $h$ decays exponentially in $h$, while the nets $N_{2^{-h}}$ get larger as they get increasingly finer. 
 The overall error rate then is the normalized sum (or integral) over all distances $2^{-h}$ ranging from $0$ to $\varepsilon^{-2}$,
  \[\sum_{h \leq 2 \log 1/\eps} \sqrt{\frac{\sigma_h^2}{m} \cdot \log |N_{2^{-h}}|}.\]
  Generally speaking, the improvement in chaining comes from carefully trading off the variance with the logarithm of the net size at each scale. 
Roughly, the naive union bound in generation $1$ results uses the worst net size and the worst variance bound over all scales, which could be very wasteful.

 Combining the older variance bounds of $k$ (times $2^{-2h}$ at distance scale $h$) from \cite{FeldmanSS20,LS10}, \cite{CGSS22} used the chaining framework to obtain an error rate of
 \[\sum_{h\in \mathbb{N}}  \sqrt{\frac{k2^{-2h}}{m}  \cdot k2^{2h}} \approx \sqrt{\frac{k^2}{m}},\]
 leading to $m = \eps^{-2} k^2$ coreset size bound.

To obtain the optimal bounds, \cite{Cohen-AddadLSSS22} gave a significantly more sophisticated trade-off between variance and net size. We will describe this next as we will need it.
Starting with a constant factor approximation $A$ to the optimal $k$-means clustering, 
the algorithm classifies each cluster of $A$ 
according to its costs in $S \in \mathcal{S}$, and obtains the variance and net size tradeoffs in terms of these parameters (see \Cref{table:core}).
Here, a cluster of $A$ has type $t$ if it costs $\approx 2^{t}$ times more in the solution $S$.

Let $k_t$ denote the number of clusters of type $t$.

Specifically, 
they give two different net constructions, one which is better for small $k_t$ and other if $k_t$ is large. In particular, they show $\log |N_{2^{-h}}| = O(k\min(k_t,2^{t})2^{2h})$, and variance bounds of $\min( k/(k_t2^{t}) , 1)$ (times $2^{-2h}$ at distance scale $h$), which gives an error of
\[\sqrt{ \frac{\min( k/(k_t2^{t}) , 1) 2^{-2h}}{m} \cdot  k\min(k_t,2^{t}) 2^{2h}}\approx \sqrt{\frac{k}{m}\cdot \min\left(\frac{k}{k_t},k_t,2^{t}\right)}\leq \sqrt{\frac{k}{m}\cdot \min(\sqrt{k},2^{t})}. \]
Since, we can assume that $2^t \leq \eps^{-2}$ this leads to the optimal  coreset of size ${m = \tilde{O}(k \eps^{-2} \cdot \min(\sqrt{k}, \eps^{-2}))}$.

\paragraph{Our work on Sensitivity Sampling.}
There are two key difficulties in analyzing sensitivity sampling. First, to extend the error rate of \cite{Cohen-AddadLSSS22} for group sampling, without being able to rely on the carefully chosen structural properties afforded to the analysis by group sampling. Our main contribution is to show that the existence of these structural properties suffices to obtain the desired net/variance tradeoffs without requiring an algorithm to compute and enforce them. 

Second, and our main novel contribution, is a similar refined variance/net size tradeoff for stable instances. 
The key property that stable instances satisfy is that if two clusters 
{\em interact} in a solution, at least one of them must pay at least $\beta\cdot k$ times more than their cost in $A$\footnote{This bound makes several simplifying assumptions that do not hold in general for stable instances, but they capture the essence of the analysis in a faithful way.}. 
This property already ensures that for most ranges of $2^{t}$, the previous variance/net tradeoffs from \cite{CSS21,CGSS22,Cohen-AddadLSSS22} cannot occur unless $2^{t}$ is very large. Specifically, two clusters can only interact if $2^{t} \geq \beta\cdot k$ for at least one of them.
This property is crucially \emph{not} satisfied for the groups used by group sampling. In particular, those groups can be worst-case instances that would, by themselves, require a coreset of size $k^{1.5}/\varepsilon^2$, even if the entire data set is stable. 

Consequently, one would hope that this increase in cost by a factor of at least $k$ should decrease the variance of the estimator by that factor. Unfortunately, this only holds for sufficiently large choices of $h$, in particular for $h>(\log k)/2$. To get around this, we use a different bound on the variance of the coreset estimator for smaller values of $h$, originally proposed by \cite{CSS21} for group sampling and extended here to sensitivity sampling. This gives an error rate of
\[\sum_{h\leq (\log k)/2} 2^{-h}\sqrt{\frac{1}{m}  \cdot k} + \sum_{h>  (\log k)/2}2^{-h}\sqrt{\frac{1}{m\beta\cdot k_t}  \cdot k\cdot k_t\cdot 2^{2h}}  \approx \sqrt{\frac{k}{\beta\cdot m}},\]
giving the optimal coreset bound of $\tilde{O}(k \eps^{-2})$ for $\Omega(1)$ stable instances. 

To the best of our knowledge, this is the first sensitivity sampling analysis that uses different variance bounds for the estimator at different precision scales, even when considering non-clustering objectives.
Finally, as mentioned earlier, the improved error rate does not require a change in the sensitivity sampling algorithm or any advanced knowledge of $\beta$. 
Thus, sensitivity sampling is appropriately sensitive to the stability of a data set while being oblivious to it.

\subsection{Notation and Definitions}
We now describe the notation and definitions used in the paper. 

For two expressions $A,B$ we use the notation $A \lesssim B$ to denote $A  = O(B)$.  We use  $\mathcal{S}$ to denote the set of all possible ordered $k$-tuples of centers in $\R^d$. For two points $p,q \in \R^d$, we let $\cost(p,q):= \norm{p -q}_2^2$ denote the squared distance between the points. Given a set of $k$ centers $S\in \mathcal{S}$ we let $\cost(p,S):= \min_{c \in S} \cost(p, c)$ to be the cost of $p$ to the nearest center in $S$.
For a set of of points  $P' \subset \R^d$,  we denote $\cost(P',S) := \sum_{p \in P'} \cost(p,S)$ to be the $k$-means cost of $P'$ wrt $S$. Similarly for a set of points $\Omega = \{q_1, \ldots, q_m\}$ with weights $\{w_{q_1}, \ldots, w_{q_m}\}$, we let $\costom(P',S) = \sum_{q_i \in P' \cap \Omega} w_{q_{i}} \cost(q_i, S)$.

With the notation in place, we formally define an  $\eps$-coreset.

\begin{definition}[$\eps$-Coreset]
    Given a set $P \subset \R^d$ of $n$ points and $\eps \in (0,1)$, an $\eps$-coreset is a set $\Omega = \{q_1, \ldots, q_m\} \subset \R^d$ of points with weights $\{w_{q_1}, \ldots, w_{q_m}\}$, that for any set of $k$ centers $S\in \mathcal{S}$ approximately preserves the $k$-means objective of $P$ with respect to $S$, i.e., 
    \begin{align*}
        \costom(P,S) \in  (1\pm \eps) \cost(P,S).
    \end{align*}
The number of points $m$ in the coreset is called its \emph{size}.
\end{definition}

We now define the notion of well-clusterability that we will use.
 \begin{definition}[$\beta$-Stability \cite{ORSS12}]\label{def: stability}
    For $\beta > 0$, a set $P \subset \R^d$ is $\beta$-stable if its optimal $k$-means cost $\mathrm{OPT}_k$ and optimal $(k-1)$-means cost $\mathrm{OPT}_{k-1}$  satisfy $\mathrm{OPT}_k\cdot (1+\beta) \leq \mathrm{OPT}_{k-1}$.
\end{definition}

\subsection{Further Related Work}
\label{sec:related-work}
\paragraph{Coresets.}
Coresets for clustering have received substantial attention over the years. 
Besides Euclidean metrics, they have also been studied for many other metrics, such as doubling metrics \cite{HuangJLW18,Cohen-AddadSS21} and shortest path metrics in graphs \cite{BakerBHJK020,BravermanJKW21}.
There have also been several works on extending coresets to constrained clustering objectives including clustering with cardinality constraints \cite{Cohen-AddadL19,HuangJV19,SSS19, BFS21,BravermanCJKST022}.

Further work on corsets considers objects other than points as centers \cite{FeldmanMSW10,HuangSV21,BJKW21} or other objectives altogether 
\cite{BoutsidisDM13,MK18,MunteanuSSW18,KarninL19,PhillipsT20,
TukanMF20,HuangSV20,MaiMR21}.
For further reading, we refer the interested reader to surveys \cite{Feldman20,MunteanuS18}.

\paragraph{Stability.} There has been a lot of work on understanding why clustering algorithms work unusually well in practice. Inspired by the elbow method\footnote{The elbow method is a heuristic used to compute the number of clusters in a data set. This involves plotting the cost of clustering 
 as a function of the number of clusters and picking the number of clusters to be the elbow of this curve.  Aside from being a popular heuristic in practice, there exists some theoretical justification for this approach, see \cite{BhattacharyyaKK22}.}, \cite{ORSS12} introduced cost stability to model the well-clusterability of real-life data.  Over the years, several works have shown that popular algorithms such as local search \cite{Cohen-AddadS17}, $k$-means++ \cite{JaG12, AgarwalJP15}, and even Lloyd's algorithm \cite{KuK10, AwS12} perform better for $\beta$-cost stable instances with $\beta= \Omega(1)$. Additionally, 
\cite{ABS10} also gave a PTAS for such inputs, which is unlikely to exist for worst-case instances. Clustering of cost-stable inputs has also been studied in various constrained models such as streaming and privacy \cite{ShechnerSS20}. Notably, better streaming algorithms for stable instances were known \cite{BMORST11} before they were achievable using coresets \cite{BravermanFLR19,CWZ23}. 

Various other models for the beyond-worst-case analysis of clustering, most notably perturbation resilience\footnote{Bilu-Linial stability, also known as perturbation resilience, is the other important stability notion that has received the most analysis for various problems. It is comparatively unlikely to hold for real-world instances as it is not robust to any form of noise, but it allows for solving clustering instances optimally. Unfortunately, assuming this stability notion does not yield any improvement for coresets, which we demonstrate in \Cref{sec:lowerbound}.}\cite{BiL12}, have been proposed.  These assumptions allow the recovery of good clusterings for various objectives  \cite{ABS12,BalcanHW20,BalcanL16,MakarychevMV14}, including $k$-means \cite{AngelidakisMM17}.
Lastly, there exist numerous assumptions for recovering mixtures, all of which bear a strong similarity to stability criteria \cite{Das99,ArK01,VeW04,DaS07,BrV08,KSV08,AwS12}.

\tikzstyle{block} = [rectangle, draw, fill=green!20, 
    text width=25em, text centered, rounded corners, minimum height=7em, font=\huge]
\tikzset{
    arrow/.style={
        draw,
        -stealth,
        line width=1.5pt, 
        line cap=round, 
        line join=round 
    }
}
\tikzstyle{blockred} = [rectangle, draw, fill=red!20, 
    text width=25em, text centered, rounded corners, minimum height=7em, font=\huge]
\tikzset{
    arrow/.style={
        draw,
        -stealth,
        line width=2pt, 
        line cap=round, 
        line join=round 
    }
}
\tikzstyle{blockblue} = [rectangle, draw, fill=blue!20, 
    text width=25em, text centered, rounded corners, minimum height=7em, font=\huge]
\tikzset{
    arrow/.style={
        draw,
        -stealth,
        line width=1.5pt, 
        line cap=round, 
        line join=round 
    }
}

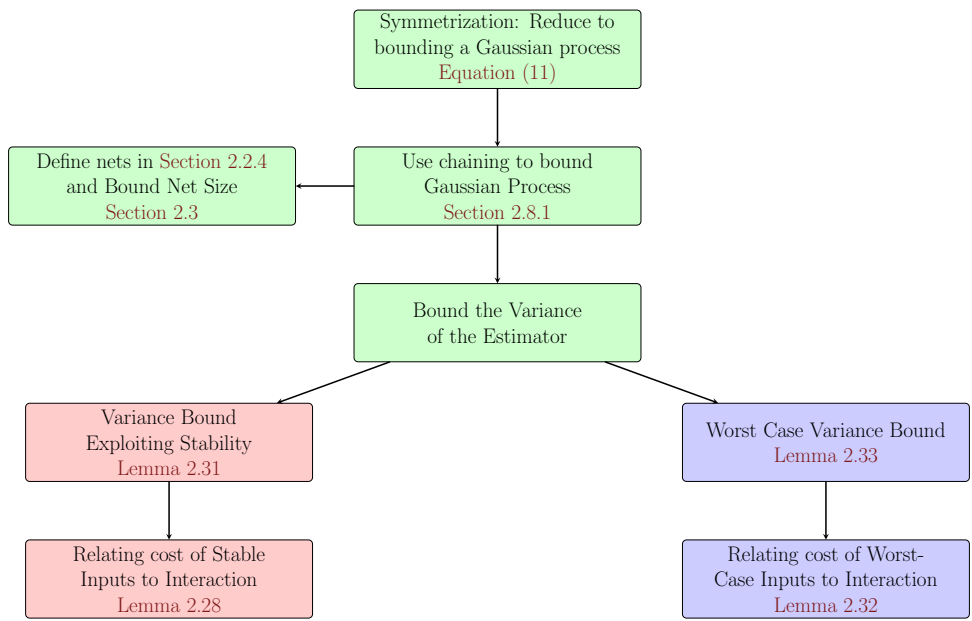
\begin{figure}[hbt!]

\centering
\resizebox{0.78\textwidth}{!}{
\begin{tikzpicture}[node distance=2cm]

\node[block](lemma6){Symmetrization: Reduce to bounding a Gaussian process \\
\Cref{eq:sym}};

\node[block, below=of lemma6](lemma8){Use chaining to bound Gaussian Process\\ \Cref{sec: completing-the-proof-kmeans}};

\node[block, left= of lemma8] (lemma9) {Define nets in \Cref{sec: chaining-nets} and Bound  Net Size\\\Cref{sec: cost-vector-net}};

\node[block, below = of lemma8](lemma10){Bound the \mbox{Variance} of the Estimator}; 

\node[blockred, below left = of lemma10](lemma11){Variance Bound \mbox{Exploiting} Stability\\\Cref{lem: stable-var}};

\node[blockred,below= of lemma11](lemma13){Relating cost of \mbox{Stable} Inputs to Interaction\\ \Cref{lem: cost-increase}};

\node[blockblue, below right = of lemma10](lemma12){Worst Case Variance Bound\\\Cref{lem: worst-case-var}};

\node[blockblue,below= of lemma12](lemma14){Relating cost of Worst-Case Inputs to Interaction\\ \Cref{lem: structured-cost}};

\draw [arrow] (lemma6) -- (lemma8);
\draw [arrow] (lemma8) -- (lemma9);
\draw [arrow] (lemma8) -- (lemma10);

\draw [arrow] (lemma10) -- (lemma11);
\draw [arrow] (lemma10) -- (lemma12);
\draw [arrow] (lemma11) -- (lemma13);
\draw [arrow] (lemma12) -- (lemma14);
\end{tikzpicture}
}
\caption{A roadmap of the key steps of the analysis. The steps in the green boxes are common to the proof for stable inputs and worst-case inputs. The key difference is in analyzing the variance of the estimator, and these steps appear in the red and blue boxes, respectively.}
\label{fig: roadmap}
\end{figure}

\section{The Analysis of Sensitivity Sampling}\label{sec: analysis}
\subsection{Roadmap of the Analysis}

We now prove  \Cref{thm: worst-case-coreset-thm} and \Cref{thm: main-coreset-thm}, which bound the coreset sizes for worst-case inputs and stable inputs, respectively.  To show that the output $\Omega$ of 
Sensitivity Sampling (\Cref{alg: sensitivity-sampling}) is an $\eps$-coreset, we need to show that the maximum relative error of $\Omega$ over all possible placements of centers $S \in \mathcal{S}$ is at most $\eps$. This is equivalent to showing the following: 
\begin{align}\label{eq: goal}
	\begin{split}
		\E_{ \Omega} \sup_{S \in \mathcal{S}} \left\lvert \frac{\cost(P,S) - \costom(P,S)}{\cost(P,S) } \right\rvert \leq \eps.
	\end{split}
\end{align}
\Cref{thm: worst-case-coreset-thm,thm: main-coreset-thm} then immediately follow from \eqref{eq: goal} via Markov's inequality.

Showing \eqref{eq: goal} involves several steps which we now sketch (also see \Cref{fig: roadmap} for an illustration).

\paragraph{Symmetrization.} The first key step is to apply a symmetrization argument (in \Cref{sec: symmetrizaion}) to reduce this task to bounding the supremum of a Gaussian process.  
Specifically, \eqref{eq: goal} reduces to showing the following:
\begin{align}\label{eq: roadmap-sym}
\E_{ \Omega} \E_{g} \sup_{S \in \mathcal{S}} \left\lvert \frac{\sum_{q_i \in \Omega}g_i \cdot w_{q_i} \cdot \cost(q_i,S) }{\cost(P,S)} \right\rvert  \leq \eps,
\end{align}
where the $g_i$ are $m$ independent standard Gaussian random variables. Symmetrization allows us to fix $\Omega$ and bound the inner expectation in \eqref{eq: roadmap-sym}, which is only over the randomness of the Gaussians. A key advantage of fixing $\Omega$ is that the nets for the chaining argument can now depend on $\Omega$.

\paragraph{Constructing Nets. } The supremum in \eqref{eq: roadmap-sym} is over the infinitely many centers in $\mathcal{S}$. To handle this, the next important step is to construct nets approximating $\mathcal{S}$  at various distance scales. This is done in \Cref{sec: cost-vector-net}. At a high level, a net of $\mathcal{S}$ has the following property:  For any $S \in \mathcal{S}$ the net contains an $S'$  such that for each point  $q \in \Omega$, we have $\cost(q,S) \approx \cost(q,S')$.  

\paragraph{Chaining. } Next, we use a chaining argument  (in \Cref{sec: chaining}) to bound the Gaussian process \eqref{eq: roadmap-sym}. To do this, we decompose it into a sum of multiple Gaussian processes at different distance scales,
using the nets of $\mathcal{S}$ at these scales. The main technical challenge then is to bound the variance of these various Gaussian processes and trade it off with the corresponding net size.

To bound the variance, we crucially exploit that points are sampled roughly proportional to their sensitivity.
To obtain the trade-off, we identify a key parameter called the ``interaction number" (in \Cref{sec: itneraction-of-centers}) that quantifies the complexity of interactions between solutions and the points. The key idea is to show that while centers with high interaction numbers require large nets, they have a large cost and their (normalized) variance becomes low.

Bounding the variance requires different arguments for worst-case inputs and stable inputs.
 
The variance bound for worst-case inputs is described in \Cref{sec:worst-case-variance} and the trade-off with the net sizes is described in \Cref{sec: completing-the-proof-kmeans}.
For stable inputs, we show stronger guarantees on the variance in terms of the stability parameter $\beta$. This is done in \Cref{sec: variance-stable}. The trade-off with the net sizes to obtain the final coreset bound for stable inputs is described in \Cref{sec: trading-net-variance-stable}.

The above roadmap simplifies some of the details of the analysis. Before performing the key steps described above, in the next section, we decompose the problem of proving \eqref{eq: goal} into various structured cases that help us perform the analysis cleanly.

\subsection{Classifying Clusters and Centers}\label{sec: classifying}
Consider the $O(1)$-approximate clustering $ C_1, \ldots, C_k$, with centers  $A = \{ a_1, \ldots, a_k \}$, computed by \Cref{alg: sensitivity-sampling}. To bound the coreset error, we partition these clusters into groups of similar clusters and separately bound the error of the coreset on each group.  We also partition the centers in $\mathcal{S}$  based on how they ``interact'' with these clusters.   Such a classification of clusters and centers is useful in defining the nets and controlling their sizes. It also helps us control the variance of the estimator more cleanly.

\subsubsection{Partitioning Clusters into  Far and Close Clusters}\label{sec: far-close-points}
For a given $S \in \mathcal{S}$, we first partition the clusters $C_i$ into \textit{far} and \textit{close} clusters depending on the distance of their centroids $a_i$ from $S$. The high-level idea is that one can show via standard concentration inequalities that the size of each cluster $C_i$ is (approximately) preserved by the coreset. This is useful because if a cluster $C_i$ is far from $S$, the coreset automatically preserves the cost of $C_i$ to $S$ as all points in this cluster have roughly the same cost to $S$. Thus, far clusters are easily dealt with. This reduces our task to analyzing close clusters. We formalize these ideas below.

Let $\Delta_j = \cost(C_j, A)/|C_j|$ denote the average cost of a point in cluster $C_j$. We say that cluster $C_j$ is \textit{far} from $S$ if $\cost(a_j, S) > \Delta_j \eps^{-2}$; otherwise we say $C_j$ is \textit{close} to $S$. A point $p \in P$ that lies in a far cluster is called a far point (with respect to $S$); otherwise, it is called a close point. Let $P_F(S)$ and $\notfar{S}$ denote the set of far and close points with respect to $S$. 

\paragraph{Bounding the costs separately.} To prove 
\Cref{thm: main-coreset-thm}, we separately bound the contributions to the error by far and close points. The following lemmas summarize these results.

\begin{lemma}[Cost Preservation of Close Clusters]\label{lem: goal-close}
	If the coreset size is ${\tilde{\Omega}(k \eps^{-2} \cdot \min (\sqrt{k}, \eps^{-2}))}$ for worst-case inputs or $ \tilde{\Omega}(k \eps^{-2})$ for $\beta$-stable inputs then we have,
	\begin{align*}
		\begin{split}
			\E_{ \Omega} \sup_{S \in \mathcal{S}} \left\lvert \frac{\cost(\notfar{S},S) - \costom(\notfar{S},S)}{\cost(P,S)} \right\rvert \leq \eps/2.
		\end{split}
	\end{align*}
\end{lemma}

\begin{lemma}[Cost Preservation of Far Clusters]\label{lem: goal-far}
	If the coreset size is  $\tilde{\Omega}( k \eps^{-2})$ then we have,
	\begin{align*}
		\begin{split}
			\E_{ \Omega} \sup_{S \in \mathcal{S}} \left\lvert \frac{\cost(P_F(S),S) - \costom(P_F(S),S)}{\cost(P,S)} \right\rvert \leq \eps/2.
		\end{split}
	\end{align*}
\end{lemma}
 The proof of \Cref{lem: goal-far} will be relatively simple and is presented in \Cref{sec:far-points}. Most of the work will be in proving \Cref{lem: goal-close}, which requires careful use of symmetrization and chaining. This will be done in the following sections.

        \subsubsection{Partitioning Close Clusters Based on Cost}
\label{subsec: partitioning-close-points}
We now perform a clean-up step by further classifying close clusters into high-cost and low-cost clusters. The key idea is that we only need to bound the error for high-cost clusters, as ignoring the low-cost clusters only has a tiny effect on the total cost. We now formally define this classification.

 \begin{definition}[Low-Cost and High-Cost Close Clusters]
       Let $T := \eps^3 \cdot \cost(P,A)/k$.
 Let $\lowc{S}$ denote the set of (low-cost) clusters $C_j$ that are close to $S$ and satisfy $\cost(C_j, A) < T$. Similarly, let $\highc{S}$ denote the set of (high-cost) clusters that are close to $S$ and satisfy $\cost(C_j, A) \geq T$. 
 \end{definition}
 Abusing the notation slightly, we also use $\lowc{S}$ (resp. $\highc{S}$) to denote points in these clusters.

The following lemmas show that the coreset preserves the cost of low-cost and high-cost clusters, respectively. Since the low-cost and high-cost clusters partition the set of close clusters,  these lemmas imply \Cref{lem: goal-close} after rescaling $\eps$.

\begin{lemma}[Handling High Cost Clusters]\label{lem: high-cluster}
If the coreset size is ${\tilde{\Omega}(k \eps^{-2} \cdot \min (\sqrt{k}, \eps^{-2}))}$ for worst-case inputs or $ \tilde{\Omega}(k \eps^{-2})$ for $\beta$-stable inputs then we have,
\begin{align}\label{eq: high-cost-cluster}
    \E_{ \Omega} \sup_{S \in \mathcal{S}} \left\lvert \frac{\cost(\highc{S}, S) - \costom(\highc{S},S)}{\cost(P,S) } \right\rvert \leq \eps.
\end{align}
\end{lemma}

\begin{lemma}[Handling Low Cost Clusters]\label{lem: tiny-cluster}
If the size of the coreset is  $\tilde{\Omega}( k \eps^{-2})$ then
\begin{align} 
    \E_{ \Omega} \sup_{S \in \mathcal{S}} \left\lvert \frac{\cost(\lowc{S}, S) - \costom(\lowc{S},S)}{\cost(P,S) } \right\rvert \leq \eps.
\end{align}
\end{lemma}

The proof of \Cref{lem: tiny-cluster} is straightforward and is provided in \Cref{app: tiny-cluster}. Proving \Cref{lem: high-cluster} is much more interesting, and we now focus on this.

\subsubsection{Further Classifying Clusters Based on Cost: Bands and Types}\label{sec:further-classifying}

Let $A$ be the approximate optimal solution computed by the algorithm, and let $S$ be a fixed solution in $\mathcal{S}$. We now partition the high-cost clusters such that two clusters $C_i, C_j$ in the same group satisfy  $\cost(C_i, A) \approx \cost(C_j, A)$ and $\cost(C_i, S) \approx \cost(C_j, S)$.   Such a grouping allows for an easier analysis of the variance of the estimator and net sizes and is crucial to trade-off the two quantities. 

We first group clusters with a similar cost with respect to the (approximately optimal) solution $A$ and call each such a group a \textit{band} as defined formally below.

\begin{definition}[Bands]
    For each integer $b$ satisfying $0 \leq b \leq \bmax := \floor{\log_2(k \eps^{-3})}$, let Band-$b$ be the set of clusters $C_j$ with  $\cost(C_j, A) \in [2^{b} T, 2^{b+1} T)$ where $T = \eps^3k^{-1} \cdot \cost(P,A)$. 
\end{definition}
Since each high-cost cluster $C_j$ has a cost in the range $[T,\,  k \eps^{-3}  T]$, the $(\bmax+1)$ bands defined above form a partition of high-cost clusters.

While clusters in the same band have similar costs in $A$, their cost in an arbitrary set of centers $S$ can be very different. 
Next, we group clusters into \textit{Types} based on their cost in $S$. 
Note that, unlike bands, this grouping is a function of the set $S$. 
\begin{definition}[Types]
Let $S$ be a set of centers. For an integer $t$ satisfying $1 \leq t \leq \tmax := \ceil{\log_2(\eps^{-2})}$, a cluster $C_j$ is of Type-$t$ for $S$ if (i) $C_j$ is close to $S$ and (ii) $\cost(a_j, S) \in [2^{t-1} \Delta_j, 2^{t} \Delta_j)$ where $\Delta_j = \cost(C_j, A) / |C_j|$ is the average cost of cluster $C_j$.  
On the other hand, if $C_j$ is close to $S$ and  $\cost(a_j, S) \in [0, \Delta_j)$, it is of Type-$0$.
\end{definition}
Since any close cluster $C_j$ satisfies $\cost(a_j, S) \leq \Delta_j \eps^{-2}$, the $(\tmax+1)$ types defined above also partition the close clusters. 

 We now group clusters with both the same band and the same type. For a set of centers $S$, let $B_{b,t}(S)$ denote the clusters from Band-$b$ with Type-$t$ for $S$. All clusters in $B_{b,t}(S)$ contain clusters with a similar cost in the approximately optimal solution and wrt to $S$. We shall abuse notation slightly and use $B_{b,t}(S)$ to also refer to the \textit{points} in these clusters. As there are $O(\log(k\eps^{-1}))$ bands and $O(\log ( \eps^{-1}))$ types, there are $O(\log^2(k \eps^{-1}))$ sets $B_{b,t}(S)$.

As the sets $B_{b,t}(S)$ partition points in $\highc{S}$, \Cref{lem: high-cluster} follows if we show the following: For each pair $(b,t) \in [\bmax] \times [\tmax]$ we have,
\begin{align}\label{eq: bvw}
     \E_{\Omega} \sup_{S \in \mathcal{S}}\left\rvert\frac{\cost(B_{b,t}(S), S) - \costom(B_{b,t}(S), S)}{\cost(P,S)} \right\rvert \lesssim \frac{\eps}{\log^2(k \eps^{-1})}.
\end{align}
 Henceforth, we fix $b \in [\bmax]$, $t \in [\tmax]$ and show that \eqref{eq: bvw} holds for this choice of $b, t$. When clear from context, we also suppress the subscripts and refer to the set $B_{b,t}(S)$ as $B(S)$. We use the notation $k_{B(S)}$ to denote the number of clusters in $B(S)$.

\subsubsection{Defining Cost Vectors and Cost Vector Nets}\label{sec: chaining-nets}
This section defines nets to approximate the infinitely many solutions in $\mathcal{S}$. Instead of directly discretizing the space of centers in $\mathcal{S}$, we discretize the set of \textit{cost vectors} induced by $S \in \mathcal{S} $. For a fixed coreset $\Omega$, the cost vector induced by $S \in \mathcal{S}$ is a vector in $\R^m$ whose entries are costs of points in $\Omega \cap B(S)$ with respect to $S$. 
\begin{definition}[Cost Vectors] Let $\Omega = \{q_1, \ldots, q_m\}$ be the set of $m$ points sampled by \Cref{alg: sensitivity-sampling} and $S$ be a set of centers in $\mathcal{S}$. The cost vector of $\Omega$ induced by $S$ is a vector $u^S(\Omega) \in \R^m$ defined as follows:
    \begin{align}\label{eq: us}
    \text{For each $i \in[m]$}  \text{ we have } u^S_{i}(\Omega) = \mathbbm{1}[q_i \in  B(S)] \cdot \cost(q_i,S).
    \end{align}
\end{definition}
 Note that $u^S(\Omega)$ is a random vector (where the randomness comes from the choice of $\Omega$), and its coordinates are mutually independent since \Cref{alg: sensitivity-sampling} obtains samples each $q_i$ independently.

For a fixed $\Omega$ and a subset $\mathcal{T}$ of $\mathcal{S}$,  let  $M(\Omega, \mathcal{T}) = \{u^S(\Omega) | \, S \in \mathcal{T}\}$ be the set of cost vectors induced by centers in $\mathcal{T}$. If $\mathcal{T}$  has infinitely many centers, the set $M(\Omega, \mathcal{T})$ may have infinitely many cost vectors. We now define \textit{cost vector nets}, which are discretizations of $M(\Omega, \mathcal{T})$.  A cost vector net $M_\alpha(\Omega, \mathcal{T})$ at scale $\alpha$ is a finite set of representative vectors such that any $u^S(\Omega)$ in $M(\Omega, \mathcal{T})$, is  ``$\alpha$-approximated '' by some vector in $M_\alpha(\Omega, \mathcal{T})$. We formalize this notion below.

\begin{definition}[Cost Vector Nets] \label{def: clustering-nets} Consider a fixed set $\Omega$ and let $\mathcal{T}$ be a subset of $\mathcal{S}$ (recall that $\mathcal{S}$ denotes the set of all possible $k$ centers in $\R^d$). For a real $\alpha \in (0,1/2]$, an $(\alpha, \mathcal{T})$ cost vector net, denoted by  $M_\alpha(\Omega, \mathcal{T})$, is a finite subset of $\R^m$ with the following properties. For any set of centers $S \in \mathcal{T}$ there exists some $v \in M_\alpha(\Omega, \mathcal{T})$ which $\alpha$-approximates the cost vector $u^S(\Omega)$ in the following sense: for each $i \in [m]$,
\begin{enumerate}
    \item If $q_i \in  B(S)$ then \[|v_i - u^S_i(\Omega)| = |v_i - \cost(q_i, S)|   \leq \alpha \cdot \err(q_i,S),\] where  for any $p \in P$ and $S \in \mathcal{S}$,
        \[\err(p,S) := \left( \sqrt{\cost(p,S)\cost(p,A)} + \sqrt{\cost(p,S) \Delta_{p} } + \cost(p,A) + \Delta_{p}\right).\]
    \item If $q_i \notin B(S)$ then $v_i = 0$ . Note that $u^S_i(\Omega) = 0$ by definition (see \eqref{eq: us}).
\end{enumerate}
\end{definition}

\subsubsection{Grouping Centers Based on Interaction}\label{sec: itneraction-of-centers}
Our next goal is to bound the size of the cost vector nets defined above. To do this, we group sets of centers similar to each other and construct nets for each group. Specifically, for each set of centers $S \in \mathcal{S}$, we define a parameter called the \emph{interaction number} of $S$ that captures how much it ``interacts'' (formally defined below) with the set of clusters $B(S)$. This parameter will be crucial in quantifying the size of the nets that we construct. We then group centers with similar interaction numbers and then construct cost vector nets for each group. 

Roughly, we say that a center $x \in S$ interacts with a cluster $C_j$ if it is significantly far away from its center $a_j$ while still being approximately the nearest center in $S$ to it. The precise definition is the following: 

\begin{definition}[Cluster-Center Interaction] \label{def: interaction} Let $S \in \mathcal{S}$ be a set of $k$ centers. For a center $x \in S$, we say that a cluster $C_j$ in $B(S)$ \textbf{\emph{interacts}} with $x$ if both the following conditions hold:

\begin{description}
	\item{\textbf{$P_1$}: \textbf{(Point $x$ is outside average cost ball of $C_j$).}} We have  $\cost(a_j,x) \geq 32 \Delta_j$.
	\item{\textbf{$P_2$}: \textbf{(Point $x$ is an approximate nearest center to $a_j$).}} We have $\cost(a_j,x) \leq 16 \cost(a_j,S)$.
\end{description}
\end{definition}

Fix a set of centers $S \in \mathcal{S}$ and a center $x \in S$. We use $I(x)$ to denote the set of clusters $C_j \in B(S)$ that interact with center $x \in S$. We aggregate the sizes of the sets $I(x)$ for the various centers in $S = \{x_1, \ldots x_k\}$ to obtain the \textit{signature} of $S$ defined as follows: \[\sign(S) := (|I(x_1)|, \ldots, |I(x_k)|).\]

\begin{definition}[Interaction Number]\label{def: interaction-number}
For each $S \in \mathcal{S}$, with $S = \{x_1, \ldots, x_k \}$,  we define its interaction number, denoted by $N(S)$, to be $N(S) = \sum_{i = 1}^k |I(x_i)|$.
\end{definition}
 As each center of $S$ can interact with at most $|B(S)| = k_{B(S)}$ clusters, 
for any set of centers $S$, we have  $N(S) \leq k \cdot k_{B(S)}\leq k^2$.

\noindent We now group centers with similar interaction numbers and call these groups center classes.
\begin{definition}[Center Classes]
For an integer $r$ satisfying $0 \leq r \leq \rmax:= \ceil{\log_2(k^2)}$, the center class $r$ is the collection of all sets of $k$ centers $S$ that satisfy $N(S) \in [2^r, 2^{r+1})$. We use  $\mathcal{S}(r)$ to denote the center class $r$.
\end{definition}

\begin{remark}
Notice that the interaction number $N(S)$ is defined with respect to a specific set $B_{b,t}(S)$ of clusters. As $b$ and $t$ vary, for any $S$, its interaction number $N(S)$ and thus its center class can change.
\end{remark}

Note that the center classes partition the collection $\mathcal{S}$ of all possible sets $S$. Therefore, we can bound the supremum over $\mathcal{S}$ in \eqref{eq: bvw} by a sum of suprema over each center class; it suffices to prove the following: for any $r \in [\rmax]$,
\begin{align}\label{eq: bvwr}
     \E_{\Omega}\, \sup_{S \in \mathcal{S} (r)}\left\rvert\frac{\cost(B(S), S) - \costom(B(S), S)}{\cost(P,S)} \right\rvert  \lesssim \frac{\eps}{ \log^3(k \eps^{-1})}. 
\end{align}
\Cref{eq: bvw} then follows since $\rmax = O(\log(k \eps^{-1}))$.
  Henceforth, we also fix a center class  $r$ and prove \eqref{eq: bvwr} for this $r$. As the supremum will always be over $\mathcal{S}(r)$ from this point on, we remove it from the subscript (i.e., we write $\sup_S$ instead of $\sup_{S \in \mathcal{S}(r)})$.

In the next section, we give a bound on cost vector nets for centers in center class $\mathcal{S}(r)$. This will be a crucial step towards proving \eqref{eq: bvwr}.

\subsection{Bounding the Cost Vector Net Size.}\label{sec: cost-vector-net}
Recall that we are currently bounding the error of the coreset for the clusters in $B(S)$, namely those clusters of band $b$ and type $t$. We now give a bound on the size of the cost vector net (\Cref{def: clustering-nets}) approximating cost vectors of $B(S) \cap \Omega$ induced by the centers in center class $\mathcal{S}(r)$.  
\begin{lemma}\label{lem: net-lemma}
 For any $\alpha \in (0,1/2)$, there is an $(\alpha, \mathcal{S}(r))$-net $M_\alpha$ with cardinality \[|M_\alpha| = \exp(O(\min(2^r+ k\alpha^{-2}, 2^t k \alpha^{-2})\cdot \log(k \alpha^{-1}{\eps^{-1}}))).\]
\end{lemma}
\noindent Since this lemma is based on ideas from previous work \cite{Cohen-AddadSS21}, we move its proof to \Cref{app: nets}.

\subsection{Applying a  Symmetrization Argument}\label{sec: symmetrizaion}
We now have bounds on the cost vector nets that will be useful in performing a union bound to control the expected supremum of the random process given by \eqref{eq: bvwr}. The next step is to bound the variance of the estimator. It turns out that if the coreset $\Omega$ satisfies some ``good'' properties, which we shall soon describe, we can show better bounds on the variance. Therefore, we first apply a standard symmetrization argument (see, for ex., Section 6.4 of \cite{vershynin2018high}), which allows us to fix the randomness of $\Omega$ and reduce proving \eqref{eq: bvwr} to a Gaussian process.

We first introduce some useful notation. For $\csize$  independent Gaussians $(g_i)_{i\in [m]}$, where $g_i\sim N(0,1)$, define the random variable
\begin{align}
\label{def:X}
    X^S(\Omega, g) = \sum_{i \in [m]} \frac{g_i w_{q_i} u_i^S(\Omega)}{\cost(P,S)}.
\end{align}
Applying the symmetrization technique in \Cref{lem: symmetrization}, by \Cref{lem: sym-apply} we have that,
\begin{align}\label{eq:sym}
	\E_{\Omega} \sup_{S \inSR}\left\rvert\frac{\cost(B(S), S) - \costom(B(S), S)}{\cost(P,S)} \right\rvert\leq \sqrt{2\pi} \E_{ \Omega} \E_{g} \sup_{S \inSR} \left\lvert X^S(\Omega,g)\right\rvert.
\end{align}
\noindent The remaining section is devoted to showing the following lemma:
\begin{lemma}\label{eq: inter-xs}
If the coreset size is ${\tilde{\Omega}(k \eps^{-2} \cdot \min (\sqrt{k}, \eps^{-2}))}$ for worst-case inputs or $ \tilde{\Omega}(k \eps^{-2})$ for $\beta$-stable inputs then we have,
\begin{align*}
    \E_{\Omega} \E_g \sup_{S \inSR} \left\lvert X^S(\Omega,g)\right\rvert \leq \eps.
\end{align*}
\end{lemma}
The above lemma implies \eqref{eq: bvwr} upon rescaling $\eps$ by $\Theta(\log^3(k\eps^{-1}))$ factors.
\subsubsection{\texorpdfstring{Good Properties of the Coreset: Event $\mathcal{E}$}{Good Event E}}
We now take a slight detour and prove that the output $\Omega$ of \Cref{alg: sensitivity-sampling} satisfies some nice properties with high probability. In particular, we show that $\Omega$ (approximately) preserves the number of points in each cluster $C_i$ as well as the cost of $C_i$ (where $C_1, \ldots, C_k$ is the clustering computed by \Cref{alg: sensitivity-sampling} in the first step). Moreover, $\Omega$ does not over-sample high-cost points from any clusters. These properties are summarized by an event $\mathcal{E}$ defined below.

Before we formalize these properties, it will be convenient to first partition points in a cluster into rings according to their cost from the center. We define the notation $\Delta_j = \cost(C_j,A)/|C_j|$ to be the average cost of cluster $C_j$.

\paragraph{Partitioning Clusters into Rings.} We begin by partitioning each cluster $C_j$ into rings centered around its center $a_j$; for $\ell$ satisfying $1 \leq \ell \leq \ell_{max} =  \floor{\log_2(1/\eps)}$, we define the \textit{ring} $R_j(\ell)\subset C_j$ to be the set of points $p \in C_j$ with $\cost(p,a_j) \in [2^{\ell} \Delta_j, 2^{\ell+1} \Delta_j)$. We also let $R_j(0)$ be the points $p \in C_j$ with $\cost(p,a_j) < 2 \Delta_j$ and $R_j(\ell_{max}+1)$ to be the points $p$ satisfying $\cost(p,a_j) \geq 2^{\lmax + 1}\Delta_j$. Clearly, the sets $R_j(0), \ldots, R_j(\ell_{max} +1)$ partition $C_j$.  

We now define the event $\mathcal{E}$.
\begin{definition}[Event $\mathcal E$]
\label{def: event-E}
	The event $\mathcal E$ occurs iff $\Omega$ satisfies the following properties:
	\begin{description}
		\item[$P_1$: (Cluster Size Preservation)] \label{item: event-e-1}   For each cluster $C_j$, we have \[\sum_{q\in \Omega \cap C_j}w_q \in  [(1- \eps)|C_j|,(1+\eps)|C_j|].\]
	\item[$P_2$: (Ring Size Preservation)] \label{item: event-e-2}  For each $j \in [k]$ and $0 \leq \ell \leq \ell_{max}+1$ the set $R_j(\ell)$ satisfies, \[\sum_{q\in \Omega \cap R_j(\ell)}w_q \leq |C_j|/2^{\ell-1}.\]
    \item[$P_3$: (Cluster Cost Preservation)] For each cluster $C_j$, $\costom(C_j, A) = (1 \pm \eps) \cost(C_j, A)$.
	
	\end{description}
\end{definition}

Notice crucially that $\mathcal{E}$ only depends on the sample $\Omega$ (and in particular does \textit{not} place any restriction on $S$).
The following lemma shows that $\mathcal{E}$ holds with high probability. 
\begin{lemma}
\label{lem:E}
    If $m= \Omega(k\eps^{-2}\log (k \eps^{-1}))$ event $\mathcal{E}$ holds with probability at least $1-\eps^3/k^3$.
\end{lemma}

The event $\mathcal{E}$ directly implies that the cost of $\Omega$ with respect to any set of centers $S$ is bounded up to a constant factor by the true cost. We record this observation below (and provide a proof in \Cref{sec: properties-of-sample}).

\begin{lemma}\label{lem: event-e-coreset-cost}
    If event $\mathcal{E}$ holds, then for any set of centers $S$, we have $\costom(P,S) \lesssim \cost(P,S)$.
\end{lemma}

With the good properties of $\Omega$ now laid out, in the next section, we show that it suffices to bound the Gaussian process for a fixed $\Omega$ that satisfies event $\mathcal{E}$.

\subsubsection{Fixing the randomness of \texorpdfstring{$\Omega$}{Omega}}\label{subsec:fixing-omega}

The randomness of $X^S(\Omega,g)$ arises from both the choice of $\Omega$, $g$ and the expectation in \Cref{eq: inter-xs} is over both $\Omega$ and $g$. We proceed as follows to prove \Cref{eq: inter-xs}. First, \Cref{lem: worst-case-omega} shows that for any fixed $\Omega$, the expected supremum $\E_g \sup_{S} \left\lvert X^S(\Omega,g)\right\rvert $ (with expectation only over $g$) is $O(\eps^{-2})$; furthermore, \Cref{lem: omega-event-e} shows that if $\Omega$ satisfies the nice properties guaranteed by event $\mathcal{E}$, then for any $\Omega$ this expected supremum is $O(\eps)$. These facts and the fact that $\mathcal{E}$ occurs with high probability imply \Cref{eq: inter-xs}.

\begin{lemma}[Worst Case Bound]\label{lem: worst-case-omega}
For any fixed $\Omega$, we have  $\E_{g} \sup_{S \inSR} \left\lvert X^S(\Omega, g) \right\rvert \lesssim \eps^{-2} $.
\end{lemma}
The proof of \Cref{lem: worst-case-omega} is relatively straightforward and appears in \Cref{app: proof-worst-case}.
\begin{lemma}[Bound Conditioned on $\mathcal{E}$]\label{lem: omega-event-e}
For any fixed $\Omega$ satisfying event $\mathcal{E}$ (\Cref{def: event-E}) we have $ \E_{g} \sup_{S \inSR} \left\lvert X^S(\Omega,g)\right\rvert \lesssim \eps $.
\end{lemma}

  Before proving \Cref{lem: omega-event-e}, let us see first how the above lemmas imply \Cref{eq: inter-xs}. By the law of total expectation and as $\Pr[\mathcal{E}] \geq (1-\eps^3/k^3)$  by \Cref{lem:E}, we have
\begin{align*}
    \E_{\Omega} \E_g \sup_{S \inSR} \left\lvert X^S\right\rvert &\lesssim \eps \Pr[\mathcal{E}] + \eps^{-2} \Pr[\overline{\mathcal{E}}]  
    \lesssim \eps \cdot 1 + \eps^{-2} \cdot \eps^{3}/k^3 \lesssim \eps,
\end{align*}
as desired. 

The next several sections will be devoted to proving \Cref{lem: omega-event-e}. Henceforth, we fix a set  $\Omega = \{q_1, \ldots, q_m\}$ satisfying $\mathcal{E}$; thus $u^S(\Omega)$ is deterministic and the randomness of $X^{S}(\Omega,g)$ is only due to $g$. To avoid notational clutter, when clear from context, we will write $u^S$ to mean $u^S(\Omega)$ and $X^S$ to mean $X^S(\Omega,g)$.

\subsection{The Chaining Argument}\label{sec: chaining}
We now use a chaining argument to prove \Cref{lem: omega-event-e}.

First, we express each cost vector $u^{S} \in M$ as a telescoping sum of differences of net vectors that approximate it. To do this, we will need the following notation.

 For an integer $h \geq 1$ and $S \in \mathcal{S}(r)$, we define $u^{S,h} \in \R^{\csize}$ to be the net vector from a $(2^{-h}, \mathcal{S}(r))$-net (as defined in \Cref{def: clustering-nets}) that approximates $\uvec{S}$. We also define $\uvec{S,0} \in \R^{\csize}$ as follows: for each  $i \in [m]$, if $q_i$ 
 is from a cluster $C_j \in B(S)$ then  $u^{S,0}_i = \cost(a_j,S)$; else $\uq{S,0}{i} = 0$.

With the above notation, we now write: 
\begin{align} \label{eq: chain}
    \uvec{S} = \uvec{S,0} + \sum_{h = 1}^{\hmax}(\uvec{S,h} - \uvec{S,h-1}) + (\uvec{S} - \uvec{S,\hmax}), 
\end{align} 
where $\hmax =  \ceil{2\log_2(\eps^{-1})}$. The vector $\uvec{S,0}$  is an  ``extremely coarse'' approximation of $\uvec{S}$. The next $\hmax$ summands are differences of net vectors at finer scales. The last summand takes the final step to reach $\uvec{S}$. Using \eqref{eq: chain}, we decompose the random variable $X^S$ (defined in \eqref{def:X}) as follows:  
\begin{align}\label{eq: xs-sum}
    X^S(\Omega,g) = X^{S,\init}(\Omega,g) + \sum_{h = 1}^{\hmax} X^{S,h}(\Omega,g)+X^{S,\fin}(\Omega,g),
\end{align} where we define the random variables
\begin{align}\label{eq: x-init-fin}
    X^{S,\init}(\Omega,g) {}:={} \frac{\sum_{i \in [m]} g_i w_{q_i} \uq{S,0}{i}}{\cost(P,S) } \quad\text{and}\quad  X^{S,\fin}(\Omega,g) := \frac{\sum_{i \in [m]} g_i w_{q_i} (\uq{S}{i} - \uq{S,\hmax}{i})}{\cost(P,S)},
\end{align}
\begin{align}\label{eq: xsh}
   \text{For } 1 \leq h \leq \hmax,  \qquad X^{S,h}(\Omega,g) := \frac{\sum_{i \in [m]}g_i w_{q_i} (\uq{S,h}{i} - \uq{S,h-1}{i})}{\cost(P,S)}.\qquad \qquad
\end{align}
\noindent It follows from \eqref{eq: xs-sum}  and the triangle inequality that,
\[
    \E_{g} \sup_{S \inSR} |X^S(\Omega,g)| \leq \E_{ g} \sup_{S \inSR}|X^{S,\init}(\Omega,g)|+\sum_{h = 1}^ {\hmax} \E_{g} \sup_{S\inSR} |X^{S,h}(\Omega,g)|+\E_{g} \sup_{S\inSR}|X^{S ,\fin}(\Omega,g)|.
\]
Thus, it suffices to bound each term on the right-hand side of the equation above.

\begin{remark}
     We emphasize that the random variables $X^S, X^{S,\init}, X^{S,h}$, and $X^{S,\fin}$ are functions all functions of $\Omega$ and $g$. But since we have fixed an $\Omega$ (that satisfies $\mathcal{E}$), their randomness is entirely due to $g$. Henceforth, we suppress the arguments $\Omega$ and $g$ while writing these random variables.
\end{remark}

The following lemmas bound the suprema of the Gaussian processes $X^{S,\init}$ and $X^{S,\fin}$. 

\begin{lemma}
\label{lem:Xinit}
For any coreset $\Omega$ that satisfies event $\mathcal{E}$ with size at least $\Omega(k \eps^{-2} \log k) $ we have $\E_{g} \sup_{S \inSR} |X^{S,\init}| \leq \eps/6.
$
\end{lemma}

\begin{lemma}
\label{lem:Xfin}For any coreset $\Omega$ that satisfies event $\mathcal{E}$ we have 
   $\E_{ g} \sup_{S\inSR}|X^{S,\fin}| \leq \eps/6$.
\end{lemma}

The proofs of these lemmas use basic arguments and are given in \Cref{sec:lemXinit} and \Cref{sec:lemXfin}, respectively. Handling the summand corresponding to $X^{S,h}$ requires much more work. 
This is accomplished by the following two lemmas, one for the case when the input is $\beta$-stable and the other for worst-case inputs. The proofs of these lemmas are given in \Cref{sec:lemXh}, \Cref{sec: variance-stable} and \Cref{sec:worst-case-variance}.

\begin{lemma}[Bound for Stable Inputs]
\label{lem: stable-xsh}
Let $P$ be a $\Omega(1)$-stable input and $h$ be an integer in $[\hmax]$. If the coreset size is at least $\tilde{\Omega}( k \eps^{-2})$ then we have, 
    $
        \E_{ g}\sup_{S \inSR} |X^{S,h}| \leq \eps/(6\hmax).
   $ 
\end{lemma}

\begin{lemma}[Worst Case Bound]\label{lem: worst-case-xsh}
Let $P$ be any input and $h$ be an integer in $[\hmax]$. If the coreset size is at least $\tilde{\Omega}( k \eps^{-2} \cdot \min(\sqrt{k}, \eps^{-2}))$ then we have, 
    $
        \E_{ g}\sup_{S \inSR} |X^{S,h}| \leq \eps/(6\hmax).
   $ 
\end{lemma}

Together with \eqref{eq:sym}, the above lemmas imply \cref{lem: goal-close}. 

The following section shows how a bound on the variance of $X^{S,h}$ can be traded off with the net size to bound the suprema of $X^{S,h}$.

\subsection{\texorpdfstring{Bounding the Supremum of the Gaussian Process $X^{S,h}$}{Bounding the Supremum of the Gaussian Process XSh}}
\label{sec:lemXh}

In this section, we give a generic upper bound on $\E_{g} \sup_S |X^{S,h}|$ in terms of the variance of $X^{S,h}$ and the size of the net at scale $2^{-h}$. 

First, observe that, for any $h$ satisfying $1 \leq h \leq \hmax$ we have,
\begin{align*}
 \sup_{S\inSR} |X^{S,h}| =  \sup_{S\inSR} \frac{\left\lvert\sum_{i \in [m]}g_i w_{q_i} (\uq{S,h}{i} - \uq{S,h-1}{i})\right\rvert}{\cost(P,S)} 
    \leq \sup_{\substack{(v^{h-1}, v^{h}) \in \\ M_{2^{-(h-1)}} \times M_{2^{-h}}}} \frac{\left\lvert \sum_{i \in [m]}g_i w_{q_i} (v_i^h - v_i^{h-1}) \right\rvert}{\cost(P,S)},
\end{align*}
where $M_{2^{-(h-1)}}$ and $M_{2^{-h}}$ are cost vectors nets at scales $2^{-(h-1)}$ and $2^{-h}$ respectively. While the first supremum was over infinitely many $S$, the final supremum is only over $|M_{2^{-(h-1)}}| \cdot |M_{2^{-h}}| \leq |M_{2^{-h}}|^2$ pairs of vectors. Thus, we have that $\sup_S |X^{S,h}|$ is the maximum of (the absolute value of) $|M_{2^{-h}}|^2$ Gaussians. If $\sigma_h^2$ is an upper bound on $\var[X^{S,h}]$ for $S \in \mathcal{S}(r)$, then we can bound the supremum of Gaussians using \Cref{fact: max-of-gaussians} to obtain:
\begin{align}\label{eq: combine-net-variance}
    \E_g \sup_{S\inSR} |X^{S,h}| \lesssim \sigma_h \sqrt{ \log |M_{2^{-h}}|}.
\end{align}

Therefore, our task reduces to finding an upper bound $\sigma_h^2$ on the variances of the Gaussians $X^{S,h}$ and combining them with the bounds on the net sizes given in \Cref{lem: net-lemma}. We do this for the stable inputs and worst-case inputs in  \Cref{sec: variance-stable} and \Cref{sec:worst-case-variance}, respectively.

\subsection{Bounding the Variance for Stable Inputs}\label{sec: variance-stable}

In the following lemma, we bound the variance of the random variable $X^{S,h}$.

\begin{lemma}[Variance Bound in Terms of the Cost]\label{lem: var-first}Let $S$ be any set of centers and $h\geq1$ be an integer. The random variable $X^{S,h}$ is a mean zero Gaussian with the following upper bound on its variance: 
$\var[X^{S,h}] \lesssim (2^{-2h} \cost(P,A))/(m \cost(P,S)).$
\end{lemma}
\begin{proof}
Recall that $X^{S,h}$ is given by the following expression:
\begin{align*}
X^{S,h} :={}& \frac{\sum_{i \in [m]}g_i w_{q_i} (\uq{S,h}{i} - \uq{S,h-1}{i})}{\cost(P,S)}.
\end{align*}
Since $X^{S,h}$ is a sum of independent centered Gaussians (using \Cref{fact: gaussian}), it is also a centered Gaussian with variance: 
\begin{align}\label{eq: var-0}
    &\var[X^{S,h}] = (1/\cost(P,S))^2 \cdot \sum_{i \in [m]}  w_{q_i}^2 (u_i^{S, h} -u_i^{S, h-1})^2 .
\end{align}
Recall that the net vector $u^{S,h}$ is an $m$-dimensional vector whose entries approximate the cost of points in $B(S) \cap \Omega$ with respect to $S$ (see \Cref{def: clustering-nets}) where $B(S)$ is the set of clusters in band $b$ and type $t$. It follows from the definition of $u^{S,h}$ that,
\begin{enumerate}[(i)]
    \item If $q_i \in B(S)$ then, 
    $$|u_i^{S,h} - u_i^{S,h-1}| \leq |u_i^{S,h} - u_i^{S}| + |u_i^{S,h-1} - u_i^{S}| \leq 2^{-h} \err(q_i, S) + 2^{-(h-1)} \err(q_i, S)\lesssim 2^{-h} \err(q_i, S).$$
    \item If $q_i \notin B(S)$ then $u_i^{S,h} = u_{i}^{S,h-1} = 0$.
\end{enumerate}  Using this in \eqref{eq: var-0}, we get,
\begin{align}\label{eq: var-1}
    \var[X^{S,h}]&\lesssim (1/\cost(P,S))^2 \cdot 2^{-2h} \cdot   \sum_{q \in B(S) \cap \Omega } w_q^2 \err(q,S)^2.
\end{align}
We now show that $\sum_{q \in B(S) \cap \Omega} w_q^2 \err(q,S)^2 \lesssim (\cost(P,A) \cdot \cost(P,S))/m$. This, together with \eqref{eq: var-1}, will then complete the proof of the lemma.  By the definition of $\err(q,S)$ (\Cref{def: clustering-nets}) we have $\err(q,S)^2 \lesssim \cost(q,S) ( \cost(q,A) + \Delta_q) + \cost(q,A)^2 + \Delta_q^2$. By \Cref{lem: weight-bound}, we also have $w_q \cost(q,A) \lesssim \cost(P,A)/m$ and $w_q \Delta_q \lesssim \cost(P,A)/m$. These equations imply that
\begin{align}
\nonumber
\sum_{q \in B(S) \cap \Omega} &w_q^2 \err(q,S)^2 
\lesssim  (\cost(P,A)/m) \cdot \sum_{q \in B(S) \cap  \Omega}  w_q  (\cost(q,S)  + \cost(q,A) + \Delta_q) 
  \\
\lesssim & (\cost(P,A)/m) \cdot \sum_{q \in  \Omega}  w_q  (\cost(q,S)  + \cost(q,A) + \Delta_q)  \tag{Summing over more points}\nonumber\\
\lesssim & (\cost(P,A)/m) \cdot \cost(P,S). \tag{\cref{lem: event-e-coreset-cost}, $A$ is a $O(1)$-approx}
\end{align}
This completes the proof of the lemma.
\end{proof}

Since $A$ is a $O(1)$-approximate solution we have the following immediate corollary.
\begin{corollary}\label{lem: easy-var}
    For any set of centers $S$, $\var[X^{S,h}] \lesssim 2^{-2h}/m$.
\end{corollary}

\subsubsection{Exploiting Stability: Relating Cost to Interaction}
For $\beta$-stable instances, we can get a tighter handle on the variance of $X^{S,h}$. Specifically, we will show that the cost with respect to a set of centers $S$ is proportional to the interaction number $\ccint(S)$ (see \Cref{def: interaction}) and $\beta$; this together with \Cref{lem: easy-var} gives tighter guarantees on the variance (when $N(S)$ and $\beta$ are large).

First, we show that if a center $x\in S$   interacts with a lot of clusters in $B(S)$ (according to \Cref{def: interaction}), then $P$ has a high cost to $S$. The intuition is that if $P$ is stable, its clusters are well separated; thus, if a center interacts with many clusters, most of these clusters pay a significant cost.
\begin{lemma}\label{lem: cost-increase}
	Let $P$ be a $\beta$-stable instance and $S\in \mathcal S(r)$ be a set of centers.
    For a center $x \in S$, 
    let $I(x)$ denote the set of clusters that interact with $x$; then $\cost(P,S) \geq \tfrac{1}{768} (|I(x)|-1) \beta \opt_k$ . 
\end{lemma}
\begin{proof}
	Fix a center $x$ and let $r:= |I(x)|$. Without loss of generality, we can assume that $r \geq 2$ (otherwise, the bound holds trivially). For convenience, we re-index the clusters so that the set of clusters that interact with $x$ is $I(x) = \{C_{1}, \ldots, C_{r}\}$, and such that the clusters are indexed in the increasing order of the distance of their centers from $x$. 
    We show that all clusters except possibly $C_{1}$ have cost $\Omega(\beta \opt_k)$ with respect to $S$. This then implies the desired bound.

Using the guarantee of \Cref{lem: approxstability} on the distance between centers $a_i, a_j$ of stable instances and the triangle inequality, we get	\begin{align}\label{eq: stab-cost-1}
		(\beta \opt_k)/2|C_{j}| \leq \cost(a_{1}, a_{j}) 
        \leq 2(\cost(a_{1}, x) + \cost(a_{j},x)) 
        \leq 4\cost(a_j,x) 
        \leq  64\cost(a_j,S).
	\end{align}	
	The last inequality above uses the fact that $x$ is an approximate nearest center to $a_j$ (see \Cref{def: interaction}). Rearranging we get $\cost(a_j,S) \geq (\beta \opt_k)/(128 |C_j|)$. Since the cost of the center of $C_j$ is high, we can argue that the cost of $C_j$ is lower bounded by $\Omega(|C_j| \cost(a_j, S))$. This is formalized by the following claim whose proof is in \Cref{app: missing-proofs}.

\begin{claim}\label{lem: cost-from-center}
    Let $S$ be a set of $k$ centers. Suppose that $C_j$ is a cluster whose center $a_j$ satisfies $\cost(a_j,S) \geq 32\Delta_j$; then $\cost(C_j,S) \geq \tfrac{1}{6}  |C_j| \cost(a_j,S)$.
\end{claim}

 By \Cref{def: interaction} we have $\cost(a_j,x) \geq 32 \Delta_j$. This together with \eqref{eq: stab-cost-1} and \Cref{lem: cost-from-center} gives $\cost(C_j,S) \geq \tfrac{1}{6} |C_j| \cost(a_j,S) \geq \tfrac{1}{768} \beta \opt_k$.\qedhere
\end{proof}

\begin{lemma}\label{lem: cost-increase-2}
	Let $P$ be a $\beta$-stable instance and $S$ be any set of centers; then we have $\cost(P,S) \geq \opt_k \cdot \max(1,  \tfrac{1}{768} (\tfrac{\ccint(S)}{k} - 1) \beta)$. 
\end{lemma}
\begin{proof}
	 The first bound follows from the fact that we trivially have bound $\cost(P,S) \geq \opt_k$. 
  
  The second bound is a corollary of \Cref{lem: cost-increase}. Let $x^* \in S$ be the center which satisfies  $x^* = \argmax_{x \in S} |I(x)|$. We then have $|I(x^*)| \geq \tfrac{1}{k} \cdot \sum_{x \in S} |I(x)| \geq N(b)/k $. The bound follows by applying  \Cref{lem: cost-increase} to $x^*$.\qedhere
\end{proof}

Combining \Cref{lem: var-first} and  \Cref{lem: cost-increase-2} yields the following variance bound.

\begin{lemma}[Variance Bounds for Stable Instances]\label{lem: stable-var} Let $P$ be a $\beta$-stable instance and $S$ be a set of centers with $N(S) \geq 2k$ then $\var[X^{S,h}] \lesssim (2^{-2h} k)/(m N(S) \beta)$.
\end{lemma}

\subsubsection{Trading the Net and Variance Bounds for Stable Inputs}\label{sec: trading-net-variance-stable}
We have now proved the required net and variance bounds to upper bound $ \E_g \sup_{S \in \mathcal{S}(r)} |X^{S,h}|$. We consider two separate cases depending on the value of $2^r$(recall that $2^r$ is approximately the interaction number of centers considered). 

\paragraph{Case 1 ($2^r \geq 2k$).} \Cref{lem: net-lemma} shows that there exists a $(2^{-h}, \mathcal{S}(r))$ cost vector net $M_{2^{-h}}$ of size $|M_{2^{-h}}| = \exp(O\left((2^r + k 2^{2h})\cdot \log(k \eps^{-1}2^{h})\right))$. Since $2^r \geq 2k$ and $h = O(\log \eps^{-1})$, we have $|M_{2^{-h}}| = \exp(O(2^{2h+r} \log(k \eps^{-1} ))$. By \Cref{lem: stable-var},  for any $S \in \mathcal{S}(r)$, since $N(S)$ lies in the interval $[2^r, 2^{r+1}]$, we have that $\var[X^{S,h}] \lesssim (k \cdot 2^{-r-2h}  )/ (m \beta)$. Combining the variance, net bounds using  \eqref{eq: combine-net-variance} we obtain: 
\begin{align*}
\E_g \sup_{S \in \mathcal{S}(r)} |X^{S,h}| &\lesssim \sqrt{\log |M_{2^{-h}}|} \cdot \sqrt{\var[X^{S,h}]}\\
&\lesssim \sqrt{  2^{r+2h} \cdot\log(k \eps^{-1})} \cdot \sqrt{(k \cdot 2^{-r-2h} )/(m  \beta)}\\
&\lesssim \sqrt{k\log(k \eps^{-1})/(m \beta)}.
\end{align*}
\paragraph{Case 2 $(2^r < 2k)$.}
In this case, we use the following alternate bound on the variance given by \Cref{lem: easy-var}: $\var[X^{S,h}] \lesssim 2^{-2h}/m$. Using the net bounds given by \Cref{lem: net-lemma} and simplifying we also have $ |M_{2^{-h}}|  \leq \exp(O(k2^{2h} \cdot \log(k \eps^{-1})))$. By \Cref{eq: combine-net-variance}, we have 
\begin{align*}
\E_g \sup_{S \in \mathcal{S}(r)} |X^{S,h}|\lesssim \sqrt{k \cdot 2^{2h} \log(k \eps^{-1})} \cdot \sqrt{2^{-2h}/m}   \lesssim \sqrt{k \log(k \eps^{-1})/m}.
\end{align*}
It follows that if $m = \tilde{\Omega}(k \eps^{-2}\max(1, \beta^{-1}) \log(k \eps^{-1}))$ then $\E_g \sup_{S \in \mathcal{S}(r)} |X^{S,h}| \leq \eps$ for any $r$. Rescaling $\eps$ by ${\hmax = O(\log(\eps^{-1}))}$ completes the proof of \Cref{lem: stable-xsh}.

\subsection{Bounding the Variance for Worst Case Inputs} \label{sec:worst-case-variance}
We can also show another lower bound on the $\cost(P,S)$, which we will use to get a new upper bound on the variance. Since all clusters in the set $B(S) = B_{b,t}(S)$ are of type $t$ and band $b$, we can show that their cost wrt $S$ is approximately $2^{b+t} T$; thus the $\cost(P,S)$ is at least $k_{B(S)} \cdot 2^{b+t}$. We formalize these calculations in the lemma below. 
\begin{lemma}\label{lem: structured-cost}
 Let $S$ be a set of centers. Recall that $B(S)$ contains $k_{B(S)}$ clusters, all of which are from band $b$ and are of type $t$ for $S$. If $t\geq 6$ then $\cost(P,S) \gtrsim k_{B(S)} \cdot 2^{b+t}T$.
\end{lemma}
\begin{proof} 
    We trivially have that $\cost(P,S) \geq \cost(B(S), S) = \sum_{C_j \in B(S)} \cost(C_j, S)$. We show a lower bound on the cost of any cluster in $B(S)$.
    
    Fix a cluster $C_j \in B(S)$. Since $C_j$ is of type $t \geq 6$, we have $\cost(a_j,S) \geq 2^{t-1} \Delta_j \geq 32\Delta_j$. Applying \Cref{lem: cost-from-center}, we get $\cost(C_j,S) \geq \tfrac{1}{6}|C_j| 2^{t-1} \Delta_j \gtrsim 2^t \cost(C_j,a_j) \gtrsim 2^{b+t}T$. Summing over all clusters in $B(S)$ yields the desired lower bound on $\cost(P,S)$.
\end{proof}

 Plugging this cost lower bound into the variance bound given by \Cref{lem: var-first}, we get  \begin{align*}
     \var[X^{S,h}] \lesssim 2^{-2h} \cost(P,A)/(m \cdot k_{B(S)} \cdot 2^{b+t}T).
 \end{align*}
 However, in the lemma below, we see that this bound is sub-optimal when the clusters in $B(S)$ have cost less than $\cost(P,A)/k$.
\begin{lemma}\label{lem: worst-case-var}
    Let $S$ be any set of centers. We have $\var[X^{S,h}] \lesssim (2^{-2h-t}k)/(m \cdot k_{B(S)})$.
\end{lemma}
\begin{proof}
If $t < 6$ then the bound follows trivially from \Cref{lem: easy-var} as $\var[X^{S,h}] \lesssim (2^{-2h})/m \lesssim (2^{-2h}k)/(m k_{B(S)} 2^t)$ where we used that $k \geq k_{B(S)}$ and $2^t$ is a constant.

Suppose instead that $t \geq 6$. We proceed similarly as in \Cref{lem: easy-var} to obtain,
\begin{align}\label{eq: var-1-2}
    \var[X^{S,h}]&\lesssim (1/\cost(P,S))^2 \cdot 2^{-2h} \cdot   \sum_{q \in B(S) \cap \Omega } w_q^2 \err(q,S)^2.
\end{align}

Using the definition of $\err(q,S)$ (see \Cref{def: clustering-nets}) we now obtain:
\begin{equation*}
    \sum_{q \in B(S) \cap \Omega }  w_q^2 \cdot \err(q,S)^2 \lesssim  \sum_{q \in B(S)\cap \Omega}  w_q^2  (\cost(q,S) \cost(q,A) + \cost(q,S) \Delta_q  + \cost(q,A)^2 + \Delta_q^2).
\end{equation*}
Next observe that by \Cref{lem: weight-bound}, the weights satisfy $w_q \cost(q,A) \lesssim k \cost(C_j, A)/m$ and also $w_q \Delta_q \lesssim k \Delta_j |C_j|/m) \lesssim k \cost(C_j, A)/m)$ given by \Cref{lem: weight-bound}. Since each cluster  $C_j \in B(S)$ is in band $b$, we have $\cost(C_j, A) \lesssim 2^b T$. Thus we get,
\begin{align*}
    \Gamma := \sum_{q \in B(S) \cap \Omega} w_q^2 \err(q,S)^2 \lesssim (k 2^b T)/m \cdot  \sum_{q \in B(S) \cap \Omega} w_q (\cost(q,S) + \cost(q,A) + \Delta_q). 
\end{align*}
 For any point $q \in B(S) \cap \Omega$, we have $\cost(q,S) \leq 2(\cost(q,A) + 2^t \Delta_q) \lesssim 2^t( \cost(q,A) + \Delta_q)$, where the last inequality holds since we focus on types with $t\geq 6$. Using this, we get,
 \begin{align*}
     \Gamma \lesssim (k 2^{b+t} T)/m \cdot  \sum_{q \in B(S) \cap \Omega} w_q ( \cost(q,A) + \Delta_q).
 \end{align*}
 By property $P_3$ of \Cref{def: event-E} we have 
 \begin{align*}
     \sum_{q \in B(S) \cap \Omega} w_q \cost(q,A) \lesssim \sum_{C_j \in B(S)} \sum_{q \in C_j} w_q \cost(q,a_j) \lesssim k_{B(S)} \cdot 2^{b}T.
\end{align*}
Similarly using property $P_1$ of \Cref{def: event-E}, 
\begin{align*}
  \sum_{q \in B(S) \cap \Omega} w_q \Delta_q \lesssim \sum_{C_j \in B(S)} \sum_{q \in C_j} w_q \Delta_q  \lesssim  \sum_{C_j \in B(S)} \cost(C_j, S) \lesssim k_{B(S)} \cdot 2^{b}T.  
\end{align*}
Using the three previous equations, 
we get $\Gamma \lesssim (k\cdot 2^{2b +t} \cdot  T^2 \cdot k_{B(S)})/m $. The above inequality together with \Cref{eq: var-1-2} and \Cref{lem: structured-cost} gives:
\begin{align*}
    \var[X^{S,h}] \lesssim 2^{-2h}\Gamma/ \cost(P,S)^2     \lesssim 2^{-2h}\Gamma/ (k_{B(S)}\cdot 2^{b+t}T)^2 \lesssim (2^{-2h-t} k)/(m\cdot k_{B(S)}),
\end{align*}
which completes the proof.
\end{proof}
\noindent Using the bound $ N(S) \leq k \cdot k_{B(S)}$ and the previous lemma we obtain the following corollary.
\begin{corollary}\label{cor: worst-case-var} 
Let $S$ be any set of centers; then  $\var[X^{S,h}] \lesssim (2^{-2h-t}k^2)/(m \cdot N(S))$.
\end{corollary}

\subsubsection{Trading off the Net Size and Variance to obtain the Worst Case Bound}\label{sec: completing-the-proof-kmeans}
We shall now use the net and variance bounds shown thus far to bound $\E \sup_{S \in \mathcal{S}(r)}|X^{S,h}|$ for worst-case inputs thus completing the proof of \Cref{lem: worst-case-xsh}. We consider two different cases depending on the interaction number $2^r$.

\paragraph{Case 1 ($2^r \geq 2k)$}
By the variance bounds due to \Cref{lem: easy-var} and \Cref{cor: worst-case-var} we have that for any set of centers  $S \in \mathcal{S}(r)$,
\begin{align*}
    \var[X^{S,h}] \lesssim  \min(1, k^2 2^{-(t+r)}) \cdot 2^{-2h} m^{-1}.
\end{align*}
Using the net bounds from \Cref{lem: net-lemma}, the fact that $2^r + k2^{2h} \lesssim 2^{r+2h}$, we know that there exists a $(2^{-h}, \mathcal{S}(r))$ net $M_{2^{-h}}$ which satisfies:
\begin{align*}
    \log |M_{2^{-h}}| \lesssim  \min(k2^{t}, 2^{r}) \cdot 2^{2h} \log(k\eps^{-1}).
\end{align*}
Using \Cref{eq: combine-net-variance} to combine the corresponding net and variance bounds in the two previous inequalities, we get, 
\begin{align}
    \Gamma := \E_g \sup_{S \in \mathcal{S}(r)} |X^{S,h}| &\lesssim \sqrt{\min(k2^t, k^2 2^{-t}) \cdot \log(k \eps^{-1})m^{-1}} \label{eq: interm-combine}\\
    \notag
    &\lesssim \sqrt{k^{1.5} \cdot \log(k \eps^{-1})m^{-1}}.
\end{align}
where we used the fact that $\min(k2^t, k^2 2^{-t}) \leq k^{1.5}$ for any $t$.
Therefore  if $m =  \Omega(k^{1.5} \eps^{-2} \log(k \eps^{-1}))$ then we have $\Gamma \leq\eps$. Moreover, as $2^t \leq 2^{\tmax}\lesssim \eps^{-2}$, the first bound in \Cref{eq: interm-combine} gives $\Gamma \lesssim \sqrt{k\eps^{-2}\log(k \eps^{-1}) m^{-1}}$; so if  $m = \Omega(k\eps^{-4} \log(k \eps^{-1}))$ then $\Gamma \leq \eps$. Rescaling $\eps$ by $\hmax = O(\log(\eps^{-1}))$ completes the proof of \Cref{lem: worst-case-xsh} for this case.

\paragraph{Case 2 $(2^r < 2k)$.} The calculations in this case are identical to those described for the case of stable inputs (see \Cref{sec: trading-net-variance-stable}). 
The variance bound $\var[X^{S,h}] \lesssim 2^{-2h}/m$ 
 due to \Cref{lem: easy-var}: can be combined with the net bound of  $ |M_{2^{-h}}|  \leq \exp(O(k2^{2h} \cdot \log(k \eps^{-1})))$ 
 by \Cref{lem: net-lemma} to obtain: 
\begin{align*}
\E_g \sup_{S \in \mathcal{S}(r)} |X^{S,h}|\lesssim \sqrt{k \cdot 2^{2h} \log(k \eps^{-1})} \cdot \sqrt{2^{-2h}/m}   \lesssim \sqrt{k \log(k \eps^{-1})/m}.
\end{align*}
It follows that, in this case, if $m = \tilde{\Omega}(k \eps^{-2} \log(k \eps^{-1}))$ then $\E_g \sup_{S \in \mathcal{S}(r)} |X^{S,h}| \leq \eps$.  Rescaling $\eps$ by ${\hmax = O(\log(\eps^{-1}))}$ factors completes the proof.

        \section{Lower Bounds}
\label{sec:lowerbound}

We start by proving that any coreset consisting of input points must have size $\Omega(k/\varepsilon^2)$, even if the data is very separable.

\begin{theorem}\label{thm: lower-bound}
For any $\beta>0$, there exists a $\beta$-stable instance $A$ such that any $k$-means corset $S\subseteq A$ must satisfy $|S|\geq c\cdot k\cdot \varepsilon^{-2}$ for some absolute constant $c$.
\end{theorem}
\begin{proof}
We will argue that for the $1$-mean problem, any coreset must have $\Omega(\varepsilon^{-2})$ points. The claim then follows by adding multiple copies of the $1$-mean instance with an arbitrarily large separation.

 Our instance is the $n$-dimensional simplex, where $n=\varepsilon^{-2}$. More concretely, the points in $A$ correspond to the standard unit basis of $\mathbb{R}^{n}.$
Note that $\mu(A) = \frac{1}{n}\mathbf{1}$ and $\sum_{p\in A} \|p-\mu(A)\|^2 = n-1$.

Let $\Omega$ be the designated coreset with non-negative weight function $w:\Omega\rightarrow \mathbb{R}_{\geq 0}$.
We must have $w(\Omega):=\sum_{p\in \Omega} w_p\geq (1-\varepsilon)\cdot n$, as otherwise the cost of far away center cannot be well approximated.
Consider the point $s:=(\sum_{p\in \Omega}p)/\|\sum_{p\in \Omega}p\|$. Notice that 
We have 
\[\|p-s\|^2 = 2-2/\sqrt{|\Omega|} \text{for any $p\in \Omega$, and }  \|q-s\|^2 = 2
\text{ for every } q\in A\setminus \Omega,
\]
and thus  $\cost(A,s) = 2(n-\sqrt{|\Omega|})$.
At the same time, we know that 
\[\cost(\Omega,s) = w(\Omega)\cdot \left(2-2|\Omega|^{-1/2}\right)\leq 2(1+\varepsilon)\cdot\left(n-\sqrt{|\Omega|} + \sqrt{|\Omega|}-n|\Omega|^{-1/2}\right),\]
and thus
\[\cost(A,s)-\cost(\Omega,s)\geq -\varepsilon\cdot 2(n-|\Omega|^{1/2}) + 2(1+\varepsilon)\cdot (n|\Omega|^{-1/2}-|\Omega|^{1/2}).\]
With some straightforward, but tedious calculation, this term is greater than $\varepsilon\cdot 2(n-|\Omega|^{1/2})$ whenever
$|\Omega|<\left(-\varepsilon^{-1}/2 + \sqrt{\varepsilon^{-2}/4+(1+\varepsilon)\cdot \varepsilon^{-2}}\right)^2\approx 0.61\cdot  \varepsilon^{-2}$.

Therefore, any $1$-mean coreset on $A$ must have size at least $0.61\eps^{-2}$. Let $A'$ be a set of $k$ copies of $A$, placed arbitrarily far away from each other. We argue that a coreset for $k$-means on $A'$ must be a coreset for $1$-mean on most of the copies. 

First, the weight of each copy of $A$ is preserved, up to $(1\pm \eps)$: consider a set of centers $S$ where each copy of $A$ but one gets one center at the origin, and the last copy gets a center far away (say, distance $kn^2$). This copy will dictate the entire cost of the instance, and therefore its weight must be well approximated.

Preserving the weight implies that, for each simplex, the cost of clustering it to its origin is preserved. Suppose now that $|\Omega| \leq \frac{k(10\eps)^{-2}}{100}$. Then, at least $9k/10$ copies get fewer than $0.61\cdot (10\eps)^{-2}$ points: let $A_1, ..., A_{k'}$ be those copies. In $A_i$, the cost of center $s_i$ (defined previously as the average of coreset points in $A_i$) will be underestimated by $10 \eps \cdot 2(n-\sqrt{|\Omega \cap S_i|}) \geq 10\eps (2n-2\sqrt{n}) \geq 10\eps n$. For each other copy, the cost of clustering to the origin is overestimated by at most $\eps n$, since the weight is preserved.

Consider therefore the set of centers consisting the origins and the $s_i$. Its cost is $k/10 n + 9k/10 \cdot 2n- \sum_{i=1}^{k'}\sqrt{|\Omega \cap A_i|} \leq k/10 n + 9kn/5 \leq 2k n$. 
The previous discussion shows that the estimation error from $\Omega$ is at least $9k/10 \cdot 10\eps n - k/10 \cdot \eps n$. This is at least $8k \eps n $, which is strictly more than an $\eps$ fraction of the cost of the full instance. Hence, such a small $\Omega$ cannot be a coreset, which concludes the proof.
\end{proof}

\paragraph{Allowing non-input coreset points.} We further complement these lower bounds by showing that using non-input points, it is possible to compute significantly smaller coresets for certain ranges of $\varepsilon$, $k$, and $\beta$. Therefore, the restriction to using input points is necessary. We did not attempt to optimize the dependency on the ranges of parameters and believe that significantly better parameters are possible.

We will use a slightly different coreset notion known as a coreset with offset originally proposed by \cite{FeldmanSS20}. It requires that for some $\Delta\geq 0$, it holds that 
\[\left\vert \cost(\Omega,S) + \Delta - \cost(P,S)\right\vert \leq \varepsilon\cdot \cost(P,S).\]
In the following, we will show that there exists a coreset with offset of size $k$ if $\beta$ is sufficiently large.
It is possible to set the offset to $0$ by adding two points per center and increasing $\beta$ slightly, but we forgo this to simplify the presentation.

\begin{proposition}
For any cost-stable instance $P$ with $\beta > 512\varepsilon^{-2}$, there exists a coreset with offset of size $k$.
\end{proposition}
\begin{proof}
    Let $C=\{\mu_1,\mu_2,\ldots,\mu_k\}$ be an optimal solution with induced clusters $\{C_1,\ldots,C_k\}$ and let $\Delta = \cost(P,C)$. We claim that $C$, where $\mu_i$ is weighted by $|C_i|$ is the desired coreset.

    Consider some solution $S$. We first observe that if the points in $C_i$ are served by a single center in $s_i\in S$, then $\cost(C_i,s_i) = \cost(C_i,\mu_i) + |C_i|\cost(\mu_i,s_i)$ due to Lemma \ref{lem: centroid-fact}. Therefore, for any solution $S$ where each cluster is served by a single center, the cost is exactly preserved.

    Next, suppose that some cluster $C_i$ is served by at least two centers $s,s'\in S$. Then there must exist two clusters $C_a$ and $C_b$ that are served by the same center $s''$.
    For any cluster $C_i$, we have
    \begin{align*}
& \left\vert \cost(C_i,S) - (|C_i|\cost(\mu_i,S) + \cost(C_i, \mu_i))  \right\vert =        \bigg\vert \sum_{p\in C_i} \cost(p,S) - (\cost(\mu_i,S) + \cost(p,\mu_i)) \bigg\vert\\
&\leq   \sum_{p\in C_i} \varepsilon/2\cdot \cost(p,S) + (8/\varepsilon)\cdot \cost(p,\mu_i) = \varepsilon/2\cdot \cost(C_i,S) + (8/\varepsilon)\cdot \cost(C_i,C).
    \end{align*}
Summed over all clusters, this ensures 
\begin{equation}
    \label{eq:non_input_cost_diff}
    | \cost(P,S) - (\sum_i |C_i|\cost(\mu_i,S) + \cost(P, C))  | \leq \varepsilon/2\cdot \cost(P,S) + (8/\varepsilon)\cdot \opt_k.
\end{equation}
We use the stability property to bound $\opt_k$.
Without loss of generality, assume that $\|\mu_a - s''\| \leq \|\mu_b-s''\|$.
Combined with the properties of $\beta$-stable instances ( \Cref{lem: approxstability}), this ensures $\cost(\mu_b,S) > \cost(\mu_b,\mu_a)/4 \geq \beta\cdot OPT_k/(8|C_b|)$. 
Therefore, we get:
\begin{align*}
    \opt_k &\leq  8 |C_b|\cost(\mu_b, S)/\beta
    \leq (\eps^2/64) \cdot (\sum_i |C_i|\cost(\mu_i,S) + \cost(P, C))\\
    &\leq  (\eps^2/64) \big(\cost(P, S) + \big\vert \cost(P,S) - \big(\sum_i |C_i|\cost(\mu_i,S) + \cost(P, C)\big)  \big\vert\big).
\end{align*}
 Plugging this into \eqref{eq:non_input_cost_diff} and rearranging gives the result.
\end{proof}

\subsection{A Brief Remark on Perturbation Resilient Instances}\label{sec: perturbation}
Finally, we briefly remark on the second most popular stability notion for clustering termed perturbation resilience.
Unlike cost stability, perturbation resilience does not help with coreset construction.
We first define the condition and then argue why it cannot help.
\begin{definition}
    Let $(X,d)$ be a metric. We say that $(X,d')$ is an $\alpha$-perturbation of $(X,d)$ if for any pair of points $p,q\in X$
    $$d(p,q) \leq d'(p,q) \leq \alpha\cdot d(p,q).$$
\end{definition}
We note that the function $d'$ is not necessarily a distance function, that is $(X,d')$ might not be a metric. 
\begin{definition}
    Let $(X,d)$ be a metric. For any set of points $P\subseteq X$, let $C_{opt} = \{C_1,\ldots C_k\}$ be the optimal $k$-means clustering of $P$. We say that $P$ is $\alpha$-perturbation resilient, if for any  $(X,d')$, $C_{opt}$ is an optimal $k$-means clustering of $P$.
\end{definition}

\begin{proposition}
For any $\alpha$-perturbation resilient instance $P$, any $k$-means coreset must have size $\Omega(k\cdot \varepsilon^{-2}\min(\sqrt{k},\varepsilon^{-2}) \cdot \alpha^{-8})$.
\end{proposition}
\begin{proof}
The argument is very reliant on the worst-case lower bound construction by \cite{HLW23}.
We will only describe the properties that are necessary to ensure that their lower bound applies to a perturbation-resilient instance and encourage the reader to refer to their proof for details.
The simplest variant of their lower bound assumes that $\varepsilon^{-1} = k^{1/4}.$
The points $P$ are split into $k/2$ groups with the following properties.
\begin{enumerate}
    \item The groups have equal size $\Theta(k)$,
    \item Two points in the same group have pairwise distance $\Theta(1)$,
    \item Two points in different groups have pairwise distance $\Theta(\varepsilon^{-1})$.
\end{enumerate}

Notice that this instance would be perturbation-resilient for $\alpha\in \Theta(k^{1/4})$, if we had $k$ groups, each inducing a cluster.
To create a perturbation-resilient point set, we introduce a modified copy $P'$ of $P$ located at an arbitrary distance with the following properties:
\begin{enumerate}
    \item The groups have equal size $\Theta(k)$,
    \item Two points in the same group have pairwise distance $\Theta(1)$,
    \item Two points in different groups have pairwise distance $\Theta(\alpha)$.
\end{enumerate}

Observe that now the groups of $P\cup P'$ are the clusters of an optimal solution and that the instance is $\Omega(\alpha)$-perturbation stable. We further observe that if we serve all the points of $P'$ with one center, the cost increases by roughly a factor $\Theta(\alpha^2)$.
However,  there are now $k/2$ unused centers with which $P$ becomes a worst-case instance for coresets with $\varepsilon' = \Theta(\varepsilon/\alpha^2)$. Thus, a coreset lower bound of $\Omega(k\cdot \varepsilon^{-2}\min(\sqrt{k},\varepsilon^{-2}) \cdot \alpha^{-8})$ for $\alpha$-perturbation resilient instances follows from the $\Omega(k\cdot \varepsilon^{-2}\min(\sqrt{k},\varepsilon^{-2}))$ worst case lower bound of \cite{HLW23}.
\end{proof}
        \bibliographystyle{alpha}
        \bibliography{general}
        \appendix
        \section{Useful Facts}\label{app: missing-proofs}
Here, we collect some standard facts that we have used repeatedly.
\begin{fact}[Bernstein's Inequality]\label{lem: bernstein}
	Let $Y_1, \ldots, Y_\ell$ be independent mean-zero bounded random variables satisfying $|Y_i| \leq M$, and let $S=\sum_{i=1}^\ell Y_i$ and $\sigma^2 =\sum_{i = 1}^\ell \E[Y_i^2]$. Then, for all $t>0$,
 \begin{align*}
		\Pr\left( |S| \geq t \right) \leq 2\exp\left(-t^2/(2\sigma^2 + 2Mt/3)\right).
	\end{align*}
	
\end{fact}
\begin{fact}[Approximate Triangle Inequality]\label{fact: triangle}
   Let $p_1, p_2, p_3, p_4$ be points in $\R^d$. We then have,
    \begin{enumerate}[(i)]
    \item\label{item:weaktri} \label{item: tri-1} For any $\eps > 0$, $\|p_1-p_3\|^2 \leq (1+\varepsilon)\|p_1-p_2\|^2 + \left(1+\frac{1}{\varepsilon}\right) \|p_2-p_3\|^2$
        \item \label{item: tri-2}$\left|\norm{p_1 - p_3}^2 - \norm{p_1- p_2}^2 \right| \leq 2 \cdot\norm{p_1 - p_2} \cdot \norm{p_2-p_3} + \norm{p_2- p_3}^2.$
         \item \label{item: tri-3} $\norm{p_1 - p_3}^2 \leq 2 (\norm{p_1 - p_2}^2 + \norm{p_2 - p_3}^2)$.    
         \item \label{item: tri-4} $\norm{p_1 - p_4}^2 \leq 3(\norm{p_1 - p_2}^2 + \norm{p_2 - p_3}^2 + \norm{p_3 - p_4}^2)$
    \end{enumerate}

\end{fact}
\begin{proof}
    The first inequality follows by squaring  $\norm{p_1-p_3} \leq \norm {p_1-p_2} + \norm{p_2-p_3}$ and using that $2ab \leq \eps a^2 + \eps^{-1} b^2$, for any $a,b$. 
    The second and third inequality follows from the first by setting $\eps = \norm{p_2 - p_3}/\norm{p_1 - p_2}$ and  $\eps = 1$ respectively. The fourth inequality follows by squaring
  $\norm{p_1 - p_4} \leq (\norm{p_1 - p_2} + \norm{p_2 - p_3} + \norm{p_3 - p_4})$ and as  $(a+b+c)^2 \leq 3(a^2 + b^2 + c^2)$ by  Cauchy-Schwarz. \qedhere
\end{proof}

\begin{fact}\label{fact: gaussian}
    Let $a_1, \ldots, a_t$ be reals and $g_1, \ldots, g_t$ be independent $N(0,1)$ random variables; then $\sum_{i=1}^t {a_i g_i}$ has the same distribution as $ N(0, \sum_{i=1}^t a_i^2)$.
\end{fact}

\begin{fact} [See for example \cite{kamath2015bounds}]\label{fact: max-of-gaussians}
	For $n \geq 2$ and $i = 1,\ldots, n$, let $g_i \sim \mathcal{N}(0,\sigma_i^2)$ be Gaussians such that $\max_i \sigma_i \leq \sigma$; then $\E \max_{i \in [n]} |g_i|\leq \sigma\sqrt{2\ln 2n}$.
\end{fact}

We have the following standard property of the centroid.
\begin{fact}\label{lem: centroid-fact}
    Let $P \subset \R^d$ be a set of $n$ points and $\Delta(P) = \tfrac{1}{n}\sum_{p \in P} p$ be the centroid. For any $x \in \R^d$ we have,
    $\sum_{p \in P} \norm{x - p}^2 = n \norm{x - \Delta(P)}^2 + \sum_{p \in P} \norm{\Delta(P) - p}^2 $.
\end{fact}

\begin{claim}
    Let $S$ be a set of $k$ centers. Suppose that $C_j$ is a cluster whose center $a_j$ satisfies $\cost(a_j,S) \geq 32\Delta_j$; then $\cost(C_j,S) \geq \tfrac{1}{6}  |C_j| \cost(a_j,S)$
\end{claim}
\begin{proof}
    We show that since the centers in $S$ are significantly far away from $S$, at least half the points in cluster $C_j$ pay a cost of $\Omega(\cost(a_j,S))$ to $S$.
    
    Let $p$ be a point of $C_j$ that satisfies $\cost(a_j,p) \leq 2\Delta_j$. For any such point, we have,
    \begin{align*}
        |\cost(p,S)- \cost(a_j,S)| &\leq 2 \sqrt{\cost(a_j,S) \cost(a_j,p)} + \cost(a_j,p) \tag{By \Cref{item: tri-2} of \Cref{fact: triangle}}\\
        &\leq 2\cost(a_j,S)/3 \tag{as $\cost(a_j,p) \leq 2 \Delta_j \leq \cost(a_j,S)/16$}
    \end{align*}
    Rearranging terms we get $\cost(p,S) \geq \cost(a_j,S)/3$. By Markov's inequality the number of points $p$ with $\cost(a_j,p) \leq 2 \Delta_j$ is at least $|C_j|/2$. It therefore follows that $\cost(C_j,S) \geq |C_j|/2 \cdot \cost(a_j,S)/3 \geq \tfrac{1}{6} |C_j| \cost(a_j,S)$.
\end{proof}

\section{Basic Properties of the Algorithm}

\subsection{Properties of the Sample}\label{sec: properties-of-sample}
\noindent We state some properties of \Cref{alg: sensitivity-sampling} which are useful in the analysis. Firstly, since the weight of a point is the inverse of its probability of being sampled, it is easily verified from \eqref{defn:mu}  that the weights satisfy the following bound. 
\begin{fact}\label{lem: weight-bound}
	For a point $q \in \Omega$ from cluster $C_j$, its weight satisfies  
	\begin{equation*}
 w_q \leq 4\min\left(\frac{k\,|C_j|}{\csize}, \frac{k \, \cost(C_j,A)}{\csize\, \cost(p,A)}, \frac{\cost(P,A)}{\csize\, \cost(q,A)}, \frac{\cost(P,A)}{\csize \, \Delta_q}\right). 
 \end{equation*}
\end{fact}
These bounds will be extremely useful to bound the variance of various quantities.

Secondly, we show that with high probability, the coreset preserves the number of points in each cluster as well as its cost. Moreover, it ensures that the cluster does not over-sample high-cost points. These properties are summarized by an event $\mathcal{E}$ defined below.

Before we do this, it will be convenient to first partition points in a cluster into rings according to their cost from the center. We define the notation $\Delta_j = \cost(C_j,A)/|C_j|$ to be the average cost of cluster $C_j$.

\paragraph{Partitioning Clusters into Rings.} We begin by partitioning each cluster $C_j$ into rings centered around its center $a_j$; for $\ell$ satisfying $1 \leq \ell \leq \ell_{max} =  \floor{\log_2(1/\eps)}$, we define the \textit{ring} $R_j(\ell)\subset C_j$ to be the set of points $p \in C_j$ with $\cost(p,a_j) \in [2^{\ell} \Delta_j, 2^{\ell+1} \Delta_j)$. We also let $R_j(0)$ be the points $p \in C_j$ with $\cost(p,a_j) < 2 \Delta_j$ and $R_j(\ell_{max}+1)$ to be the points $p$ satisfying $\cost(p,a_j) \geq 2^{\lmax + 1}\Delta_j$. Clearly, the sets $R_j(0), \ldots, R_j(\ell_{max} +1)$ partition $C_j$.  

We now define the event $\mathcal{E}$.
\begin{definition}[Event $\mathcal E$]
	The event $\mathcal E$ occurs iff $\Omega$ satisfies the following properties:
	\begin{description}
		\item[$P_1$: (Cluster Size Preservation)] 
  For each cluster $C_j$, we have \[\sum_{q\in \Omega \cap C_j}w_q \in  [(1- \eps)|C_j|,(1+\eps)|C_j|].\]
	\item[$P_2$: (Ring Size Preservation)] 
 For each $j \in [k]$ and $0 \leq \ell \leq \ell_{max}+1$ the set $R_j(\ell)$ satisfies, \[\sum_{q\in \Omega \cap R_j(\ell)}w_q \leq |C_j|/2^{\ell-1}.\]
        \item[$P_3$: (Cluster Cost Preservation)] For each cluster $C_j$, $\costom(C_j, A) = (1 \pm \eps) \cost(C_j, A)$.
	\end{description}
\end{definition}

Notice crucially that $\mathcal{E}$ only depends on the sample $\Omega$ (and in particular does \textit{not} place any restriction on $S$).
The following lemma shows that $\mathcal{E}$ holds with high probability. 
\begin{lemma}
    If $m= \Omega(k\eps^{-2}\log (k \eps^{-1}))$ event $\mathcal{E}$ holds with probability at least $1-\eps^3/k^3$.
\end{lemma}
\noindent We now show that each of the three properties holds with high probability. The above lemma then follows by a union bound. We first consider the properties $P_1$ and $P_2$.
\begin{lemma}
 If $\csize \geq 48 k\eps^{-2} \log(10k\eps^{-1})$, then 
 $\Pr[P_1] \geq 1-\eps^4/5k^3$.
\end{lemma}
\begin{proof}
This follows from a standard application of probabilistic tail bounds.
 Fix a cluster $C_j$, and consider the random variable $W=\sum_{q \in C_j\cap \Omega} w_q$. We will show that
 \[  \Pr[ |W-|C_j| | > \eps |C_j|] \leq \eps^4/5k^4.\]
For $1 \leq i \leq \csize$, define the random variable $X_i$ to be the weight $w_{q_i}$ of the $i$-th sample $q_i$ if $q_i \in C_j$ and $0$ otherwise. Note that the $X_i$ are independent, $\sum_{i=1}^\csize  X_i = W$ and as  $w_q \mu(q)=1/\csize$, we have
\[\E[X_i] = \sum_{q \in C_j}  \mu(q) w_q= |C_j|/\csize,\] and therefore $ \E[W] = |C_j|$. 

Next, by \Cref{lem: weight-bound}, $w_q \leq 4k|C_j|/\csize$ for each $q\in C_j$ and thus $X_i \in [0,4k|C_j|/\csize]$, so that
\[ \E[X_i^2] \leq  \sum_{q \in C_j} \mu(q) w_q^2 = \sum_{q \in C_j} w_q/ \csize\leq \sum_{q \in C_j} 4k|C_j|/\csize^2  \leq 4k|C_j|^2 /\csize^2. \]
Let $Y_i = X_i - \E[X_i]$. Then $\sum_{i=1}^\csize Y_i = W- |C_j|$  and thus by Bernstein's inequality,
	\begin{align*}
		\Pr (|W -|C_j|| \geq \eps |C_j|) &
		\leq 2\exp\left(\frac{-(\eps|C_j|)^2/2}{4k|C_j|^2 /\csize + 4k\eps|C_j|^2/3\csize}\right) \\ 
		&= 2\exp\left(- \frac{\eps^2 \csize}{12k}\right) = 2 \left(\eps / 10k\right)^4 \leq \eps^4/5k^4.
	\end{align*}
where we use that $\sum_{i=1}^\csize \E[Y_i^2] = \sum_{i=1}^\csize \var[X_i] \leq  \sum_{i=1}^\csize \E[X_i^2] \leq 4kC_j^2/\csize$.
The bound on $\Pr[P_1]$ follows by a union bound over the $k$ clusters. 
\end{proof}
The same proof shows that $P_2$ occurs with high probability:
\begin{lemma}
     If $\csize \geq 48 k\eps^{-2} \log(10k\eps^{-1})$, then 
 $\Pr[P_2] \geq 1-\eps^3/5k^3$.
\end{lemma}
\begin{proof}
    Following the exact same steps as in the proof of property $P_1$, one can show that with probability at least $1-\eps^4/5k^3$, for all cluster $C_j$ and ring $R_j(\ell)$, it holds that $\sum_{q \in \Omega \cap R_j(\ell)} w_q = (1\pm \eps) |R_j(\ell)|$. 
    The only differences with the proof of property $P_1$ are that the variable $X_i$ is 0 when $q_i \notin R_j(\ell)$, instead of $q_i \notin C_j$; and the last union-bound should be over all the $k \log_2(1/\eps)$ many rings, instead of the $k$ clusters. 
    
    We can conclude from this: Markov's inequality ensures that $|R_j(\ell)| \leq |C_j|/2^\ell$, as otherwise the cost of points in $R_j(\ell)$ would exceed $\cost(C_j, A)$. Therefore, taking $\eps = 1$, we get $\sum_{q \in \Omega \cap R_j(\ell)} w_q \leq |C_j|/2^{\ell-1}$.
\end{proof}

Next, we show that property $P_3$ of \Cref{def: event-E} also occurs with high probability. 
\begin{lemma}
    If $\csize \geq 48 k\eps^{-2} \log(10k\eps^{-1})$, then $\Pr[P_3] \geq 1 - \eps^4/5k^3$.
\end{lemma}
\begin{proof}
The proof is almost the same as the previous properties, and we repeat it for completeness.

 Fix a cluster $C_j$, 
for $1 \leq i \leq \csize$, define the random variable $X_i$ to be the cost $w_{q_i} \cost(q_i, A)$ of the $i$-th sample $q_i$ if $q_i \in C_j$ and $0$ otherwise. Note that the $X_i$ are independent, that $\sum X_i = \costom
(C_j, A)$, and as  $w_q \mu(q)=1/\csize$,  we have
\[\E[X_i] = \sum_{q \in C_j}  \mu(q) w_q \cost(q, A) = \cost(C_j, A)/\csize.\]
Next, by \Cref{lem: weight-bound}, $w_q \leq \frac{4k \cost(C_j, A)}{\csize \cost(q, A)}$ for each $q\in C_j$ and therefore $|X_i| \leq 4k \cost(C_j, A) / \csize$. We also have: 
\begin{align*}
\E[X_i^2] &\leq  \sum_{q \in C_j} \mu(q) w_q^2 \cost(q, A)^2 
= \sum_{q \in C_j} w_q \cost(q, A)^2\\
&\leq \sum_{q \in C_j} 4k \cost(C_j, A) \cost(q, A)/\csize^2  \leq 4k\cost(C_j, A)^2 /\csize^2.    
\end{align*}
Let $Y_i = X_i - \E[X_i]$. Then $\sum_{i=1}^\csize Y_i = \costom(C_j, A) - \cost(C_j, A)$  and thus by
 Bernstein's inequality,
	\begin{align*}
		\Pr \left(|\costom(C_j, A) - \cost(C_j, A)| \geq \eps \cost(C_j, A)\right) &
		\leq 2\exp\left(\frac{-(\eps \cost(C_j, A))^2/2}{4k\cost(C_j, A)^2 /\csize + 4k\eps\cost(C_j, A)^2/3\csize}\right)\\ 
		&= 2\exp\left(- \frac{\eps^2 \csize}{12k}\right) \leq \eps^4/5k^4.
	\end{align*}
where we use that $\sum_{i=1}^\csize \E[Y_i^2] = \sum_{i=1}^\csize \var[X_i] \leq  \sum_{i=1}^\csize \E[X_i^2] \leq 4k\cost(C_j, A)^2/\csize$.
The bound on $\Pr[P_3]$ follows by a union bound over the $k$ clusters. 
\end{proof}

This completes the proof of \Cref{lem:E}. We now show that the event $\mathcal{E}$ directly implies that the cost of $\Omega$ with respect to any set of centers $S$ is bounded above by a constant factor times the true cost. 

\begin{lemma}
    If event $\mathcal{E}$ holds, then for any set of centers $S$, we have $\costom(P,S) \lesssim \cost(P,S)$.
\end{lemma}
\begin{proof} 
Suppose event $\mathcal{E}$ holds. For any point $q$ in cluster $C_j$, by the triangle inequality $\cost(q,S) \leq 2(\cost(q, a_j) + \cost(a_j, S))$. As  $\costom(P,S) = \sum_{j=1}^k \sum_{q \in C_j \cap \Omega} w_q \cost(q,S)$, this gives
    \begin{align}
     \costom(P,S) 
    \leq \sum_{j=1}^k \sum_{q \in C_j \cap \Omega} 2w_q\,\cost(q,a_j) + \sum_{j=1}^k\sum_{q \in C_j \cap \Omega}2w_q\,\cost(a_j,S) 
    \label{eq:frt6}
      \end{align}
     The first term can be bounded using property $P_3$,
\begin{equation} 
\sum_{j=1}^k \sum_{q \in C_j \cap \Omega} w_q\,\cost(q,a_j) = \sum_{j=1}^k \costom(C_j,A) \leq (1+\eps) \cost(C_j,A).
\label{eq:frt8}
\end{equation}
The second term of \Cref{eq:frt6} can be bounded since property $P_1$ holds (i.e., for each cluster $C_j$ we have 
 $(\sum_{q \in C_j \cap \Omega} w_q) \leq(1+\eps) |C_j|$).
\begin{align} & \sum_{j=1}^k\sum_{q \in C_j \cap \Omega} w_q\,\cost(a_j,S) = \sum_{j=1}^k \cost(a_j,S) (\sum_{q \in C_j \cap \Omega} w_q) \leq \sum_{j=1}^k \cost(a_j,S) ((1+\eps)|C_j|) \notag \\
& \leq \sum_{j=1}^k 2 (1+\eps) (\cost(C_j,S) + \cost(C_j,A)) = 2(1+\eps)(\cost(P,S) + \cost(P,A)),
\label{eq:frt7}
\end{align}
where the last inequality uses that for any set of centers $S$ and a cluster $C_j$, we have that
\begin{equation*}
\cost(a_j, S) |C_j|\leq 2(\cost(C_j, S) + \cost(C_j,A)),
\end{equation*}
which follows noting that for any point $q\in C_j$, $ \cost(a_j,S)\leq 2(\cost(q,S) + \cost(q,a_j))$ by the triangle inequality, and summing up over all points $q\in C_j$.

Plugging the bounds \eqref{eq:frt8}
 and \eqref{eq:frt7} in \eqref{eq:frt6}, and using that $\eps\leq 1/2$, gives the claimed bound  
 \[ \costom(P,S) \leq 9 \,\cost(P,A) + 6\, \cost(P,S).   \]
 Since $A$ is a $O(1)$ approximation the claim then follows. \qedhere
\end{proof}
 Next, we give a worst-case bound on the cost estimated by $\Omega$ when the event $\mathcal{E}$ does not occur.
\begin{lemma}\label{lem: worst-case-coreset-cost}
    For any set of centers $S$, $\costom(P,S) \lesssim k \,\cost(P,S)$.
\end{lemma}
\begin{proof}
The proof shares similarity with the one of \cref{lem: event-e-coreset-cost}, and we re-use several parts of it.
We start with the bound in \eqref{eq:frt6} on $\costom(P,S)$. By \eqref{eq:frt8}, the first term is always bounded by $6 \cost(P,A)$. To bound the second term, we first note that 
for any $q\in C_j$, 
\begin{equation*}
    w_q\, \cost(a_j,S)  \leq \frac{3k|C_j|}{m} \cost(a_j,S)  \leq \frac{3k}{m}  \left(2\,\cost(C_j,S)+2\,\cost(C_j,a_j)\right)
\end{equation*}
where we use that $w_q\leq 3k|C_j|/m$ by \Cref{lem: weight-bound}, and that $|C_j|\cost(a_j,S) \leq  \sum_{p\in C_j} (2 \cost(p,S) + 2\cost(p,a_j))$ by the triangle inequality.

As $|C_j \cap \Omega| \leq |\Omega|=m$ for each $j$, 
we can bound the second term in \eqref{eq:frt6} using the equation above as
\begin{align*}\sum_{j=1}^k \sum_{q \in C_j \cap \Omega} 2w_q \cost(a_j,S) & \leq m \sum_{j=1}^k  \frac{6k}{m} (2\cost(C_j,S)+2\cost(C_j,a_j)) \\
& = O(k (\cost(P,S) + \cost (P,A)). 
\end{align*}
Since $A$ is a $O(1)$ approximation the claim then follows.
\end{proof}

\section{Stability of an Approximate Clustering}\label{ap:computeStable}
Given a $\beta$-stable instance and an optimal $k$-means  solution $\{a_1, \ldots, a_k\}$, we can show that the centers $a_i$ are pairwise well separated. Moreover, this property holds even for sufficiently good approximate $k$-means solutions. 

\begin{lemma}
\label{lem: approxstability}
    Let $P \subseteq \R^d$ be a $\beta$-stable instance with respect to the $k$-means objective. Let $\opt_k$ be the value of the optimal solution with $k$ centers. For $1 \leq \gamma \leq 1 + \beta/2$, let $\{C_1, \ldots C_k \}$ be a  $\gamma$-approximation to the optimal solution and $a_1, \ldots, a_k$ be the respective centroids of the clusters. Then for any $i, j \in[k]$ with $i \neq j$, it holds that $$\cost(a_i, a_j) \geq  \frac{ \beta \opt_k}{2\min\left(|C_i|, |C_j|\right)}.$$    
\end{lemma}
\begin{proof}
    Fix two distinct clusters $C_i$ and $C_j$ and assume without loss of generality assume that $|C_i| \leq |C_j|$. Since the solution is $\gamma$-approximate with $\gamma \leq 1 + \beta/2$, we have $\cost(P,A) \leq (1+ \beta/2) \opt_k$.  By the $\beta$-stability of the instance, $\opt_{k-1} \geq (1+\beta) \opt_k$. Therefore, if $A' = A \setminus \{a_i\}$, we have $\cost(P,A') \geq (1+\beta) \opt_k$. It follows that: 
    \begin{align*}
        \frac{\beta \opt_k}{2}  &\leq \cost(P,A') - \cost(P,A)\\
        & = \cost(C_i,A') - \cost(C_i, a_i) \tag{For $j \neq i$,  $\cost(C_j, A) = \cost(C_j, A') = \cost(C_j,a_j)$}\\
        & \leq \cost(C_i, a_j) - \cost(C_i, a_i) \tag{$\cost(p,a_j) \geq \cost(p, A')$ since $a_j \in A'$}\\
        & =|C_i| \cdot \cost(a_i, a_j)\tag{Using \Cref{lem: centroid-fact} and that $a_i$ is the centroid of $C_i$}.
    \end{align*}
    Rearranging terms proves the inequality.
\end{proof}

Fortunately, even for small values of $\beta$, it is possible to compute the desired $\gamma$-approximation in time $n^{O(1/\gamma)}$, see \cite{ABS10}. This dependency on $n$ can be reduced to $\text{poly}(k/\eps)$ via an initial coreset computation. Moreover, a polynomial dependency on $1/\gamma$ is not possible as Euclidean $k$-means is APX hard \cite{Cohen-AddadSL21}. For large (constant) values of $\beta$, more conventional (and fast) approximation algorithms are sufficient. The fastest known true approximation algorithms for Euclidean $k$-means are local search based, see \cite{ChooGPR20,KMNPSW04,LattanziS19}. Moreover, local search is also an algorithm that can be used to solve stable $k$-means \cite{Cohen-AddadS17}, so one can always seed the coreset construction with the computed solution of local search.

Thus, for constant $\beta$, which is the regime of interest, the algorithm always runs in polynomial time, and in fact, for sufficiently large $\beta> 18$, there is effectively no drop-off in performance.

\section{\texorpdfstring{Missing Details from \Cref{sec: analysis}}{Missing Details from Analysis}}

\subsection{Analysis of Far Points}
\label{sec:far-points}
We now prove \Cref{lem: goal-far}, which says that for any set of centers, $S$ the cost of coreset points from far clusters approximates the cost of far clusters to within an error of $\pm \eps \,\cost(P,S)$. 
More formally, our goal is to show  that when $\csize = \Omega( k \eps^{-2}  \cdot \log(k \eps^{-1}))$, then
	\begin{align*}
		\begin{split}
			\E_{ \Omega} \sup_{S\inSR} \left\lvert \frac{\cost(P_F(S),S) - \costom(P_F(S),S)}{\cost(P,S)} \right\rvert \leq \eps/2,
		\end{split}
	\end{align*}
where $P_F(S)$ is the set of clusters $C_j$ such that $\cost(a_j, S) > \Delta_j \eps^{-2}$.

The idea is that when $C_j$ is a far cluster, then  $a_j$ is so far away from $S$ that the cost of any point of $C_j$ will essentially be $\cost(a_j, S)$. This might not be true for all $q \in C_j$, as some $q$ may be far away from $a_j$ as well, but it holds for most of the points, which is enough for us. 
More precisely, we show in \Cref{lem: far-cluster} below that $\cost(C_j, S) \approx |C_j| \cost(a_j, S)$ and $\costom(C_j, S) \approx \sum_{q \in \Omega \cap C_j} w_q \cdot \cost(a_j, S)$.
When the coreset $\Omega$ satisfies event $\mathcal E$, those two quantities are (almost) equal, and the cost difference for any far cluster is at most $\eps \cost(C_j, S)$. 
Note that this holds deterministically for any $S$, as long as $\Omega$ satisfies event $\mathcal E$: therefore, 
\[\E_{ \Omega} \left[\sup_{S\inSR} \left\lvert \frac{\cost(P_F(S),S) - \costom(P_F(S),S)}{\cost(P,S)} \right\rvert \mathcal{E}   \right] \leq \eps/2.\]
When $\Omega$ does not satisfy event $\calE$, then we bound $|\cost(P_F(S),S) - \costom(P_F(S),S)| \leq \cost(P,S) + \costom(P,S) \leq O(k\, \cost(P, S))$ (from \cref{lem: worst-case-coreset-cost}). This holds simultaneously for all solution $S$. 

Since $\Omega$ satisfies event $\calE$ with probability at least $1-\eps/k$ (\Cref{lem:E}), the Law of Total Expectation ensures that
\[\E_{ \Omega} \sup_{S\inSR} \left\lvert \frac{\cost(P_F(S),S) - \costom(P_F(S),S)}{\cost(P,S)} \right\rvert \lesssim \eps,\] 
which concludes the proof of \Cref{lem: goal-far}.

What remains to show is the following:
\begin{lemma}\label{lem: far-cluster}
If the coreset $\Omega$ satisfies the event $\mathcal E$ (\Cref{def: event-E}), then for any set of centers $S$, any cluster $C_j$ that is far from $S$  (deterministically) satisfies  $|\cost (C_j,S) - \costom(C_j, S)| \lesssim  \eps\, \cost(C_j, S).$
\end{lemma}
\begin{proof}
Adding and subtracting $|C_j| \cdot \cost(a_j,S)$ and applying the triangle inequality we have, $$|\cost(C_j,S)- \costom(C_j,S)| \leq \underbrace{|\cost(C_j,S) - |C_j| \cdot \cost(a_j,S)|}_{\text{Term 1}} + \underbrace{|\costom(C_j,S) - |C_j| \cdot \cost(a_j,S)|}_{\text{Term 2}}$$
The lemma follows if bound both terms are $O(\eps\, \cost(C_j, S)$. We show this  next.
\paragraph{Bounding Term 1.}
As the rings $R_j(0), \ldots, R_j(\lmax + 1)$ partition $C_j$, we have
\begin{equation*}
     \text{Term 1} \leq \sum_{\ell= 0}^{\ell_{max}+1} |\cost(R_j(\ell), S) - |R_j(\ell)| \cdot \cost(a_j,S)| \leq \sum_{\ell= 0}^{\ell_{max}+1} \sum_{p \in R_j(\ell)} |\cost(p,S) - \cost(a_j,S)|.
\end{equation*} 
We start by bounding the term corresponding to points in $R_j(\lmax + 1)$. These points are at most $\eps |C_j|$ in number and therefore are easily dealt with.

Let $p$ be a point in $R_j(\lmax +1)$. By triangle inequality (\Cref{item: tri-3} of \Cref{fact: triangle}), we have  $\cost(p,S) \leq 2(\cost(p,a_j) + \cost(a_j,S))$. Therefore, $|\cost(p,S) - \cost(a_j,S)| \lesssim \cost(p,a_j) + \cost(a_j,S)$. Noting that there are at most $\eps |C_j|$ points in $R_j(\lmax + 1)$ we get,
\begin{align*}
     \sum_{p \in R_j(\lmax +1)} &|\cost(p,S) - \cost(a_j,S)| \lesssim \cost(C_j,a_j) + \eps |C_j| \cost(a_j,S).
\end{align*}
Here, note that $|C_j|\cost(a_j, S) \leq \cost(C_j, a_j) + \cost(C_j, S)$. Furthermore, 
since $C_j$ is a far cluster, we have that $\cost(C_j, a_j) \leq \eps^2 |C_j| \cost(a_j, S) \lesssim \eps^2 (\cost(C_j, S) + \cost(C_j, a_j)$. Rearranging, we get $\cost(C_j, a_j) \lesssim \eps^2 \cost(C_j, S)$. Therefore, 
\begin{align*}
     \sum_{p \in R_j(\lmax +1)} &|\cost(p,S) - \cost(a_j,S)|\lesssim \eps^2 \cost(C_j,S) + \eps \cost(C_j,S) \lesssim \eps \cost(C_j,S).
\end{align*}
Let $p$ be a point in ring $\ell$ where $0 \leq \ell \leq \lmax$. By an application of the approximate triangle inequality  (\Cref{item: tri-2} of \Cref{fact: triangle}) and the fact that $\cost(p, a_j) \ll \cost(a_j, S)$,
\begin{align}\label{eq: reuse-2}
|\cost(p,S) - \cost(a_j,S)| \leq 2 \sqrt{ \cost(p,a_j)\cost(a_j,S)} + \cost(p,a_j) 
\lesssim  \sqrt{\cost(p,a_j)\cost(a_j,S)}.
\end{align}
Since $p$ is in $R_j(\ell)$ and $C_j$ is far from $S$ we obtain $\cost(p,a_j) \leq 2^{\ell+1} \Delta_j \leq 2^{\ell+1} \eps^{2} \cost(a_j,S)$. Plugging this into  \Cref{eq: reuse-2} gives 
\begin{align}
\label{eq: reuse-3}
|\cost(p,S) - \cost(a_j,S)| \lesssim 2^{\ell/2} \eps \cost(a_j,S).
\end{align}
Using this inequality and the fact that $|R_j(\ell)|$ is at most $|C_j|/2^{\ell}$ (by Markov's inequality), we get
\begin{align*}
    \sum_{\ell = 0}^{\lmax}\sum_{p \in R_j(\ell)} |\cost(p,S) - \cost(a_j,S)| \lesssim \sum_{\ell = 0}^{\lmax} \frac{|C_j|}{2^\ell} \cdot 2^{\ell/2} \eps \cost(a_j,S)
    \lesssim \eps|C_j| \cost(a_j,S). 
\end{align*}

Combined with the bound on $R_j(\lmax+1)$, we get:
\begin{align}
\notag
    \text{Term 1} &= |\cost(C_j,S) - |C_j| \cdot \cost(a_j,S)| \lesssim \eps \cost(C_j,S) + \eps|C_j| \cost(a_j,S)\\ 
    \label{eq:term1}
    &\leq \eps \cost(C_j, S) + \eps (\cost(C_j, S) + |\cost(C_j,S) - |C_j| \cdot \cost(a_j,S)|).
\end{align}
Rearranging, this implies $\text{Term 1} \lesssim \eps \,\cost(C_j, S)$, as desired.
\paragraph{Bounding term 2.} This proof will be very similar to the previous one, relying on properties for event $\calE$ to show the weight of each cluster is preserved in $\Omega$.

  For $\ell$ in range $0 \leq \ell \leq \ell_{max} +1$, let $W_\ell = \sum_{q \in \Omega \cap R_j(\ell)} w_q$ denote the sum of weights of coreset points in the $R_j(\ell)$ and  let $W = \sum_{\ell = 0}^{\ell_{max}+1} W_{\ell}$ be the total weight of coreset points from $C_j$. Since $ \costom(C_j,S) = \sum_{\ell= 0}^{\ell_{max}+1} \costom(R_j(\ell), S)$ we have,
  \begin{align}
\label{eq:frt4}
     |\costom(C_j,S) - |C_j| \cost(a_j,S) | \leq  |\costom(C_j,S) - W \cost(a_j,S) | + |(W - |C_j|)\cdot \cost(a_j,S) |
\end{align} 
If event $\mathcal E$ holds, it follows that $|W - |C_j|| \leq \eps |C_j|$ and by the same reasoning as \Cref{eq:term1} that the second term of \Cref{eq:frt4} is $O(\eps\, \cost(C_j, S) + \eps \text{(Term 1)})$. It now remains to bound the first term. 
\[  |\costom(C_j,S) - W \cost(a_j,S)| \leq \sum_{\ell = 0}^{\ell_{max}+1} |\costom(R_j(\ell),S) -W_\ell \cost(a_j,S)|. \]
For $0 \leq \ell \leq \ell_{max}$, the event $\mathcal{E}$ and Markov's inequality ensure that $W_\ell \lesssim |C_j|/2^\ell$. Therefore, a similar calculation as in \eqref{eq: reuse-3} shows that 

\begin{flalign}\notag
&|\costom(R_j(\ell),S) -W_\ell\, \cost(a_j,S)| = O(W_\ell \cdot \eps 2^{\ell/2}\, \cost(a_j,S))\\
\label{eq: fake2}
& = O(|C_j|/2^\ell \cdot \eps 2^{\ell/2} \,\cost(a_j,S))= O(\eps 2^{-\ell/2} |C_j|\, \cost(a_j,S)) = O(\eps 2^{-\ell/2}\, \cost(C_j,S)).
\end{flalign}
Finally, we also get, for $\ell = \ell_{max}+1$
\begin{flalign}\label{eq: fake3}
\begin{split}
&|\costom(R_j(\ell_{max}+1),S) -W_{\ell_{max}+1}
\cdot \cost(a_j,S)| \leq \sum_{q \in \Omega \cap R_j(\ell_{max}+1)} |w_q \, \cost(q,S) - w_q\, \cost(a_j,S)|  \\& \leq \sum_{q \in \Omega \cap R_j(\ell_{max}+1)} w_q \, \cost(q,S) + w_q\,  \cost(a_j,S) \overset{(i)}{\leq} \sum_{q \in \Omega \cap R_j(\ell_{max}+1)}  (2w_q\,\cost(q,a_j) + 3w_q\,\cost(a_j,S))\\
&\overset{(ii)}{\leq} 2\costom(C_j,A) + 6\eps |C_j|\, \cost(a_j,S)
\overset{(iii)}{\leq} O(\cost(C_j,A)) + O(\eps\,\cost(C_j,S)) = O(\eps \,\cost(C_j,S)),
\end{split}
\end{flalign}
where step $(i)$ uses \Cref{fact: triangle}, step $(ii)$ uses the bound ${W_{\ell_{max}+1} \leq 2\eps|C_j|}$ which holds whenever event $\mathcal{E}$ occurs.
Step $(iii)$ also follows by the third property of event $\mathcal{E}$.

Summing  \eqref{eq: fake2} and \eqref{eq: fake3} gives that Term 2 $\lesssim \eps\, \cost(C_j, S)$, which concludes the proof.
\end{proof}

\subsection{Handling Clusters with Low Cost}\label{app: tiny-cluster}
\begin{proof}[Proof of \Cref{lem: tiny-cluster}]
Let us denote $\psi_T(\Omega, S) = \left\lvert \sum_{j \in J_S}(\cost(C_j, S) - \costom(C_j,S))\right\rvert$.  
Our goal is to show that
\[\E_{\Omega} \sup_{S} \psi_T(\Omega, S) \lesssim \eps \cost(P,A).\]
This will follow directly from the following simple bounds. 
\begin{claim}\label{clm: useful-tiny}
    Fix a set $\Omega$ and a set of centers $S$. For any cluster $C_j$ with $j \in J_S$, it always holds that (i)   $\cost(C_j, S) \lesssim \tfrac{\eps}{k} \cost(P,A)$  and (ii)  $\costom(C_j, S) \lesssim  \cost(P,A)$.
  
    Moreover if $\Omega$ satisfies event $\mathcal{E}$, we have the stronger bound (iii)  $\costom(C_j, S) \lesssim \tfrac{\eps}{k} \cost(P,A)$.
\end{claim}

Let us first see how this claim directly implies the result.

For any $S$, we trivially have
\[\psi_T(\Omega, S) \leq \sum_{j \in J_S}(\cost(C_j, S) + \costom(C_j, S)). \]

If $\Omega$ satisfies the event $\mathcal{E}$, then plugging the bounds (i) and (iii) in Claim \ref{clm: useful-tiny} for each $j\in J_S$, we get that $\psi_T(\Omega, S)\lesssim \eps \cost(P,A)$.

If $\Omega$ does not satisfy $\mathcal{E}$, the worst case bound (ii) in Claim \ref{clm: useful-tiny} gives $\costom(C_j, S) \lesssim \cost(P,A)$ and hence $\psi_T(\Omega, S)\lesssim k \cost(P,A)$.

Using these bounds on $\psi_T(\Omega, S)$, and as $\Pr[\mathcal{E}] \geq (1 - \eps/k)$  by \Cref{lem:E}, the Law of total Expectation gives that 
\begin{align*}
\E_{\Omega} \sup_{S} \psi_T(\Omega, S) &= \E_{\Omega} \sup_{S} [\psi_T(\Omega, S)| \mathcal{E}] \cdot \Pr[\mathcal{E}] + \E_{\Omega} \sup_{S}[\psi_T(\Omega, S)| \overline{\mathcal{E}}] \cdot \Pr[\overline{\mathcal{E}}] 
\\
&\leq \eps \cost(P,A) \cdot 1 + k\cost(P,A) \cdot \tfrac{\eps}{k} \lesssim \eps \cost(P,A),
\end{align*}
as desired.

We now prove \Cref{clm: useful-tiny}.
\begin{proof}[Proof of \Cref{clm: useful-tiny}]
For any point $p \in C_j$, using the triangle inequality (\Cref{fact: triangle}) and that $C_j$ is close to $S$ (so that $\cost(a_j,S)\leq \Delta_j \eps^{-2})$ gives 
\begin{align}\label{eq: tiny3}
\cost(p,S) &\leq 2(\cost(p,a_j) + \cost(a_j,S))
\leq 2(\cost(p,a_j) + \Delta_j \eps^{-2}).
\end{align}
Summing this over points $p$ in $C_j$ and using that $\cost(C_j, A) \leq T = \tfrac{\eps^3}{k} \cost(P,A)$, we get
\begin{align*}\label{eq: tiny4}
    \cost(C_j,S) \leq 2 (\cost(C_j,A) + \cost(C_j,A) \eps^{-2}) \leq 4 \eps^{-2}\cost(C_j,A) \lesssim \tfrac{\eps}{k} \cost(P,A). 
\end{align*}
This proves the bound (i). 

To prove the bound (ii), summing \Cref{eq: tiny3} over points in $C_j \cap \Omega$ gives,
\begin{align*}
    \costom(C_j, S) &\leq 2 \sum_{p \in C_j \cap \Omega} (w_p \cost(p,a_j)) + 2\Delta_j \eps^{-2} \sum_{p \in C_j \cap \Omega} w_p\\
    &\overset{(i)}{\leq} 2 \sum_{p \in C_j \cap \Omega} (\tfrac{3\cost(P,A)}{m\cdot \cost(p,a_j)} \cdot \cost(p,a_j)) + 2\Delta_j \eps^{-2} \sum_{p \in C_j \cap \Omega} \tfrac{4k|C_j|}{m} \tag{By \Cref{lem: weight-bound}}\\
    &\leq 6 \cost(P,A) + 8k\eps^{-2} \cost(C_j,A) \lesssim \cost(P,A). \tag{As $\cost(C_j, A) \leq \tfrac{\eps^3}{k} \cost(P,A)$}
\end{align*}
\allowdisplaybreaks
If the event $\mathcal{E}$ holds, then step $(i)$ above can be tightened using the inequalities: $\sum_{p \in C_j \cap \Omega} w_p \cost(p,a_j) = \costom(C_j, A) \leq (1+\eps) \cost(C_j, A)$ which holds by property $P_3$ (see \Cref{def: event-E} of event $\mathcal{E}$) and $\sum_{p \in C_j \cap \Omega} w_p \leq (1+\eps) |C_j|$ which holds by property $P_1$. 
This gives, 
\begin{align*}
    \costom(C_j, S) &\leq 2(1+\eps) \cost(C_j,A) + 2\Delta_j \eps^{-2} (1+\eps) |C_j|\\
    & \lesssim \eps^{-2} \cost(C_j,A) 
    \lesssim \frac{\eps}{k} \cost(P,A).\qedhere
\end{align*}
\end{proof}
\end{proof}

\subsection{Symmetrization}
A random process is a collection of random variables $(X(t))_{t \in T}$ on the same probability space and is indexed by elements $t$ of some set $T$. Consider $m$ random processes $X_1(t), \ldots X_m(t)$ with the same index set $T$ and defined on the same probability space. We say these random processes are independent if for any $t' \in T$, the random variables $X_1(t'), \ldots, X_m(t')$ are mutually independent. We state the following useful symmetrization lemma for independent random processes.

\begin{fact}[Symmetrization for Random Processes] \label{lem: symmetrization} Let $X_1(t), \ldots, X_m(t)$ be independent random processes indexed by $t \in T $. Let $g_1, \dots, g_m \sim \mathcal{N}(0,1)$ be independent Gaussians. Then
	\begin{align*}
		\E \sup_{t \in T} \left\lvert\sum_{i = 1}^n (X_i(t) - \E X_i(t))\right\rvert \leq \sqrt{2 \pi} \E  \sup_{t \in T} \left\lvert\sum_{i = 1}^n g_i X_i(t) \right\rvert.
	\end{align*}
\end{fact}
\noindent We now use the above result to prove \Cref{eq:sym}. 
\begin{lemma}\label{lem: sym-apply}
	Let $g_1, g_2, \cdots, g_\csize$ be independent Gaussians sampled from $\mathcal{N}(0,1)$ where $\csize$ is the number of samples in $\Omega$, then
	\begin{align*}
		\E_{ \Omega} \sup_{S\inSR} \left\lvert \frac{ \cost(B(S), S)- \costom(B(S),S)}{\cost(P,S) }\right\rvert \leq \sqrt{2\pi} \E_{ \Omega} \E_{g} \sup_{S\inSR} \left\lvert X^{S}(\Omega,g)\right\rvert.
	\end{align*}
\end{lemma}

\begin{proof}
Consider a fixed $S \inSR$. Recall that $u^S(\Omega)$ is a vector in $\R^m$ whose $i$-th entry given by $u_i^S(\Omega) = \mathbbm{1}[q_i \in B(S)] \cost(q_i,S)$. 
 For each $i \in [\csize]$, define the random variable $Y^S_{i}(\Omega) = (w_{q_i} \uq{S}{i}(\Omega))/\cost(P,S)$ where the randomness of $Y^S_i(\Omega)$ is in the choice of the $i$-th sample $q_i$ of $\Omega$.  Firstly, we have $ \sum_{i \in [m]} Y_i^S(\Omega) = \costom(B(S), S)/\cost(P,S) $. Moreover, we have,

\begin{align*}
    \E_{\Omega} Y^{S}_{i} (\Omega) &\overset{(i)}{=} \frac{1}{\cost(P,S)}\sum_{p \in P} \mu(p) \cdot 1/(m \mu(p)) \cdot \mathbbm{1}[p \in B(S)] \cdot \cost(p,S) \\
    &= \frac{1}{\cost(P,S)} \sum_{p \in B(S)} 1/\csize \cdot \cost(p,S) = \frac{\cost(B(S),S)}{\csize \cdot \cost(P,S)}
\end{align*}
 where step $(i)$ uses that \Cref{alg: sensitivity-sampling} assigns weight $1/(m \mu(p))$ if the $i$-th sample is $p$. By the linearity of expectation, $\E_{\Omega} \sum_{i \in [\csize]} Y^S_{i}(\Omega) = \cost(B(S),S)/\cost(P,S)$. Since the coordinates of $u^S(\Omega)$ are mutually independent, the random processes $Y^S_1(\Omega),\ldots, Y^S_m(\Omega)$ indexed by $S \in \allcenters(r)$ are also independent. Applying \Cref{lem: symmetrization} gives,
\begin{align*}
		&\E_{ \Omega} \sup_{S\inSR} \left\lvert \frac{\cost(B(S),S)- \costom(B(S),S))}{\cost(P,S) }\right\rvert \\
        &= \E_{ \Omega} \sup_{S\inSR} \left\lvert \sum_{i \in [\csize]} \Bigl(Y^S_i(\Omega) - \E_{\Omega} Y^{S}_i(\Omega)\Bigr) \right\rvert \leq \sqrt{2\pi} \E_{ \Omega} \E_{g}\sup_{S\inSR} 
        \left\lvert \sum_{i \in [\csize]} g_i Y^S_i(\Omega)\right\rvert.
\end{align*}
which concludes the proof since $X^S(\Omega, g) = \sum_{i \in [m]} g_i Y^{S}_i(\Omega)$.
\end{proof}

\subsection{\texorpdfstring{Worst Case Bound on Gaussian Process: Proof of \Cref{lem: worst-case-omega}}{Worst Case Bound on Gaussian Process}}\label{app: proof-worst-case}

\begin{proof}
    Define the vector $z^S \in \R^m$ such that for $q \in \Omega$, the $q$-th coordinate $z_q^S = (w_q\uq{S}{q}) /\cost(P,S)$. This allows us to write $X^{S} = \langle g, z^S \rangle \leq \norm{g}\cdot {\norm{z^S}}$ where the last step uses the Cauchy-Schwarz inequality.

    \noindent We begin by bounding the norm of $z^S$. 
    For a point $q \in C_j \cap B(S)$, we have $\cost(q, S) \lesssim \cost(q, A) + \cost(a_j, S) \lesssim \cost(q, A) + \Delta_j / \eps^2$. 
    Using \cref{lem: weight-bound} $w_q \leq \min \left(\frac{\cost(P, A)}{m \cost(q, A)}, \frac{\cost(P, A)}{m\Delta_j}\right)$, we therefore get $\cost(q, S) \lesssim \frac{\cost(P, A)}{\eps^2 m}$, and so:
    \begin{align}\label{eq: zS-bound}
        \norm{z^S}^2 = \sum_{q \in B(S)\cap \Omega} \frac{(w_q \cost(q,S))^2}{\cost(P,S)^2} \lesssim \sum_{q \in B(S)\cap \Omega} \frac{(\cost(P,A)^2/(\eps^{4}m^2))}{\cost(P,S)^2} \lesssim \frac{1}{\eps^4m},
    \end{align}
    where the last step uses $\cost(P, A) \leq \gamma \cost(P, S) \lesssim \cost(P, S)$.

    Therefore using \Cref{eq: zS-bound} and the fact that $\E_{g} \norm{g} \lesssim \sqrt{m}$ we get,
    \begin{align*}
    E_g \sup_{S\inSR} |X^S| \lesssim       E_g \sup_{S\inSR} \norm{g} \cdot \norm{z^S} \lesssim E_g \sup_{S\inSR} \norm{g} \cdot 1/(\eps^2 \sqrt{m} ) \lesssim  1/\eps^2
    \end{align*}
\end{proof}

\section{Existence of Cost Vector Nets of Small Size} \label{app: nets}
Consider a fixed signature $(N_1, \ldots, N_k)$ (as defined in \Cref{sec:further-classifying}) and let $N = \sum_{i= 1}^k N_j$. Also, let  $\mathcal{T}$ be the set of all centers with this signature. We show  the following:
\begin{lemma}\label{lem: net-lemma-signature}
    There is an $(\alpha, \mathcal{T})$-net $M_\alpha$ of size $\exp(O(\min(k2^t \alpha^{-2}, N + k \alpha^{-2}) \cdot \log(k \eps^{-1} \alpha^{-1})))$. 
\end{lemma}
\noindent 
 The interaction number $N(S)$ for any $S \in \mathcal{S}(r)$ is, by definition, in range $[2^r, 2^{r+1})$. Moreover, the total number of signatures is at most $(k+1)^k=\exp(O(k \log k))$. These two facts together with \Cref{lem: net-lemma-signature} imply the net bound guaranteed by \Cref{lem: net-lemma}.
\begin{definition}[Single Center Nets]\label{def: single-center-net}
    Let $\alpha \in (0,1/2]$, $j \in [k]$ and $\mathcal{T}$ be the set of centers with signature $(N_1, \ldots, N_k)$. An $(\alpha, \mathcal{T},j)$-net is a set of vectors $M_{\alpha,j}$ which for each $S \in \mathcal{T}$ contains a vector $y$ that approximates the costs of points in $\Omega$ to the $j$-th center $x_j$ of $S$ as follows. For each $q_i \in B(S) \cap \Omega$, 
	\begin{enumerate}[(a)]
            \item \label{item: scn-item1} $ y_i \geq \cost(q_i,S) - \alpha \cdot \err(q_i,S)$ where $\err(\cdot, \cdot)$ is as defined in \Cref{def: clustering-nets}.
		\item  If $\cost(q_i,S) = \cost(q_i, x_j)$ then we further have $y_{i} \leq \cost(q_i,S) + \alpha \cdot \err(q_i,S)$ 
	\end{enumerate}
\end{definition}
In words, the net $M_{\alpha,j}$ contains vectors approximating the cost of points in $\Omega$ with respect to out of the $k$ centers in $S$ (specifically the $j$-th center $x_j$). Moreover, if a point $q \in \Omega$, has $x_j$ as the nearest center among all centers in $S$, then its cost is $\alpha$ approximated. For other points, we only guarantee that their cost wrt $x_j$ is not underestimated by the vector.

We give the following bound on the size of single center nets.
\begin{lemma}[Single Center Net Bound]\label{lem: single-center-net}
	For any $\alpha \in (0,1/2]$ and $j \in [k]$ there exists a $(\alpha, \mathcal{T},j)$ single center net of size $\exp(O(\min(2^t  \alpha^{-2}, N_j + \alpha^{-2}) \cdot \log(k \eps^{-1} \alpha^{-1})))$.
\end{lemma}

\noindent Before giving a proof of \Cref{lem: single-center-net}, we show that constructing single center nets suffices. Specifically, the above bounds on the sizes of single center nets imply the existence of cost vector nets of the size guaranteed by \Cref{lem: net-lemma-signature}.
\begin{proof}[Proof of  \Cref{lem: net-lemma-signature}]
	
	For each $i \in [k]$, let $M_{\alpha,i}$ be the set of vectors given by \Cref{lem: single-center-net}. We define $$M'_{\alpha} = \{\min(v^{(1)}, \ldots, v^{(k)}) | (v^{(1)}, \ldots, v^{(k)}) \in M_{\alpha,1} \times \cdots \times M_{\alpha,k}\}$$
where $\min$ denotes the coordinate-wise minimum of vectors.

Next, we will ensure that the only non-zero coordinates correspond to points from $B(S)$ (clusters with band $b$ and type $t$). We achieve this  by ``guessing'' which clusters are in $B(S)$ and zero out the coordinates of points in the remaining clusters. Formally, for each vector $v$ in $M'_\alpha$ and each subset $\mathcal{C}_0$ of clusters $\{C_1, \ldots, C_k\}$, define $v(\mathcal{C}_0) \in \R^m$ as follows: for $q_i \in \Omega$ if $q_i$ lies in a cluster in $\mathcal{C}_0$, set $(v(\mathcal{C}_0))_i = v_i$; otherwise set $(v(\mathcal{C}_0))_i = 0$. Let $M_{\alpha}$ be the set of all vectors $v(\mathcal{C}_0)$ obtained by varying $v$ over points in $M'_\alpha$ and $\mathcal{C}_0$ over all subsets of $\{C_1, \ldots, C_k\}$. It follows that:
\begin{align*}
    |M_{\alpha}| = 2^k \cdot |M'_\alpha| = 2^k \cdot  \Pi_{i \in [k]} |M_{\alpha,i}| = \exp(O(\min(k2^t \alpha^{-2}, N + k \alpha^{-2}) \cdot \log(k \eps^{-1} \alpha^{-1}))).
\end{align*}
We conclude by proving that $M_{\alpha}$ satisfies the properties required by \Cref{def: clustering-nets}. Consider a set $S =\{x_1, \ldots, x_k\}$ of centers from $\mathcal{T}$.  Let $v^{(j)}$ be the vector in $M_{\alpha,j}$ corresponding to $S$ and define $u_{min} =\min(v^{(1)}, \ldots, v^{(k)})$. Finally, let $u$ be the vector obtained by setting the $i$-th coordinate of $u_{min}$ to zero if $q_i$ is not in $B(S)$. By the construction above, $u$ is contained in $M_{\alpha}$.  We claim that $u$ approximates the cost of points in $\Omega$ as required. First, if $q_i \notin B(S)$ then $u_i = 0$. On the other hand, if $q_i \in B(S)$, then by \Cref{lem: single-center-net}, $u_i \geq  \cost(q_i,S) - \alpha \cdot \err(q_i,S)$. Moreover if $x_j \in S$ is the nearest center to $q_i$, then $u_i = \min_{j'} v^{(j')}_i \leq v^{(j)}_i \leq \cost(q_i,S)+\alpha \cdot \err(q_i,S)$ completing the proof of the lemma.
\end{proof}

Our main goal thus reduces to proving \Cref{lem: single-center-net}. 

\subsection{A Warm-up}
 As a first step towards proving \Cref{lem: single-center-net}, we give a prove a  warm-up lemma which shows the existence of single center nets whose size depends exponentially on the dimension $d$. In the next section, we will show that we can use dimensionality reduction ideas to avoid such a dependence on $d$. We begin by recalling some facts about the net size of a Euclidean ball in $d$-dimensions.

 Let $B_2^d(R)$ denote the $d$-dimensional Euclidean $\ell_2$-ball of radius $R$. A set $D_\eta \subset \R^d$ is called an $\eta$-net  of $B_2^d(R)$ if for any point $x \in B_2^d(R)$ there exists $y \in D_\eta$ such that $\norm{x - y}_2 \leq \eta$. The following is a standard result that bounds the size of the $\eta$-net of $B_2^d(R)$.

\begin{fact}[Lemma 5.2 of \cite{vershynin2010introduction}]\label{lem: eps-net}
	The ball $B_2^d(R)$ has an $\eta$-net of size at most $(2R/\eta + 1)^d.$
\end{fact}

We use the above fact, to prove the following useful lemma. 
\begin{lemma}[Warm Up]\label{lem: warm-up}
    Let $\Omega = \{q_1, \ldots, q_m\} \subset \R^d$ and suppose that each point $q$ in $\Omega$ has an associated threshold distance $R_q$. For $\alpha \in (0,1/2]$, there is a set $N_\alpha \subset \R^m$ of size at most $m (3/\alpha)^d$ with the following property. For any point $x \in \R^d$ there is a vector $y \in N_\alpha$ that approximates costs of the points $q_i$ with respect to $x$ as follows: for $i \in [m]$,  
	\begin{enumerate}[(a)]
		\item \label{item: net-cond-1} If $\cost(q_i, x) \leq R_{q_i}^2$ then 
		$| y_{i} - \cost(q_i, x)| \leq \alpha R_{q_i}^2$
            \item \label{item: net-cond-2}If $\cost(q_i, x) > R_{q_i}^2$ then $y_i \geq R_{q_i}^2 - \alpha R_{q_i}^2$.
	\end{enumerate}
\end{lemma}
\begin{proof}
For each $q \in \Omega$, consider the $d$-dimensional ball of radius $R_q$ centered at this point. Let $D(q)$ be an $(\alpha R_q)$-net of this ball (as defined in \Cref{lem: eps-net}).  By \Cref{lem: eps-net}, we have that $|D(q)| \leq (3/\alpha)^d$ and that the set $D(\Omega) = \bigcup_{q' \in \Omega} D(q')$  has at most $m (3/\alpha)^d$ points.

For $x' \in D(\Omega)$, define the cost vector $y(x') \in \R^m$ as follows: for $i \in [m]$ set $y_i(x') = \cost(q_i, x')$. Let $N_\alpha$ be the set of all such vectors. Note that $N_\alpha$ is of the claimed size.
 
Let $x$ be a point in $\R^d$ and $x'$ be the point in  $D(\Omega)$ nearest to $x$. We show that $y(x') \in N_\alpha$ satisfies both \Cref{item: net-cond-1} and \Cref{item: net-cond-2} for the point $x$. 

\textbf{Proof of \Cref{item: net-cond-1}.} If $q_i \in \Omega$ is a point satisfying $\cost(q_i,x) \leq R_{q_i}^2$ then we have by our choice of $x'$ that $\norm{x - x'} \leq \alpha R_{q_i}$. By triangle inequality (\Cref{fact: triangle}), 
\begin{align*}
    |\cost(q_i,x) - y_i(x')| &= |\cost(q_i,x) - \cost(q_i,x')| 
            \leq 2\sqrt{\cost(q_i,x) \cost(x,x')} + \cost(x,x') \\
            &\leq 2 \sqrt{R_{q_i}^2 \cdot \alpha^2 R_{q_i}^2} + \alpha^2 R_{q_i}^2 \leq 3\alpha R_{q_i}^2.
\end{align*}

\textbf{Proof of \Cref{item: net-cond-2}.} Suppose that $q_i \in \Omega$ is a point with $\cost(q_i,x)> R_{q_i}^2$. We now want to show that $y_i(x')$ does not underestimate the cost of $q_i$ by too much.   Consider the ball of radius $R_{q_i}$ centered at $q_i$ and let $p$ be the point in this ball nearest to $x$. Clearly, $D(\Omega)$ contains a point within distance $(\alpha R_{q_i})$ of $p$. Therefore, we have $\norm{x-x'} \leq \norm{p-x} + \alpha R_{q_i}$. By triangle inequality,
\begin{align}
                \norm{{q_i}-x'} &\,\,\geq \,\, \norm{{q_i}-x} - \norm{x - x'} \,\,\geq\,\,\norm{{q_i} - x} - \norm{p-x} - \alpha R_{q_i} \,\,= \,\,R_{q_i} - \alpha R_{q_i} 
\end{align}
where the last line uses that ${q_i},p,x$ are collinear. Therefore, $y_i(x') = \cost({q_i},x') \geq (R_{q_i} - \alpha R_{q_i})^2 \geq R_{q_i}^2 - 2\alpha R_{q_i}^2$. The result follows by rescaling $\alpha$.
\end{proof}

We now show that the above lemma implies a bound on the size of  $(\alpha, \mathcal{T}, j)$-single center nets. The idea is that since $q \in B(S)$ is a close point, by the triangle inequality, it satisfies $\cost(q,S) \leq 2(\cost(q,A) + \Delta_q \eps^{-2})$. Therefore, if $x_j$ is the nearest center to $q$ then it is within this cost radius of $q$.   Therefore, plugging $R_q^2 = 2(\cost(q,A) + \Delta_q \eps^{-2})$ in the above lemma and rescaling $\alpha$ by an $O(\epsilon^{-2})$ factor, yields a single center of $\poly(k/\eps) \cdot (1/\alpha \eps)^d$. 

The size of the net constructed above is not ideal (we want net sizes independent of $d$). To eliminate this dependence on the dimension, we show that we can use the above lemma after first projecting the points in $\Omega$ onto a suitably chosen low-dimensional subspace. This then yields the final bound on the net size. 

\subsection{Dimensionality Reduction Lemma}
We now present the dimensionality reduction lemma that we shall use in our net construction. 

The lemma below shows that for any set of points $\Omega$ and a point $x$, there is a low dimensional subspace $U(x)$ with the following two properties: (i) $U(x)$ is spanned by a subset of points in $\Omega$, (ii) If $\Pi$ is the orthogonal projection matrix of $U(x)$, 
 then $\langle q, x \rangle \approx \langle \Pi q, x \rangle$. Furthermore, one can augment any set $U_0 \subset \R^d$ with a small subset of points from $\Omega$ to obtain a subspace with such a property.

\begin{lemma} \label{lem: subspace}
	Let $\Omega \subset \R^d$ be an arbitrary set of points, $x$ be a point in $\R^d$ and $\alpha \in (0,1)$. Any set $U_0 \subset \R^d$ can be extended to obtain a set $U(x) = U_0 \cup R(x)$ such that the following properties hold:
	\begin{enumerate}[]
		\item The set $R(x)$ is a subset of $\Omega$ with cardinality $|R(x)| = O(1/\alpha^2)$.
		\item $R(x)$ contains the point in $\Omega$ nearest to $x$.
		\item Let $\Pi$ denote the orthogonal projection on to the span of $U(x)$.  For any point $q \in \Omega$ we have,
		\begin{align}\label{eq: progress-cond}|\langle q, x \rangle - \langle \Pi q, x \rangle | \leq \alpha  \cdot \norm{q - \Pi q}\cdot \min_{q' \in \Omega} \norm{q' - x}.\end{align}
	\end{enumerate}
    We call the set of points $U(x)$ with the above properties an $\alpha$-good set for $x$.
\end{lemma}
\begin{proof}
    If $|\Omega| = O(1/\alpha^2)$ then we are trivially done by taking $R(x) = \Omega$. If this is not the case, we will construct $R(x)$ iteratively in several rounds. Let $R_i$ denote the set $R(x)$ in the $i$-th round and $\Pi_i$ denote the orthogonal projection onto the span of $U_0 \cup R_i$.

    Begin by setting $R_1 = \{q'\}$ where the $q'$ is the point in $\Omega$ nearest to $x$. If in the $i$-th round (where $i \geq 1$) there exists a point $q$ violating \Cref{eq: progress-cond}, then add it to $R_i$ to obtain $R_{i+1}$. If no such $q$ exists, then terminate. We show that the algorithm above terminates in $O(1/\alpha^2)$ rounds. Observe first that by the Pythagorean theorem, 
    \begin{align}\label{lem: proj-lower-bound}
    \norm{\Pi_1 x}^2 = \norm{x}^2 - \norm{(I-\Pi_1)x}^2 \geq \norm{x}^2 - \norm{q' - x}^2
    \end{align}where the last inequality used the fact that $q'$ is in the subspace $R_1$.
    
    Suppose that point $q \in \Omega$ violates condition $3$ in round $i$; we show that augmenting $q$ to the subspace significantly increases the projection of $x$ onto it. We have:
    \begin{align}\label{eq: v-eq}
        |\langle q - \Pi_i q, x \rangle| = |\langle q, x \rangle - \langle \Pi_iq, x \rangle| > \alpha \cdot \norm{q - \Pi_i q} \cdot \norm{q' - x}.
    \end{align}
    For $v =\tfrac{ q - \Pi_i q}{\norm{q - \Pi_i q }}$, rearranging \Cref{eq: v-eq} we get $|\langle v, x \rangle| > \alpha \norm{q' - x}$.  Moreover, we have 
    \begin{align}
    \norm{\Pi_{i+1} x}^2 = \norm{\Pi_i x}^2 + \langle v, x \rangle ^2 > \norm{\Pi_i x}^2 + \alpha^2 \norm{q' - x}^2.
    \end{align}
    Therefore, after each round, the norm of the projection of $x$ increases by $\alpha ^2 \norm{q' - x}^2$. But by \Cref{lem: proj-lower-bound} and the fact that $\norm{x}^2 \geq \norm{\Pi_i x}^2$, this can happen for at most $O(1/\alpha^2)$ rounds. 
\end{proof}

\subsection{\texorpdfstring{Bound $1$: $\exp(\tilde{O}(2^t \alpha^{-2}))$}{Bound 1}}\label{subsec: net-bound-1}

We now combine \Cref{lem: subspace} together with \Cref{lem: warm-up} to construct the single center nets of size guaranteed by \Cref{lem: single-center-net}. 

In this subsection, we prove the first bound on the size of nets, which depends on the type $t$ of the points in $B(S) \cap \Omega$. (we recall that $B(S)$ is defined as the set of points in type $t$ and band $b$; see \Cref{sec:further-classifying} for the precise definitions). In particular, we  construct an $(\alpha, \mathcal{T}, 1)$ net $M_{\alpha,1} \subset \R^m$ of size $\exp(\tilde{O}(2^t \alpha^{-2}))$. We will describe the construction of $M_{\alpha,1}$ by fixing $S \in \mathcal{T}$ with first center $x_1$ and showing how to construct a net vector for $x_1$. It will be clear from the procedure used to construct the net vector that the number of possible vectors obtainable is bounded by $\exp(\tilde{O}(2^t \alpha^{-2}))$.

 We begin by applying \Cref{lem: subspace} with $U_0 = \emptyset$ to obtain an $(\alpha2^{-t/2})$-good subset $U(x_1)$ for $x_1$. Let  $U$ denote the subspace spanned by the points in $U(x_1)$ and $\Pi$ be its orthogonal projection matrix. By the guarantees of the lemma $\dim(U) = O(\alpha^{-2} 2^t)$. The first step is to use the Pythagorean theorem to write the cost of any point $q_i \in \Omega$ to $x_1$ as follows:

\begin{align}\label{eq: pythagoras}
\begin{split}
    \cost(q,x_1)=\underbrace{\norm{\Pi(q-x_1)}^2}_{\text{Term 1}} 
    + \underbrace{\norm{(I-\Pi)q}^2}_{\text{Term 2}} 
    + \underbrace{\norm{(I-\Pi)x_1}^2}_{\text{Term 3}} -\underbrace{2  \big(\langle q, x_1 \rangle - \langle \Pi q, x_1 \rangle\big)  }_{\text{Term 4}}. 
\end{split}
\end{align}
To obtain a single center net, we shall construct small ``nets'' to approximate each of the first three terms and then use \Cref{lem: subspace} to argue that the Term $4$ is small and need not be approximated. We then combine these nets to obtain the single center net.

The following claim gives bounds on Term $4$ and is an immediate consequence of \Cref{lem: subspace}. It states that if $x_1$ is an approximate nearest neighbor of $q$, then term $4$ is bounded by $O(\alpha \err(q,S))$; if not, we have a guarantee which is weaker but sufficient for our analysis.

\begin{claim}\label{clm-term-4-interim}
Term $4$ is at most $O(\alpha 2^{-t/2} \cost(q,x_1))$.
\end{claim}
\begin{proof}
    Since $\Pi$ is an orthogonal projection onto the subspace spanned by an  $(\alpha 2^{-t/2})$-good subset, by \Cref{lem: subspace},  for any $q \in \Omega$ we have, 
\begin{align}\label{eq: inter-guarantee-1}
    |\langle q,x_1 \rangle - \langle q, \Pi x_1\rangle| \leq (\alpha2^{-t/2}) \cdot \norm{q - \Pi q} \cdot \min_{q' \in \Omega} \norm{q' - x_1} \leq (\alpha2^{-t/2}) \cdot \norm{q - \Pi q} \cdot \norm{q - x_1}.
\end{align}

Let $q^{*}$ denote the nearest point to $x_1$  in $\Omega$. Since $q^*$ is guaranteed to lie in the good subset (by \Cref{lem: subspace}), by an application of triangle inequality we have  $\norm{q - \Pi q} \leq \norm{q - q^*} \leq \norm{q - x_1} + \norm{q^{*} - x_1} \leq 2 \norm{q - x_1}$. Plugging this into \Cref{eq: inter-guarantee-1} proves the claim.
\end{proof}
\begin{corollary}[Bound on Term 4]\label{clm: term4-bound}  Let $q$ be a point in $B(S) \cap \Omega$. (i)Term $4$ for point $q$ is at most $\cost(q,x_1)/4$. (ii) If we further have $\cost(q, x_1) \leq 4 \cost(q,S)$ then Term $4$ is $O(\alpha \err(q,S))$.
\end{corollary}
\begin{proof} 
Bound $(i)$ follows from \Cref{clm-term-4-interim} by picking a sufficiently small $\alpha$ and noting that $t \geq 0$. Now we prove bound $(ii)$, which is more interesting. Suppose that $\cost(q, x_1) \leq 4 \cost(q, S)$, i.e., $x_1$ is an approximate nearest center.  Since $q$ is in a type $t$ cluster, by the triangle inequality, we have, 
\begin{align}
\cost(q,S) \leq 2(\cost(q,A) + 2^t \Delta_q) 
    \leq 2^{t+1}(\cost(q,A) + \Delta_q).
\end{align}
This implies that 
\begin{align*}
    2^{-t/2} \cost(q,x_1) &\lesssim 2^{-t/2} \cdot \cost(q,S) \tag{$x_1$ is an approximate nearest center}\\ &\lesssim  2^{-t/2}\sqrt{\cost(q,S) \cdot 2^{t} (\cost(q,A) + \Delta_q) } \tag{Point $q$ is of type $t$}\\
    &\lesssim \sqrt{\cost(q,S) (\cost(q,A)+\Delta_q)}\\
    &\lesssim \err(q,S)
    .
\end{align*}
completing the proof of $(ii)$.
\end{proof}

If $x_1$ is not an approximate nearest center  to $q$ then by the previous lemma term $4$ is at most $\cost(q,x_1)/4$. Therefore, the sum of the Terms $1$-$3$ is at least $3/4 \cdot  \cost(q,x_1) \geq 3 \cost(q,S)$; hence at least one of them is greater than or equal to $\cost(q,S)$. This observation ensures that our net construction satisfies  \Cref{item: scn-item1}. In particular, our nets for terms $1$-$3$ ensure that if the value of the term is larger than $\cost(q,S)$, then the coordinate of the net vector corresponding to $q$ is at least $\cost(q,S) - \alpha \err(q,S)$; thus ensuring that we never underestimate the cost of $q$ by too much.

We construct nets $V_1(U), V_2(U)$ and $V_3(U)$ which approximate terms $1$, $2$ and $3$ respectively. We guarantee that if $x_1$ is an approximate nearest center to $q$, then the net vector's entry approximates $\cost(q,x_1)$ to within error $\alpha \pm \err(q,S)$. The $i$-th net also guarantees that if term $i$ is larger than $\cost(q,S)$, then the corresponding entry of the net vector is at least $\cost(q,S) -\alpha \err(q,S)$.

 The final net $M_{1, \alpha}$ that approximates the costs is obtained by their combination, 
\begin{align} \label{eq: mink-sum}M_{1,\alpha} =  \bigcup\limits_{U} \,\, (V_1(U) \oplus V_2(U) \oplus V_3(U))
\end{align}
where $\oplus$ denotes the Minkowski sum between sets and the union in the above equation is over all subspaces that are spanned by a subset of $\Omega$ of size $O(2^t \alpha^{-2})$. The final net size is then $|M_{\alpha,1}| =  \binom{m}{O(2^t\alpha^{-2})} \cdot |V_1(U)| \cdot |V_2(U)| \cdot |V_3(U)|$.

Below, we show how each $V_i(U)$ is constructed and also give a bound on its cardinality. We begin by constructing a net that approximates Term $1$.

\begin{claim}[Approximating Term 1]\label{clm: term-1}
    There is a set $V_1(U) \subset \R^{m}$ with  $|V_1(U)| = \exp(O( 2^t \alpha^{-2} \cdot \log(k \eps^{-1} \alpha^{-1})))$ with the property that for any $S \in \mathcal{T}$ with first center $x_1$, it contains a vector $y$ satisfying the following. For each $q_i \in B(S) \cap \Omega$:
    \begin{enumerate}
        \item \label{item: term1-a} If $\cost(q_i, x_1) \leq 4 \cost(q_i,S)$ then $\left|y_i - \norm{\Pi(q_i-x_1)}^2 \right| \leq \alpha\cdot  \err(q_i,S).$
        \item \label{item: term1-b} If  $\norm{\Pi(q_i-x_1)}^2 \geq \cost(q_i,S)$ then $y_i \geq  \cost(q_i,S) - \alpha \err(q_i,S).$
    \end{enumerate}
\end{claim}
\begin{proof}
    
    Let $\Omega'$ denote the set of points $\Pi q_i$ for $q_i \in \Omega$. We apply \Cref{lem: warm-up} to $\Omega'$ with $R_{\Pi q_i}^2 = O(\cost(q_i,A) + \Delta_{q_i} \eps^{-2})$ and the ``$\alpha$'' in the lemma set to $O(\alpha \eps^2)$. We pick $V_1(U)$ to be the set of vectors returned by the lemma. Next, we show that this set indeed satisfies our requirements.

    Let $y \in V_1(U)$ be the vector corresponding to the point $x_1$ (as given by \Cref{lem: warm-up}). We show that $y$ has the claimed properties.

    \noindent \textbf{Proof of \Cref{item: term1-a}.} Suppose that $q_i$ is a point such that  $\cost(q_i,x_1) \leq 4 \cost(q_i, S) \leq 8 (\cost(q_i,A) + \Delta_{q_i} \eps^{-2}) = R_{\Pi q_i}^2$. In this case, the squared distance between their projections onto the subspace $U$ is also at most $R_{\Pi q_i}^2$, i.e., $\norm{\Pi {q_i} - \Pi x_i}^2 \leq R_{\Pi q_i}^2$. Hence by \Cref{lem: warm-up} we obtain $|y_i - \norm{\Pi(q_i - x_1)}^2| \leq \alpha \eps^2 R_{\Pi q_i}^2 \lesssim \alpha \err(q_i,S)$. 

    \noindent \textbf{Proof of \Cref{item: term1-b}.} If $R_{\Pi q_i}^2 \geq \norm{\Pi(q_i - x_1)}^2 \geq  \cost(q_i,S)$ then by the first condition of \Cref{lem: warm-up} we have that $y_i \geq \norm{\Pi(q_i - x_1)}^2 - \alpha \eps^2 R_{\Pi q_i}^2 \geq  \cos(q,S) - \alpha \err(q_i, S)$. On the other hand, if $\norm{\Pi(q_i - x_1)}^2 \geq R_{\Pi q_i}^2$, then the second condition of \Cref{lem: warm-up} leads to the same conclusion.

    \noindent \textbf{Bounding the Net Size.} Since $U$ is a  $O(2^t\alpha^{-2})$ dimensional subspace, by \Cref{lem: warm-up} we have $|V_1(U)| = m \exp(O(2^t \alpha^{-2} \log( \eps^{-1} \alpha^{-1}))) = \exp(O( 2^t \alpha^{-2} \log(k \eps^{-1} \alpha^{-1})))$ where we used the fact that $|\Omega| = m= \poly(k/\eps)$ above.\qedhere
\end{proof}

Term $2$ is independent of $x$ and the vector $z \in \R^m$ with $z_i =\norm{(I-\Pi)q_i}^2$ approximates it with zero error. Therefore, picking $V_2(U) = \{z\}$ suffices.

\begin{claim}[Approximating Term 3]\label{clm: term-3}
        There  is a set $V_3(\Pi) \subset \R^m$ of size $|V_3(\Pi)| = O(\alpha^{-1}\eps^{-2}m)$ which for any $S \in \mathcal{T}$ with center $x_1$, contains a vector $\vecy$ satisfying the following. For each $q_i \in B(S) \cap \Omega$:
    \begin{enumerate}
        \item If $\cost(q_i, x_1) \leq 4 \cost(q_i,S)$ then we have $\norm{(I - \Pi)x_1}^2\leq y_i \leq \norm{(I - \Pi)x_1}^2+\alpha \err(q_i,S) $.
        \item If $\norm{(I - \Pi)x_1}^2 \geq \cost(q_i, S)$ then $y_i \geq \cost(q_i,S)$.
    \end{enumerate}
\end{claim}
\begin{proof}
The third term $\norm{(\Pi - I) x_1}^2$ is the squared distance of $x_1$ to the subspace and is independent of the point $q_i$. We show that there exists a $\lambda \in \R$ such that $\norm{(\Pi - I)x_1}^2 \leq \lambda \leq \norm{(\Pi - I)x_1}^2 + \alpha \err(q_i,S)$ for any $q_i$ for which we have $\cost(q_i,x_1) \leq 4 \cost(q_i,S)$. The $m$-dimensional vector, all of whose entries are $\lambda$, then satisfies our requirements.

For each $q_i \in B(S)$ for which $\cost(q_i, x_1) \leq 4 \cost(q_i, S)$, Term $3$ lies in the interval $[0, 4R_{q_i}^2]$ where $R_{q_{i}}^2 = 2(\cost(q_i,A) + \Delta_{q_i} \eps^{-2})$. Therefore one can find $i \in [m]$ and $j \in [ \lceil\alpha^{-1} \eps^{-2} \rceil]$ such that  picking $\lambda= j \alpha \eps^2 R_{q_i}^2$ works.\qedhere

\end{proof}

Multiplying the cardinalities of the nets obtained above, we get that the final net size is $|M_{1,\alpha}| = \binom{m}{O(2^t \alpha^{-2})} \cdot |V_1(U)| \cdot |V_2(U)|
 \cdot |V_3(U)| = \exp(O(2^t \alpha^{-2} \log(k\eps^{-1} \alpha^{-1})))$.

\subsection{\texorpdfstring{Bound $2$: $\exp(\tilde{O}(N_1 + \alpha^{-2}))$}{Bound 2}}
We now obtain the second guarantee of  \Cref{lem: single-center-net}, which upper bounds the single center net size in terms of the number of clusters $N_1$ with which $x_1$ interacts (see \Cref{def: interaction} for a precise definition of interaction). Our goal is to show that the size of the net is at most $\exp(\tilde{O}(N_1 + \alpha^{-2}))$. 
 To obtain nets of this size, we will choose the low dimensional subspace on which $\Omega$ is projected more carefully by using the information about the clusters with which the first center interacts. In particular, we enumerate over all possible choices of $I(x_1)$ which are at most $\binom{k}{N_1} 
= \exp(O(N_1 \log k))$ in number. 

Suppose that $I(x_1) = \{C_1, \ldots, C_{N_1}\}$; then we pick $U_0 = \{a_1, a_2, \ldots, a_{N_1}\}$ to be the centers corresponding to the clusters of these clusters.  Then we apply \Cref{lem: subspace} to $U_0$ and $x_1$ to obtain an $\alpha$-good set of size $O(1/\alpha^2)$ such that the set $U(x_1) = U_0 \cup R(x_1)$ satisfies the guarantees of the lemma. If $U$ denotes the subspace spanned by the points in $U(x_1)$, we now have $\dim(U) = O(N_1 + \alpha^{-2})$. We again use \Cref{eq: pythagoras} to express the cost as a sum of four terms.

The idea will again be to show that if $\cost(q,x_1) \leq 4 \cost(q,S)$, then the cost of the point $q$ is approximated by the net up to an error of $O(\alpha \err(q,S))$. On the other hand, if $\cost(q, x_1) > 4 \cost(q,S)$ then show that the net does not underestimate the cost by more than $\alpha \err(q,S)$. 

The way in which nets are constructed to approximate for Terms $1$-$3$ to obtain Bound $2$ are identical to the previous case and will not be repeated. However, since the subspace $U$ we picked has dimension $O(N_1 +\alpha^{-2})$, the size of net approximating term $1$ is $\exp((N_1 + \alpha^{-2}) \cdot \log(k \epsilon^{-1} \alpha^{-1}))$. The only other main difference in the proofs of Bound $1$ and $2$ is the argument used to show that the magnitude of term $4$ is small. This is what we present next.

\begin{claim}\label{clm:bounding-term4-2}
 Let $q$ be a point in $B(S) \cap \Omega$. (i)Term $4$ for point $q$ is at most $\cost(q,x_1)/4$. (ii) If we further have $\cost(q, x_1) \leq 4 \cost(q,S)$ then Term $4$ is $O(\alpha \err(q,S))$.
\end{claim}
\begin{proof}
    The proof of (i) is identical to the proof of \Cref{clm: term4-bound}. So, we focus on proving the second part of the lemma, which now requires a new argument.

    We consider the cluster $C_i$ containing $q$ and consider two different cases depending on whether $C_i$ and $x_1$ interact.
    
    \textbf{Cluster $C_i$ and center $x_1$ interact. } We use the fact that the subspace contains the center $a_i$ of the cluster $C_i$ to argue that Term $4$ is small.

     By \Cref{lem: subspace} the constructed subspace $U(x_1)$ with projection matrix $\Pi$ satisfies,
\begin{align}\label{eq: inter-guarantee-2}
    |\langle q,x_1 \rangle - \langle q, \Pi x_1\rangle| \leq \alpha \cdot \norm{q - \Pi q} \cdot \min_{q' \in \Omega} \norm{q' - x_1} \leq \alpha \cdot \norm{q - \Pi q} \cdot \norm{q - x_1}.
\end{align}
Since the center of $a_i$ of the cluster $C_i$ is in the subspace, 
\begin{align}\label{eq: net-eq-step-1}\norm{q - \Pi q} \leq \dist(q,A).
\end{align}
Moreover, since $x_1$ is an approximate nearest neighbor center to $a_i$ (by the \Cref{def: interaction} of interaction), i.e. it satisfies $\cost(a_i, x_1) \leq 16\cost(a_i, S)$ , we can also bound $\norm{q - x_1}$: 
\begin{align}\label{eq: net-eq-step-2}
    \norm{q - x_1} &\leq \dist(q, a_i) + \dist(a_i, x_1) \nonumber\\
    &\leq \dist(q, a_i) + 4 \dist(a_i, S)  \nonumber\\
    &\leq \dist(q,a_i) + 4(\dist(q,a_i) + \dist(q,S))\nonumber\\
    &\leq 5\dist(q,a_i) +4\dist(q,S).
\end{align}
Using \Cref{eq: net-eq-step-1}, \Cref{eq: net-eq-step-2}  we can upper bound the RHS of \Cref{eq: inter-guarantee-2} by  \begin{align*}
    O(\alpha \cdot  \dist(q,a_i) \cdot (\dist(q,a_i) + \dist(q,S)))  = O(\alpha \err(q,S)).
\end{align*}

\textbf{Cluster $C_i$ and center $x_1$ do not interact. }

Consider the center $a_i$ of the cluster $C_i$ containing $q$. If  $\cost(a_i,x_1) \leq 100(\cost(q, a_i) + \Delta_q)$ then by the triangle inequality we also have $\cost(q, x_1) = O(\cost(q,a_i) + \Delta_q)$ and hence Term $4$ is $O(\alpha( \cost(q,a_i) + \Delta_q))$ by part $(i)$ of the lemma. We are done in this case as $
\cost(q,a_i) + \Delta_q \leq \err(q,S)$.

Therefore, suppose otherwise that $\cost(a_i, x_1) > 100 (\cost(q,a_i) + \Delta_q)$. We show that in this case, we must have $\cost(q,x_1) > 4 \cost(q, S)$, thus completing the proof of the lemma. 

Since $C_i$ and $x_1$ do not interact and $\cost(a_i, x_1) > 100(\cost(q,a_i) + \Delta_q) \geq 100 \Delta_q$, it must be the case that $\cost(a_i, x_1) \geq 16 \cost(a_i,S)$, i.e.,  $\dist(a_i, x_1) \geq 4 \dist(a_i, S)$.  We now use this to show that $\dist(q, x_1) > 2 \dist(q, S)$. 
\begin{align*}
    \dist(q,x_1) &\geq \dist(a_i, x_1) - \dist(a_i, q)\\
    & > \dist(a_i, x_1)/2 + 4\dist(a_i, q) \tag{$\dist(a_i, x_1) > 10 \dist(a_i, q)$} \\ 
    &\geq 2 \dist(a_i, S) +4 \dist(a_i, q)\\
    & \geq 2 (\dist(q, S) - \dist(a_i, q)) + 4 \dist(a_i, q)\\
    & \geq 2 \dist(q,S).
\end{align*}
Therefore, $\cost(q,x_1) > 4 \cost(q,S)$ and we are done.
\end{proof}

\paragraph{Bounding the Net Size. } The final net bound is $\binom{k}{N_1} \cdot \binom{m}{O(\alpha^{-2})} \cdot |V_1(U)| \cdot |V_2(U)| \cdot |V_3(U)|= \exp(O((N_1 + \alpha^{-2}) \log(k \eps^{-1} \alpha^{-1})))$. The first factor in the product above comes from having to enumerate over all possible subsets of clusters with which $x_1$ could interact,  the second factor comes from enumerating over all possible $\alpha$-good sets for $x_1$ (which are determined by a subset of $\Omega$ of size $O(\alpha^{-2}))$, the bound on the size of $V_1(U)$ is from the above discussion and the bounds on $|V_2(U)|$ and $|V_3(U)|$ are from \Cref{subsec: net-bound-1}.

\section{Missing Proofs from Chaining Analysis}
\subsection{\texorpdfstring{Bounding the supremum of the Gaussian process $X^{S,\init}$}{Bounding the Supremum of X{S,Init}}}
\label{sec:lemXinit}
In this section, we prove  \Cref{lem:Xinit}. We wish to upper bound $\E_g \sup_{S\inSR} |X^{S,\init}|$, where  recall that $X^{S,\init}$ is defined (in \eqref{eq: x-init-fin}) as,
\[
    X^{S,\init} = \frac{\sum_{i \in [m]} g_i w_{q_i} \uq{S,0}{i}}{\cost(P,S) }.
\]
We begin by bounding the numerator and denominator of $|X^{S,\init}|$ separately. The numerator is upper-bounded as follows:
\begin{align*}
    \left\lvert\sum_{i \in [m]}g_i w_{q_i} \uq{S,0}{i}\right\rvert
    &\lesssim \,\, \sum_{j \in [k]} \left\lvert\sum_{{q_i} \in C_j \cap \Omega}g_i w_{q_i} u^{S,0}_i\right\rvert \tag{Triangle inequality}\\
    &\lesssim \,\,\sum_{C_j \in B(S)} \left\lvert\sum_{{q_i} \in C_j \cap \Omega}g_i w_{q_i} \cost(a_j,S)\right\rvert \tag{By the  definition of $u^{S,0}$}.
\end{align*}
We lower bound the denominator as,
\begin{align*}
    \cost(P,S) &\gtrsim \cost(P,S) + \cost(P,A) \tag{as $\cost(P,A) \lesssim \cost(P,S)$}\\
    &\gtrsim \sum_{C_j \in B(S)} \cost(C_j,S) + \cost(C_j,A) \gtrsim \sum_{C_j \in B(S)} |C_j| \cost(a_j,S). 
\end{align*}
where the last inequality uses that for any set of centers $S$ and a cluster $C_j$,
\[
\cost(a_j, S) |C_j|\leq 2(\cost(C_j, S) + \cost(C_j,A)). 
\]
Combining these bounds we obtain, 
\begin{align*}
    |X^{S,\init}| 
    &\lesssim   \frac{\sum_{C_j \in B(S)} \left\lvert\sum_{{q_i} \in C_j \cap \Omega}g_i w_{q_i} \cost(a_j,S)\right\rvert}{\sum_{C_j \in B(S)} |C_j| \cost(a_j,S)} \\
    &\lesssim \max_{C_j \in B(S)}  \frac{\left\lvert\sum_{q_i \in C_j \cap \Omega}g_i w_{q_i} \right\rvert}{|C_j|} \tag{By averaging over clusters in $B(S)$}.
\end{align*}
This expression is independent of $S$ and is the maximum of $k_{B(S)}$ Gaussians, one for each cluster in $B(S)$. To bound $\E_g \sup_S |X^{S,\init}|$ it suffices to bound the maximum of these Gaussians. We bound their variance below.
\begin{lemma}\label{lem: cluster-variance-xinit}
    For any cluster $C_j$ we have:
   $\var_g [{\sum_{q_i \in \Omega \cap C_j}g_i w_{q_i}} ] \lesssim (k|C_j|^2)/m$.
\end{lemma}
\begin{proof}
    Using \Cref{fact: gaussian} for the sum of independent Gaussians, and \Cref{lem: weight-bound} which gives that $w_q \lesssim (k|C_j|)/m$ for $q \in C_j$, and property $P_1$ of $\mathcal{E}$ which gives that $\sum_{q \in C_j \cap \Omega} w_q \lesssim |C_j|$, we have
\begin{align*}
        \var_g [{\sum_{q_i \in \Omega \cap C_j}g_i w_{q_i}} ] = \sum_{q \in C_j \cap \Omega} w_q^2 \lesssim k|C_j|/m \sum_{q \in C_j \cap \Omega}{w_q} \lesssim  k|C_j|^2/m. \qedhere
    \end{align*}
\end{proof}
The variance bound due to \Cref{lem: cluster-variance-xinit}  together with \Cref{fact: max-of-gaussians} completes the proof of  \Cref{lem:Xinit}: 
\[\E_{g} \sup_S |X^{S,\init}| \lesssim \sqrt{k/m} \cdot \sqrt{\log k} \lesssim \sqrt{(k\log k)/m}.\]

\subsection{ \texorpdfstring{Bounding the supremum of the Gaussian process $X^{S,\fin}$}{Bounding the Supremum of X{S,Fin}}}
\label{sec:lemXfin}
In this section, we prove \Cref{lem:Xfin}. We wish to bound $\E_{g} \sup_{S\inSR} |X^{S, \fin}|$, where $X^{S, \fin}$ is defined in \eqref{eq: x-init-fin}.
Define the vector $z^S \in \R^m$ where $z_i^S = (w_{q_i} (\uq{S}{i} - \uq{S,\hmax}{i}))/\cost(P,S)$. This allows us to write $X^{S,\fin} = \langle g, z^S \rangle$. Since $\uvec{S, \hmax}$ is a ``fine'' (in fact an $\eps^2$) approximation of $\uvec{S}$ we can show that for any $S \in \mathcal S$, the vector $z^S$ has a small norm.  A bound on $\E_{\Omega, g} \sup_S |X^{S,\fin}|$ then follows from the  Cauchy-Schwarz inequality. 

\paragraph{Bounding the norm of $z^S$. }
To show $z^S$ has a small norm, we first bound the quantity below. For  $q_i \in \Omega \cap B(S)$ we have by the definition of $u_i^{S,\hmax}$ that 
\begin{align*}
\Gamma := (\uq{S}{i} - \uq{S, \hmax}{i})^2 
&\leq 2^{-2h}\err(q_i, S)\\
&
\lesssim 2^{-2\hmax}(\cost(q_i,S)\cdot(\cost(q_i,A)+  \Delta_{q_i} )  + \cost(q_i,A)^2 + \Delta_{q_i}^2).
\end{align*}

Since $q_i$ is a close point, $\cost(q_i,S) \leq 2(\cost(q_i,A) + \Delta_{q_i} \eps^{-2}) \leq 2\eps^{-2} (\cost(q_i,A) + \Delta_{q_i})$. Also, since $\hmax = \ceil{2 \log_2(\eps^{-1}))}$ we have $2^{-2\hmax} \leq \eps^{4}$. Plugging these into the right-hand side, 
expanding and keeping only higher order terms gives,
\begin{align*}
    \Gamma &\lesssim 2^{-2\hmax} \cdot \eps^{-2} \cdot  ( \Delta_{q_i}^2  +  \cost(q_i,A)^2 ) \lesssim \eps^{2} \cdot  ( \Delta_{q_i}^2  +  \cost(q_i,A)^2 ) .
\end{align*}

\noindent By this equation,  
for any $S \in \mathcal S$ we have the following bound:
    \begin{align*}
        \norm{z^S}^2 &= \sum_{i \in [m]}  \frac{w_{q_i}^2 (\uq{S}{i} - \uq{S,\hmax}{i})^2}{\cost(P,S)^2}
        \lesssim  \eps^{2}  \sum_{q_i \in B(S) \cap \Omega}\frac{ w_{q_i}^2 \cdot    ( \Delta_{q_i}^2  +  \cost({q_i},A)^2 )}{\cost(P,S)^2}\\
        &\overset{(i)}{\lesssim} \eps^{2}   \sum_{{q_i} \in B(S) \cap \Omega}\frac{ \cost(P,A)^2}{\csize^2 \cost(P,S)^2}  \overset{(ii)}{\lesssim} \eps^{2} \csize^{-1}.
    \end{align*}
    Step $(i)$ follows from the bound $w_q \leq 4\csize^{-1} \min(\tfrac{\cost(P,A)}{\cost(q,A)}, \tfrac{\cost(P, A)}{\Delta_q})$ given by \Cref{lem: weight-bound}, and step $(ii)$ from the fact that the solution $A$ is a constant-factor approximation.

    \noindent Finally, applying Cauchy-Schwarz inequality gives, 
    \begin{align*}
    \E_{g} \sup_{S\inSR}|X^{S,\fin}| =\E_{g} \sup_{S\inSR} |\langle g, z^S \rangle|  \leq \E_{g} \sup_{S\inSR} \norm{g} \, \norm{z^S} \lesssim  \eps m^{-1/2} 
 \cdot \E_{g}  \norm{g} \lesssim \eps
    \end{align*}
    where the last inequality uses that $\E_{g}  \norm{g} \lesssim m^{1/2}$ for a standard Gaussian vector $g \in \R^m$.

\section{\texorpdfstring{Sensitivity Sampling for $k$-Median Coresets}{Sensitivity Sampling for k-Median Coresets}}\label{sec: kmedian-coreset}
Consider \Cref{alg: sensitivity-sampling} with the only difference that the costs are now distances (and not squared distances as in the case of $k$-means). One can show that this algorithm yields a $k$-median coreset with the following guarantees.

\begin{theorem}[$k$-Median Coreset for Stable Inputs]\label{thm:median1}
Let $P \subset \R^d$ be a set of $n$ points which, for some $\beta > 0$, forms a $\beta$-stable instance for the $k$-median objective. Let $\Omega$ be a weighted set of $m = \tilde{\Theta}(k \eps^{-2} \max(1, \beta^{-1}))$ points obtained using Sensitivity Sampling; then $\Omega$ is an $\eps$-coreset of $P$ with probability at least $2/3$. 
\end{theorem}

\begin{theorem}[Worst case bound on $k$-Median Coreset]\label{thm:median2}
    Let $P \subset \R^d$ be a set of $n$ points. Let $\Omega$ be a weighted set of $m = \tilde{\Theta}(\min(k^{4/3} \eps^{-2}, k \eps^{-3}))$ points obtained using Sensitivity Sampling; then $\Omega$ is an $\eps$-coreset of $P$ with probability at least $2/3$.
\end{theorem}

The worst case bounds match the lower bound by \cite{HLW23} and improve over the previously best known upper bound of $\tilde{O}(k\cdot \varepsilon^{-4})$ by \cite{huang2020coresets} for Sensitivity Sampling.

The analysis of Sensitivity Sampling for $k$-median is almost identical to the one described in complete detail for $k$-means. The high-level reason why we get better (worst-case) coreset bounds for $k$-median is that if a candidate center $s$ serves points of a cluster at distances $[x,x+1]$, there are $\frac{x+1}{x\cdot \varepsilon}$ different distance-based costs we need to consider. In contrast, squared distances induce up to $\left(\frac{x+1}{x\cdot \varepsilon}\right)^2$ different costs.

The arguments are analogous compared to those used for $k$-means, save that whenever we apply the inequality
$$|d^2(a,c)-d^2(b,c)|\leq \varepsilon\cdot d^2(a,c) + \left(1+\frac{1}{\varepsilon}\right)\cdot d^2(a,b),$$
we now apply a standard triangle inequality.
Both the triangle inequality and the approximate triangle inequality are tight whenever $a,b$, and $c$ are collinear. It turns out that following the $k$-means analysis while using the triangle inequality (in place of the approximate triangle inequality) leads to the claimed coreset bounds for $k$-median.

 Below, we sketch the changes that need to be made to the analysis presented in the paper's main section to obtain the coreset bounds claimed for $k$-median.

\noindent \paragraph{Estimator and Reduction to a Gaussian Process:} As before, we wish to show the concentration of the cost estimator 
 $$\frac{1}{|\Omega|}\sum_{q\in \Omega}w_q\cdot \cost(q,S).$$
Following normalization, a grouping of clusters, centers, and symmetrization, we wish to bound the quantity
$$\E_{\Omega} \E_g \sup_{S}X^S(\Omega, g) =  \E_{\Omega} \E_g \sup_{S}\sum_{i \in [m]} \frac{g_i w_{q_i} u_i^S(\Omega)}{\cost(P,S)}.$$
As before $u^S(\Omega)$ is the cost vector of coreset points from the group $B(S)$ (defined analogously). We show the following lemma.
\begin{lemma}
    For any fixed $\Omega$ satisfying the bounds \cref{thm:median1} for $\beta$-stable instances and \cref{thm:median2} for worst-case instances, we have 
    $ \E_{g} \sup_{S} \left\lvert X^S(\Omega,g)\right\rvert \lesssim \eps $.
\end{lemma}
Proving this implies that $\Omega$ is a coreset.

The overall proof strategy is still the same. We will give additional details for these steps in the following. 
\begin{itemize}
    \item We analyze the contribution of points that are far and of points that are close separately. 
    \item For the close points, we characterize the contribution by analogous notions of bands, types, and center classes.
    \item We give refined variance bounds and net sizes for each regime.
\end{itemize}
With these notions in place, we complete the proof by integrating it into a chaining analysis.

\noindent \textbf{Far and Close Points.} 
Given a set of centers $S$, we say that cluster $C_j$ with center $a_j$ is far from $S$ if  $ \cost(a_j, S) \geq \Delta_j \eps^{-1}$; otherwise we say the cluster is close. All points in far clusters(resp. far) are called far (resp. close) points.  One can show by an analog of \Cref{lem: goal-far} that the coreset preserves the cost of far points. This holds, regardless of stability assumptions.

For the more involved case of bounding the contribution of close points, we follow the analysis in \Cref{sec: analysis} described as follows.

\noindent \textbf{High and Low Cost Clusters.} One can handle close clusters whose cost in $A$ is $T := O(\eps^2 k^{-1}\cdot \cost(P,A))$ easily as their total cost is at most $O(\eps \cost(P,A))$ and so they do not introduce significant error in the coreset's estimate of the cost.

The analysis of clusters with cost larger than $T$ (i.e., $\Omega(\eps^2 k^{-1} \cost(P,A))$) again requires more care. In particular, we again classify the clusters into bands and types depending on their costs to $A$ and to $S$. 

\noindent \textbf{Bands, Types, and Center Classes.} 
One now defines bands, types, and center classes as in \Cref{sec:further-classifying}. To recall, we say that cluster $C_j$ is in Band-$b$ if $\cost(C_j,A) \approx 2^b T$ and $\cost(a_j,S) \approx 2^t \Delta_j$ (which if $t$  is large essentially means $\cost(C_j,S) = \Omega(2^t \cost(C_j,A))$). Also, define $B_{b,t}(S)$ to be the set of all clusters that are from band $b$ and of type $t$.

The classification of centers into center classes is identical to that described in \Cref{sec:further-classifying}. We define center cluster interaction, signature and interaction numbers in exactly the same way. We then define the center class $\mathcal{S}(r)$ to be the set of centers with interaction number roughly $2^r$.

\noindent\textbf{Chaining and Cost Vector Nets.}

There are two components. First, we require bounds on the variance of the estimator used for the $k$-median objective. Second, we require bounds on the net sizes.
Adapting the variances bounds is straightforward, after adjusting for using distances, rather than squared Euclidean distances.

The nets for the cost vectors $u^S$ have a slightly different error guarantee defined in the following. However, their construction is essentially identical to that of the cost vector nets used in the analysis of $k$-means coresets.

\begin{definition}[Cost Vector Nets]  Consider a fixed set $\Omega$. For a real $\alpha \in (0,1/2]$, an $\alpha$-net  $M_\alpha(\Omega) \subset \mathbb R^m$ is a finite set of vectors with the following properties. For any set of centers $S \in \mathcal{S}(r)$ there exists some $v \in M_\alpha(\Omega)$ which $\alpha$-approximates the vector $u^S(\Omega)$ as follows. For each $i \in [m]$:
\begin{enumerate}
    \item If $q_i \in  B(S)$ then $|v_i - u^S_i(\Omega)|  \leq \alpha \cdot \err(q_i)$ where  for $p \in P$  we define
        $\err(p) :=  (\cost(p,A) + \Delta_{p})$.
    \item If $q_i \notin B(S)$ then $v_i = 0$ . Note that we also have $u^S_i(\Omega) = 0$ in this case by definition.
\end{enumerate}
\end{definition}

To highlight the main difference, note that the coordinates of  $u^{S}$ represent distances and that the net vector $v$ approximates the distance of $p$ to $S$ with error $\alpha (\dist(p,A) + \Delta_p)$. We claim that we can use the nets constructed for the $k$-means problem in almost a black-box fashion to obtain such nets for $k$-median. In particular, we get a net with the size given in the below lemma. We remark that the net size is essentially identical with the only difference that the second net size is larger by a factor of $2^t$ (than \Cref{lem: net-lemma}) because types are defined with respect to distances rather than squared distances.  \begin{lemma}\label{lem: net-lemma-k-median}
 For any $\alpha \in (0,1/2)$, there is an $(\alpha, \mathcal{S}(r))$-net $M_\alpha$ with cardinality \[|M_\alpha| = \exp(O(\min(2^r+ k\alpha^{-2}, 2^{2t} k \alpha^{-2})\cdot \log(k \alpha^{-1}{\eps^{-1}}))).\]
\end{lemma}
\begin{proof}
The proof is identical to that of \cref{lem: net-lemma-signature}. The only moral difference is that, having approximated squared distances between point $p$ to a candidate center $s$ with an approximating candidate center $s'$ up to an error 
\begin{equation}
\label{eq:kmediannet}
|d^2(p,s)-d^2(p,s')|\leq \varepsilon\cdot \sqrt{d^2(p,S)\cdot d^2(p,A)}    
\end{equation}
we then obtain for the distances
\[|d(p,s)-d(p,s')|\leq O(\varepsilon)\cdot d(p,s).\]
Up to obtaining the bound for \cref{eq:kmediannet}, the proof is identical.
\end{proof}

The final variance bounds are almost identical to their counterparts for $k$-means. The only difference is that, due to the finer net guarantee, the variance for worst-case inputs decreases inversely proportionate to the squared cost increase, rather than inversely proportionate to the cost increase, see \cref{lem: var-worst-case-kmedian}
below. This is due to the error bounds given by the $k$-median nets in \cref{eq:kmediannet} being smaller than that of $k$-means. 
Indeed, this difference is why the worst-case bounds for $k$-median are slightly better than those of $k$-means. The proofs are otherwise completely identical.

We have the following analog of \Cref{lem: easy-var}.
\begin{lemma}[Variance Bound in Terms of the Type] \label{lem: var-easy-kmedian}Let 
$\var[X^{S,h}] \lesssim 2^{-2h}/(m2^t).$
\end{lemma}

\begin{proof}[Proof sketch]
The proof of this bound is similar to \Cref{lem: easy-var} with the only difference that we need to account for the finer nets. We have,
\begin{align}\label{eq: var-1-kmedian}
    \var[X^{S,h}]&\lesssim (1/\cost(P,S))^2 2^{-2h}  \sum_{q \in B(S) \cap \Omega } w_q^2 \cost(q,A)^2.
\end{align}
Using the fact that $w_q \cost(q,A) \lesssim \cost(P,A)/m \lesssim \cost(P,S)/m$ (due to the bounds on the weight) together with $2^t \sum_{q \in B(S) \cap \Omega }  w_{q} \cost(q,A) \lesssim  \cost(P,S)$ in \Cref{eq: var-1-kmedian} leads to the claimed bound.\qedhere
\end{proof}

If the $k$-median instance $P$ is $\beta$-stable, then for any set of centers $S$, one can give a lower bound the cost of $P$ to $S$ in terms of the interaction number of $S$ and the stability parameter $\beta$ (analogous to \Cref{lem: cost-increase}). 
\begin{lemma}\label{lem: cost-increase-kmedian}
	Let $P$ be a $\beta$-stable instance for the $k$-median objective and $S$ be any set of centers; then we have  $$\cost(P,S) \gtrsim  \opt_k \cdot \max(1,  (\tfrac{\ccint(S)}{k} - 1) \beta).$$ 
\end{lemma}

The above cost lower bound then leads to the following bound on the variance of $X^{S,h}$.
\begin{lemma}[Variance Bounds for Stable Instances]\label{lem: var-stable-kmedian} Let $P$ be a $\beta$ stable instance for the $k$-median objective and $S$ be a set of centers with $N(S) \geq 2k$ then $\var[X^{S,h}] \lesssim (2^{-2h} k)/(m N(S) \beta)$.
\end{lemma}

One can use the structure of $B(S)$ to derive the following variance bounds. Again, the finer nets yield the a bound on the variance of $X^{S,h}$ that is tighter than \Cref{cor: worst-case-var} by a factor of $2^t$.
\begin{lemma}[Variance Bounds for worst Case Instances]\label{lem: var-worst-case-kmedian}
Let $S$ be any set of centers; then  $\var[X^{S,h}] \lesssim (2^{-2h}k^2)/(m \cdot N(S) \cdot 2^{2t})$.
\end{lemma}

\subsection{Completing the Proof}
Let $\Gamma := \E_g \sup_S |X^{S,h}|$; we bound $\Gamma$ using \Cref{eq: combine-net-variance},  the net and variance bounds given above.

\paragraph{Stable Inputs.} Note that the variance and net bounds for stable $k$-median instances (\Cref{lem: var-stable-kmedian} and \Cref{lem: net-lemma-k-median}) are identical to those for $k$-means (\Cref{lem: net-lemma} and \Cref{lem: stable-var}). Therefore performing similar calculations as in \Cref{sec: completing-the-proof-kmeans}, we get that picking $m = \Omega(k \eps^{-2} \max(1, \beta^{-1}))$ ensures that $\Gamma \lesssim \eps$.

\paragraph{Worst-Case Inputs.} As in the analysis for $k$-means, the most interesting case is when the interaction number is $\Omega(k)$. Suppose that this is the case, i.e., $2^r \gtrsim k$.  By \Cref{lem: var-worst-case-kmedian} and \Cref{lem: var-easy-kmedian} we have that for any $S \in \mathcal{S}(r)$,
\begin{align}\label{eq: combine-variance}
    \var[X^{S,h}] \lesssim  \min(2^{-t}, k^2 2^{-(2t+r)}) \cdot 2^{-2h} m^{-1}.
\end{align}
Using \Cref{lem: net-lemma-k-median} and the fact that $2^r + k2^{2h} \lesssim 2^{r+2h}$, we also have
\begin{align}\label{eq: combine-net}
    \log |M_{2^{-h}}| \lesssim  \min(k2^{2t}, 2^{r}) \cdot 2^{2h} \log(k\eps^{-1}).
\end{align}
Combining the corresponding net and variance bounds in \Cref{eq: combine-net}, \Cref{eq: combine-variance}, we get,
\begin{align}
    \Gamma  &\lesssim \sqrt{\min(k2^{t}, k^2 2^{-2t}) \cdot \log(k \eps^{-1})m^{-1}} \label{eq: interm-combine-kmedian}\\
    &\lesssim \sqrt{k^{4/3} \cdot \log(k \eps^{-1})m^{-1}}.
\end{align}
Therefore  if $m =  \tilde{\Omega}(k^{4/3} \eps^{-2})$ then we have that $\Gamma \lesssim \eps$. Moreover, since $2^t \leq 2^{\tmax}\lesssim \eps^{1}$, using only the first bound from \Cref{eq: interm-combine-kmedian}, we get that  $\Gamma \lesssim \sqrt{k\eps^{-1}\log(k \eps^{-1}) m^{-1}}$; therefore if  $m = \tilde{\Omega}(k\eps^{-3})$ then $\Gamma \lesssim  \eps$. This completes the proof of the worst-case bound.

\section{Analysis for Doubling Metrics}\label{app: doubling}

We also briefly show how to analyze Sensitivity Sampling for doubling metrics. 
Given the improved variance bounds of Sensitivity Sampling, the only thing left to do is bound the size of cost vector nets. Doing so is a straightforward adaptation of the earlier analysis in this paper and of \cite{CSS21}. Replacing the bounds on the size of clustering nets in \Cref{sec: analysis} then yields \Cref{thm:doubling}.
\begin{definition}
The \emph{doubling dimension} of a metric $(X,\dist)$ is the smallest integer $D$ such that any
ball of radius $2r$ can be covered by $2^D$ balls of radius~$r$. 

A $\gamma$-\emph{net} of $V\subset X$ is a set of points $U\subseteq V$ such that for all
$v \in V$ there is an $u \in U$ such that $\dist(v, u) \leq \gamma$.
\end{definition}
\begin{lemma}[from \cite{GuptaKL03}]\label{prop:doub:net}
Let $(V, \dist)$ be a metric space with doubling dimension $D$ and diameter 
$\Delta$, and let $X$ be a $\gamma$-net of $V$. Then $|X| \leq 2^{D \cdot 
\lceil \log_2 (\Delta/\gamma)\rceil}$.
\end{lemma}

\end{document}